\documentclass{lmcs} 
\pdfoutput=1

\usepackage{lastpage}
\lmcsdoi{15}{1}{24}
\lmcsheading{}{\pageref{LastPage}}{}{}%
{Mar.~30,~2018}{Mar.~05,~2019}{}


\usepackage[utf8]{inputenc}
\usepackage[english]{babel}
\usepackage{microtype}
\usepackage{amsmath,amssymb}
\usepackage{bm}
\usepackage{tabularx}
\usepackage{color}
\usepackage{subcaption}

\usepackage{tikz}
\usetikzlibrary{arrows}
\usetikzlibrary{shapes}
\tikzstyle{every picture}=[>=stealth']
\tikzstyle{vertex}=[draw,circle,minimum size=4ex,inner sep=1pt]
\tikzstyle{vertex2}=[draw,rounded corners,inner sep=1pt,minimum size=3ex]
\tikzstyle{vertex3}=[draw,circle,minimum size=3ex,inner sep=1pt]
\tikzstyle{vertex4}=[draw,circle,line width=0.2mm,inner sep=0mm,minimum size=4mm]
\tikzstyle{dotvertex}=[fill,circle,inner sep=3pt]

\newcommand{\alias}[2]{%
  \expandafter\let\csname #1\expandafter\endcsname\csname #2\endcsname
  \expandafter\let\csname end#1\expandafter\endcsname\csname end#2\endcsname
}
\alias{theorem}{thm}
\alias{lemma}{lem}
\alias{proposition}{prop}
\alias{corollary}{cor}
\alias{observation}{obs}
\alias{definition}{defi}
\alias{example}{exa}
\alias{remark}{rem}
\alias{claim}{clm}

\theoremstyle{defC}
\newtheorem{observationC}[thm]{Observation} 

\newenvironment{proofofclaim}
  {\renewcommand{\qed}{\hfill$\lrcorner$}\proof}
{\endproof}

\newcommand{\CA}{\mathcal{A}}
\newcommand{\CC}{\mathcal{C}}
\newcommand{\CD}{\mathcal{D}}

\newcommand{\CG}{\mathcal{G}}
\newcommand{\CI}{\mathcal{I}}
\newcommand{\CK}{\mathcal{K}}
\newcommand{\CL}{\mathcal{L}}

\newcommand{\CP}{\mathcal{P}}

\newcommand{\N}{\mathbb{N}}

\newcommand{\tup}[1]{\bar{#1}}
\newcommand{\num}[2][]{\left\langle#2\right\rangle_{\! #1}}

\newcommand{\KKI}{\ensuremath{\CG_{\CK\CI}^*}}

\newcommand{\isdef}{\mathrel{\mathop:}=}

\newcommand{\set}[1]{\{#1\}}
\newcommand{\card}[1]{\lvert{#1}\rvert}
\newcommand{\dcup}{\mathop{ \dot{\cup}}}
\newcommand{\edgewedge}{\ensuremath{\scalebox{0.7}[1.25]{$\,\wedge\,$}}}

\newcommand{\nb}{\ensuremath{\#}}

\newcommand{\modout}{\ensuremath{/\!}}

\newcommand{\Dom}{\textup{Dom}}
\newcommand{\dom}{\textup{dom}}
\newcommand{\suit}{\textup{Suit}}
\newcommand{\numb}{\textup{nb}}
\newcommand{\raute}{{\textup{\begin{tiny}\#\end{tiny}}}}

\newcommand{\ar}{\operatorname{ar}}
\newcommand{\free}{\operatorname{free}}
\newcommand{\Domain}[2]{#1^{#2}}

\newcommand{\rep}{g_{\text{rep}}}
\newcommand{\graph}{g_{\text{graph}}}

\newcommand{\logic}[1]{\textup{\small\textsf{#1}}}
\newcommand{\Logic}{\logic{L}}

\newcommand{\FO}{\logic{FO}}
\newcommand{\FOC}{\logic{FO{+}C}}

\newcommand{\DTCC}{\logic{DTC{+}C}}
\newcommand{\STC}{\logic{STC}}
\newcommand{\STCC}{\logic{STC{+}C}}
\newcommand{\plusC}{\text{(}\logic{{+}C}\text{)}}

\newcommand{\FP}{\logic{FP}}
\newcommand{\FPC}{\logic{FP{+}C}}

\newcommand{\logicf}[1]{\textup{\scriptsize\textsf{#1}}}

\newcommand{\IFPf}{\logicf{IFP}}
\newcommand{\IFPCf}{\logicf{IFP{$+$}C}}
\newcommand{\LFPf}{\logicf{LFP}}
\newcommand{\FPf}{\logicf{FP}}
\newcommand{\FPCf}{\logicf{FP{$+$}C}}

\newcommand{\stc}{\operatorname{stc}}
\newcommand{\ifp}{\operatorname{ifp}}

\newcommand{\stcx}[3]{[\stc_{\,{#1},{#2}}{#3}]}
\newcommand{\ifpx}[3]{[\ifp_{\,{#1},{#2}}{#3}]}

\newcommand{\zero}{\operatorname{zero}}
\newcommand{\one}{\operatorname{one}}
\newcommand{\ord}{\operatorname{largest}}
\newcommand{\plus}{\operatorname{plus}}

\newcommand{\limplies}{\rightarrow}
\newcommand{\true}{\top}
\newcommand{\false}{\bot}

\newcommand{\cclassname}[1]{\textup{\ensuremath{\mathsf{#1}}}}
\newcommand{\PTIME}{\cclassname{PTIME}}
\newcommand{\NP}{\cclassname{NP}}
\newcommand{\LOGSPACE}{\cclassname{LOGSPACE}}

\newcommand{\LO}{\text{LO}}

\makeatletter
\newcommand{\superimpose}[2]{%
  {\ooalign{$#1\@firstoftwo#2$\cr\hfil$#1\@secondoftwo#2$\hfil\cr}}}
\makeatother

\newcommand{\mycdot}{%
  \mathchoice{\raisebox{1pt}{$\displaystyle\cdot$}}
             {\raisebox{1pt}{$\cdot$}}
             {\raisebox{0.6pt}{$\scriptstyle\cdot$}}
             {\raisebox{0.2pt}{$\scriptscriptstyle\cdot$}}}

\newcommand{\lin}{\mathpalette\superimpose{{\trianglelefteq}{\mkern4mu\mycdot}}}

\graphicspath{{./graphics/}{./graphics/moddecomp/}{./graphics/perm/}{./graphics/chclaw/}}

\begin{document}

\title[Capturing Polynomial Time using Modular Decomposition]{Capturing Polynomial Time using Modular Decomposition}
\titlecomment{{\lsuper*}This article is an extended version of \cite{Grussien17}}

\author[B.~Gru{\ss}ien]{Berit Gru{\ss}ien}	
\address{Humboldt-Universität zu Berlin, Unter den Linden 6, 10099 Berlin, Germany}	
\email{grussien@informatik.hu-berlin.de}

\begin{abstract}
\noindent 
The question of whether there is a logic that captures polynomial time 
is one of the main open problems in descriptive complexity theory and database theory. 
In 2010 Grohe showed that fixed-point logic with counting captures polynomial time on all classes of graphs with excluded minors.
We now consider classes of graphs with excluded induced subgraphs. 
For such graph classes, an effective graph decomposition, called modular decomposition, was introduced by Gallai in 1976.
The graphs that are non-decomposable with respect to modular decomposition are called \emph{prime}.
We present a tool, the Modular Decomposition Theorem, that 
reduces (definable) canonization of a graph class~$\CC$ to
(definable) canonization of the class of prime graphs of~$\CC$ 
that are colored with binary relations on a linearly ordered set.
By an application of the Modular Decomposition Theorem,
we show that fixed-point logic with counting captures polynomial time on the class of permutation graphs.
Within the proof of the Modular Decomposition Theorem, 
we show that the modular decomposition of a graph is definable in symmetric transitive closure logic with counting.
We obtain that the modular decomposition tree is computable in logarithmic space.
It follows that cograph recognition and cograph canonization is computable in logarithmic space.
\end{abstract}

\maketitle

\section{Introduction}\label{sec:intro}
    
The aim of descriptive complexity theory is to find logics that characterize, or \emph{capture}, complexity classes.
The first result in this field was made by Fagin in 1974~\cite{fag74}.
He showed that existential sec\-ond-or\-der logic captures the complexity class $\NP$.
One of the most interesting open problems in descriptive complexity theory is the question of whether there exists a logic that captures $\PTIME$.%
\footnote{ Originally, this question was asked by Chandra and Harel in a study of database query languages~\cite{chahar82}.}

Independently of each other, Immerman~\cite{imm86} and Vardi~\cite{var82} 
obtained an early result towards a logical characterization for $\PTIME$. 
They proved that fixed-point logic ($\FP$) captures $\PTIME$ on ordered structures,\footnote{
More precisely, Immerman and Vardi's theorem holds for least fixed-point logic ($\LFPf$) and the equally expressive
  inflationary fixed-point logic ($\IFPf$). Our indeterminate \FPf\ refers to either of these two logics.
} that is, 
on structures where a linear order is present. 
On structures that are not necessarily ordered, it is easy to prove that $\FP$ does not capture $\PTIME$. 
In order to obtain a candidate for a logic capturing $\PTIME$ on all structures, Immerman proposed in 1987 to add to $\FP$ the ability to count \cite{imm87a}.
Although the resulting logic, fixed-point logic with counting ($\FPC$), is not strong enough to capture $\PTIME$ on all finite structures~\cite{caifurimm92},
it does so on many interesting classes of structures:
$\FPC$ captures $\PTIME$, for example, on 
planar graphs~\cite{Grohe98planar}, 
all classes of graphs of bounded treewidth~\cite{GroheM99boundedtreewidth}
and on $K_5$-mi\-nor free graphs~\cite{Grohe08a}.
Note that all these classes can be defined by a list of forbidden minors.
In fact, Grohe showed in 2010 that $\FPC$ captures $\PTIME$ on all graph classes with excluded minors~\cite{gro10}.
This leads to the question whether a similar result can be obtained for graph classes that are
characterized by excluded induced subgraphs, i.e., graph classes that are closed under taking induced subgraphs.
For $\FPC$ such a general result is not possible:
Capturing $\PTIME$ on the class of chordal graphs, comparability graphs or co-com\-pa\-ra\-bil\-i\-ty graphs
is as hard as capturing $\PTIME$ on the class of all graphs for any ``reasonable'' logic~\cite{Grohe10linegraphs,laubner11diss}.
Yet, this gives us reason to consider the three mentioned graph classes and their subclasses more closely.
So far, there are results showing that $\FPC$ captures $\PTIME$ on the class of chordal line graphs~\cite{Grohe10linegraphs}
and on the class of interval graphs (chordal co-com\-pa\-ra\-bil\-i\-ty graphs)~\cite{Laubner10}.

We add to these results the following results:
$\FPC$ captures $\PTIME$ on the class of permutation graphs (comparability co-com\-pa\-ra\-bil\-i\-ty graphs) (see Section~\ref{sec:permutationgraphs})
and on the class of chordal comparability graphs~(see \cite{diss}).
Both results are based on modular decomposition (also called substitution decomposition), 
a graph decomposition which was introduced by Gallai in 1976~\cite{gallai67transitiv}. 
Similar to treelike decomposition for classes with forbidden minors, 
modular decomposition is a suitable efficient graph decomposition for classes with forbidden induced subgraphs.

The modular decomposition of a graph partitions the vertex set
of the graph into so called modules, that is, into subsets that share the same neighbors. 
A graph is \textit{{prime}} if only the vertex set itself and all vertex sets of size $1$ are modules of the graph.
For every class $\CC$ of graphs that is closed under induced subgraphs,
we let $\CC^*_{\textup{prim}}$ be the class of all prime graphs from $\CC$ that are colored 
with binary relations on a linearly ordered set.
Our Modular Decomposition Theorem states that 
there is an $\FPC$-can\-on\-iza\-tion of $\CC$
if there is an $\FPC$-can\-on\-iza\-tion of the class $\CC^*_{\textup{prim}}$.
Note that the Modular Decomposition Theorem also holds for reasonable extensions of $\FPC$ that are closed under $\FPC$-trans\-duc\-tions.

The Modular Decomposition Theorem can be used for multiple purposes.
One reason for this is that the existence of an $\FPC$-can\-on\-iza\-tion of a graph class $\CC$ has various consequences.
It implies that $\FPC$ captures $\PTIME$ on class~$\CC$, the existence of a poly\-no\-mi\-al-time canonization algorithm for graph class $\CC$,
and that there is an easy-to-understand algorithm, the Weisfeiler-Leman Method, that solves the graph isomorphism problem on~$\CC$.
Further, the Modular Decomposition Theorem itself can be transferred to polynomial time:
There is a poly\-no\-mi\-al-time canonization algorithm for $\CC$ if there is a poly\-no\-mi\-al-time 
canonization algorithm for the class $\CC^*_{\textup{prim}}$.
We will present even more variations of the Modular Decomposition Theorem that might be helpful for future applications.

By means of the Modular Decomposition Theorem, we can not only show that the canonization of the class of permutation graphs and 
chordal comparability graphs is definable in $\FPC$, 
but simplify the proof of Laubner in \cite{Laubner10} that there is an $\FPC$-can\-on\-iza\-tion of the class of interval graphs.

Within the proof of the Modular Decomposition Theorem, we show that the modular decomposition of a graph is definable in symmetric transitive closure logic with counting.
As a consequence, the modular decomposition tree can be computed in logarithmic space.
Previously, it was only known that the modular decomposition tree is computable in linear time \cite{Cournier94,McConnell94}, or 
in polylogarithmic time with a linear number of processors \cite{Dahlhaus95}.%
\footnote{ For a survey of the algorithmic aspects of modular decomposition see \cite{HabibP10}.}
It follows directly that cograph recognition is in $\LOGSPACE$.
As there is a log\-a\-rith\-mic-space algorithm for tree canonization \cite{Lindell:Tree-Canon}, 
it also follows that there exists a log\-a\-rith\-mic-space algorithm for cograph canonization.

\subsection*{Structure}

After setting out the necessary preliminaries in Section~\ref{chp:prelims}, 
we introduce modular decomposition and show that it is $\STCC$-de\-fin\-able, and therefore $\LOGSPACE$-com\-putable, in Section~\ref{sec:ModDecompSTCC}.
In Section~\ref{sec:ModDecompThm}, we introduce the Modular Decomposition Theorem,  we prove it, 
and we present variations of it.
Finally, we apply (a variation of) the Modular Decomposition Theorem in Section~\ref{sec:permutationgraphs} 
and show that $\FPC$ captures $\PTIME$ on the class of permutation graphs.
We close with a few concluding remarks.

\section{Basic Definitions and Notation}\label{chp:prelims}

We write $\N$ for the set of all non-negative integers.
For all ${n,n' \in \N}$, we define $[n,n']:=\{m\in \N\mid n\leq m\leq n'\}$ and $[n]:=[1,n]$.
We often denote tuples $(a_1,\dots,a_k)$ by $\bar{a}$.
Given a tuple $\tup{a} = (a_1,\dotsc,a_k)$, let $\tilde{a} \isdef \set{a_1,\dotsc,a_k}$.
Let $n\geq 1$, and $\tup{a}^i = (a_1^i,\dotsc,a_{k_i}^i)$ be a tuple of length $k_i$ for each $i\in[n]$.
We denote the tuple $(a_1^1,\dotsc,a_{k_1}^1,\,\dots\, ,a_1^l,\dotsc,a_{k_l}^l)$ by $(\tup{a}^1,\dots,\tup{a}^l)$.
Mappings $f\colon A \to B$ are extended to tuples $\tup{a} = (a_1,\dotsc,a_k)$ over $A$
via $f(\tup{a}) \isdef (f(a_1),\dotsc,f(a_k))$.

For a set $S$, we let $\CP(S)$
be the set of all subsets of $S$ and ${S\choose 2}$ be the set of all 2-element subsets of $S$.
If $\CD$ is a set of sets, then we let $\bigcup \CD$ be the union of all sets in $\CD$.
The \emph{disjoint union} of two sets $S$ and $S'$ is denoted by $S\dcup S'$.
A \emph{partition} of a set $S$ is a set $\CD$ of disjoint non-empty subsets of $S$ such that $S =\bigcup \CD$.

\subsection{Relations and Orders}

The reflexive, symmetric, transitive closure of a binary relation $R$ on $U$ is called the equivalence relation \emph{generated} by $R$ on $U$.  
Let $\approx$ be an \emph{equivalence relation} on $U\!$.
For each $a\in U$, we denote the \emph{equivalence class} of $a$ by $a\modout_\approx$.
(We also use another notation, which we specify later.) 
We let $U\modout_{\approx}$ be the set of equivalence classes.
For a $k$-ary relation $R\subseteq U^k$ we let $R\modout_{\approx}$ be the set
$\{(a_1\modout_{\approx},\dots,a_k\modout_{\approx})\mid (a_1,\dots,a_k)\in R\}$.

A binary relation $\prec$ on a set $U$ is a \emph{strict partial order} if it is irreflexive and transitive.
We say  $a$ and $b$ are \emph{comparable} with respect to a strict partial order $\prec$ if 
$a\prec b$ or $b\prec a$; otherwise we call them \emph{incomparable}.
A strict partial order where no two elements $a,b$  with $a\not= b$ are incomparable is called a \emph{strict linear order}.
For each strict linear order $\prec$ there exists an associated reflexive relation 
$\preceq$, called a \emph{linear order}, which is defined by
$a\preceq b$ if and only if $a\prec b$ or~${a=b}$.
A~binary~relation $\preceq$ is a linear order if and only if it is  transitive, antisymmetric and total.

A \emph{strict weak order} is a strict partial order where incomparability  is transitive.
Moreover, in a strict weak order incomparability is an equivalence relation.
If $a$ and $a'$ are incomparable with respect to a strict weak order~$\prec$, then $a\prec b$ implies  $a'\!\prec b$, and $b\prec a$ implies $b\prec a'\!$.
As a consequence, if $\prec$ is a strict weak order on $U$ and $\sim$ is the equivalence relation defined by incomparability, 
then $\prec$ induces a strict linear order on the set $U\modout_\sim$ of equivalence classes.

\subsection{Graphs and LO-Colored Graphs}

A \emph{graph} is a pair $(V,E)$ consisting of a non-emp\-ty finite set $V$ of \emph{vertices}
and a set $E\subseteq {V\choose 2}$ of \emph{edges}. 

Let $G=(V,E)$ be a graph.
For a subset $W\subseteq V$ of vertices, $G[W]$ denotes the \emph{induced subgraph} 
of $G$ with vertex set $W\!$.
The \emph{complement graph} of $G$ is the graph $\overline{G}:=(V,\overline{E})$ where $\overline{E}={V\choose 2}\setminus E$.
\emph{Connectivity} and \emph{connected components} are defined the usual way.

Let $G=(V,E)$ be a graph and $f\colon V\to C$ be a mapping from the vertices of $G$ to a finite set $C$. 
Then $f$ is a \emph{coloring} of $G$, and the elements of $C$ are called \emph{colors}.
Throughout this paper we color the vertices of a graph with binary relations on a linearly ordered set.\footnote{
In particular, we  color graphs with representations of ordered copies of graphs on the number sort (defined in Section~\ref{sec:representation}).}
We call graphs with such a coloring \emph{\LO-col\-ored graphs}.
More precisely, an \LO-col\-ored graph is a tuple $G^*\!=(V,E,M,\trianglelefteq,L)$ with the following properties:
\begin{enumerate}
 \item The pair $(V,E)$ is a graph. We call $(V,E)$ the \emph{underlying graph} of~$G^*\!$.
 \item The set of \emph{basic color elements} $M$ is a non-emp\-ty finite set with $M\cap V=\emptyset$.
 \item The binary relation $\trianglelefteq\ \subseteq M^2$ is a linear order on $M$.
 \item The \emph{color relation} $L\subseteq V \times M^2$ is a ternary relation that assigns to each vertex $v \in V$ a color 
 ${L_v := \{(d,d')\mid (v,d,d')\in L\}}$.
\end{enumerate}

\noindent
Let $d_0,\dots,d_{|M|-1}$ be the enumeration of the basic color elements in $M$ according to their linear order 
$\trianglelefteq$.
We call ${L^{\N}_v:=\{(i,j)\hspace*{-1pt}\in\hspace*{-1pt}\N^2\hspace*{-1pt}\mid\hspace*{-1pt} (d_i,d_j)\hspace*{-1pt}\in\hspace*{-1pt} L_v\}}$ 
the \emph{natural color} of $v\in V\!$.

We can use the linear order $\trianglelefteq$ on $M$ to obtain a linear order 
on the colors $\{L_v\mid v \in V\}$ of $G^*\!$. 
Thus, an \LO-col\-ored graph is a special kind of colored graph with a linear order on its colors.

\subsection{Structures} 

A \emph{vocabulary} is a finite set $\tau$ of 
relation symbols. Each relation symbol $R\in\tau$ has a fixed arity $\ar(R)\in \N$.
A \emph{$\tau$-struc\-ture}
consists of a non-emp\-ty finite set $U(A)$, its \emph{universe}, 
and for each relation symbol $R\in \tau$ of a relation ${R(A)\subseteq U(A)^{\ar(R)}}$.

An \emph{isomorphism} between $\tau$-struc\-tures $A$ and $B$ is a bijection 
$f\colon U(A)\to U(B)$ such that for all $R\in \tau$ and all 
$\tup{a}\in U(A)^{\ar(R)}$ we have $\tup{a}\in R(A)$ if and only if $f(\tup{a})\in R(B)$.
We write $A\cong B$ to indicate that $A$ and $B$ are \emph{isomorphic}.

Let $E$ be a binary relation symbol.
Each graph corresponds to an $\{E\}$-struc\-ture $G=(V,E)$ where the universe $V$ is the vertex set and $E$ is an 
irreflexive and symmetric binary relation, the edge relation.
To represent an \LO-col\-ored graph ${G^*\!=(V,E,M,\trianglelefteq,L)}$ as a logical structure we extend the 
$5$-tu\-ple by a set $U$ to a $6$-tu\-ple $(U,V,E,M,\trianglelefteq,L)$, and we require that $U=V\dcup M$  additionally to properties 1-4. 
The set $U$ serves as the universe of the structure, and $V,E,M,\trianglelefteq,L$ are relations on $U$.
We usually do not distinguish between (\LO-col\-ored) graphs and their representation as logical structures.
It will be clear from the context which form we are referring~to.

\subsection{Logics}\label{sec:logiken}

In this section we introduce first-order logic with counting, symmetric transitive closure logic (with counting) and fixed-point logic (with counting).
Detailed introductions of these logics can be found, e.g., in \cite{ebbflu99,groheDC,imm87}.
We assume basic knowledge in logic, in particular of \emph{first-or\-der logic \textup{($\FO$)}}. 

\bigskip

\noindent
\emph{First-or\-der logic with counting \textup{($\FOC$)}}
extends $\FO$ by a counting operator
that allows for counting the cardinality of $\FOC$-de\-fin\-able relations.
It lives in a two-sorted context,
where structures $A$ are equipped with a \emph{number sort}
$N(A) \isdef [0,\card{U(A)}]$.
$\FOC$ has two types of \emph{individual variables}: 
$\FOC$-vari\-ables are either \emph{structure variables}
that range over the universe $U(A)$ of a structure~$A$,
or \emph{number variables} that range over the number sort $N(A)$.
For each individual variable $u$,
let $A^{u} \isdef U(A)$ if $u$ is a structure variable,
and $A^{u} \isdef N(A)$ if $u$ is a number variable.
Let ${A}^{(u_1,\dotsc,u_k)} := {A}^{u_1} \times \dotsb \times {A}^{u_k}$.
Tuples $(u_1,\dotsc,u_k)$ and $(v_1,\dotsc,v_\ell)$ of variables
are \emph{compatible} if $k = \ell$,
and for every $i \in [k]$ the variables $u_i$ and $v_i$ have the same type.
An \emph{assignment in $A$} is a mapping $\alpha$
from the set of variables to $U(A) \cup N(A)$,
where for each variable $u$ we have $\alpha(u) \in {A}^{u}$.
For tuples $\tup{u} = (u_1,\dotsc,u_k)$ of variables
and $\tup{a} = (a_1,\dotsc,a_k) \in {A}^{\tup{u}}$,
the assignment $\alpha[\tup{a}/\tup{u}]$
maps $u_i$ to $a_i$ for each $i \in [k]$,
and each variable $v \not \in \tilde{u}$ to $\alpha(v)$.
By $\varphi(u_1,\dotsc,u_k)$ we denote a formula $\varphi$
with $\free(\varphi) \subseteq \set{u_1,\dotsc,u_k}$, where $\free(\varphi)$ is the set of free variables in $\varphi$.
Given a formula $\varphi(u_1,\dotsc,u_k)$, a structure $A$
and $(a_1,\dotsc,a_k) \in A^{(u_1,\dotsc,u_k)}$,
we write $A \models \varphi[a_1,\dotsc,a_k]$ if $\varphi$ holds in $A$
with $u_i$ assigned to $a_i$ for each $i \in [k]$.
We write $\varphi[A,\alpha;\tup{u}]$
for the set of all tuples $\tup{a} \in \Domain{A}{\tup{u}}$
with $(A,\alpha[\tup{a}/\tup{u}]) \models \varphi$.
For a formula  $\varphi(\tup{u})$ (with $\free(\varphi) \subseteq \tilde{u}$)
we also denote $\varphi[A,\alpha;\tup{u}]$ by $\varphi[A;\tup{u}]$, and 
for a formula  $\varphi(\tup{v},\tup{u})$ 
and $\tup{a}\in A^{\tup{v}}\!$,
we denote
$\varphi[A,\alpha[\tup{a}/\tup{v}];\tup{u}]$ also by $\varphi[A,\tup{a};\tup{u}]$.

$\FOC$ is obtained by extending $\FO$ with the following
formula formation rules:
\begin{itemize}
 \item $\phi:= p \leq q$ is a formula if $p,q$ are number variables. 
 We let $\free(\phi):= \{p,q\}$.
 \item $\phi':= \# \tup{u}\,\psi = \tup{p}$ is a formula if $\psi$ is a formula, $\tup{u}$ is a tuple of individual variables and $\tup{p}$ a tuple of number variables.
 We let $\free(\phi'):= (\free(\psi) \setminus \tilde{u}) \cup \tilde{p}$.
\end{itemize}
To define the semantics, let $A$ be a structure and $\alpha$ be an assignment.
We let 
\begin{itemize}
 \item $(A,\alpha)\models p \leq q$ \,iff\, $\alpha(p)\leq \alpha(q)$, 
 \item $(A,\alpha)\models \# \tup{u}\,\psi = \tup{p}$ \,iff\,
 $|\psi[A,\alpha;\tup{u}]|=\num[A]{\alpha(\tup{p})}$,
\end{itemize}
where for tuples $\tup{n} = (n_1,\dotsc,n_k) \in N(A)^k$
we let $\num[A]{\tup{n}}$ be the number
\begin{align*}
  \num[A]{\tup{n}}\ \isdef\
  \sum_{i=1}^k\, n_i \cdot (\card{U(A)}+1)^{i-1}.
\end{align*}

\medskip

\noindent
\emph{Symmetric transitive closure logic (with counting)~$\STC\plusC$} 
is an extension of $\FO\plusC$ with $\stc$-op\-er\-a\-tors.
The set of all $\STC\plusC$-for\-mu\-las is obtained by extending 
the formula formation rules of $\FO\plusC$ by the following rule:
\begin{itemize}
 \item $\phi:=\stcx{\tup{u}}{\tup{v}}{\psi}(\tup{u}'\!,\tup{v}')$ is a formula if $\psi$ is a formula
 and $\tup{u},\tup{v},\tup{u}'\!,\tup{v}'$ are compatible tuples of structure (or number) variables.
 We let $\free(\phi):=\tilde{u}'\cup\tilde{v}'\cup\big(\free(\psi)\setminus(\tilde{u}\cup\tilde{v})\big)$.
\end{itemize}
Let $A$ be a structure and $\alpha$ be an assignment.
We let 
\begin{itemize}
 \item $(A,\alpha)\models \stcx{\tup{u}}{\tup{v}}{\psi}(\tup{u}'\!,\tup{v}')$ \,iff\,
$(\alpha(\tup{u}'),\alpha(\tup{v}'))$ 
is contained in the symmetric transitive
closure of $\psi[A,\alpha;\tup{u},\tup{v}]$.
\end{itemize}

\bigskip

\noindent
\emph{(Inflationary) fixed-point logic (with counting) $\FP\plusC$} 
is an extension of $\FO\plusC$ with atomic second order formulas and $\ifp$-op\-er\-a\-tors.
$\FP\plusC$ has a further type of variables: \emph{relational variables}.
A relational variable $X$ of arity $k$ ranges over relations $R\subseteq W_1\times\cdots\times W_k$
where $W_i=U(A)$ (or $W_i=N(A)$) for all $i\in[k]$.
We let $A^X:=\CP(W_1\times\cdots\times W_k)$.
We say a relational variable $X$ and a tuple $\tup{u}$ of individual variables are \emph{compatible} if $A^{\tup{u}}\in A^X\!$.
We extend the formula formation rules of $\FO\plusC$ by the following two rules:
\begin{itemize}
 \item $\phi:=X\tup{u}$ is a formula if $X$ is a relational variable and $\tup{u}$ is a tuple of structure (or number) variables such that $X$ and $\tup{u}$ are compatible.
 We let  $\free(\phi):=\tilde{u}\cup\{X\}$.
 \item $\phi':=\ifpx{X}{\tup{u}}{\psi}\tup{u}'$ is a formula if $\psi$ is a formula, 
 and $X$ is a relational variable, $\tup{u},\tup{u}'$ are tuples of structure (or number) variables such that
 $X,\tup{u},\tup{u}'$ are compatible.
 We let $\free(\phi'):=\tilde{u}'\cup\big(\free(\psi)\setminus(\tilde{u}\cup\{X\})\big)$.
\end{itemize}
Let $A$ be a structure and $\alpha$ be an assignment. We let
\begin{itemize}
 \item $(A,\alpha)\models X\tup{u}$ \,iff\, $\alpha(\tup{u})\in \alpha(X)$, 
 \item $(A,\alpha)\models \ifpx{X}{\tup{u}}{\psi}\tup{u}'$ \,iff\, $\alpha(\tup{u}')\in F_\infty$,
 \end{itemize}
where $F_\infty$ is defined as follows:
Let $F\colon A^X \to A^X$ be the mapping defined by 
$F(R):=R\cup \psi[A,\alpha[R/X];\tup{u}]$ for all $R\in A^X$.
We let $F_0:=\emptyset$ and $F_{i+1}:=F(F_i)$ for all $i\geq 0$. 
Let $m\geq0$ be such that $F_m=F_{m+1}$. 
Then $F_\infty:=F_m$.

We also use the property that \emph{simultaneous inflationary fixed-point logic}
has the same expressive power as $\FP$. 
For the syntax and semantics of this logic we refer the reader to \cite{groheDC} or \cite{ebbflu99}.

\bigskip

\noindent
For logics $\logic{L},\logic{L}'$ we write $\logic{L} \leq \logic{L}'$
if $\logic{L}$ is semantically contained in $\logic{L}'\!$.
We have $\STC\leq \FP$ and $\STCC\leq \FPC$.
Note that simple arithmetics like addition and multiplication are definable in $\STCC$.

\subsection{Transductions}\label{sec:transduction-allg}

Transductions (also known as \emph{syntactical interpretations})
define certain structures within other structures.
More on transductions can be found in \cite{groheDC,diss}.
In this section, we introduce transductions, consider compositions of tranductions,
and present the new notion of counting transductions.

In the following we introduce parameterized transductions for 
$\FPC$.
As parameter variables of these transductions, we allow individual variables as well as relational variables.
The domain variables are individual variables.

\begin{definition}[Parameterized $\FPC$-Transduction]\label{def:paratransduction}
  \label{def:interpr}
  Let $\tau_1,\tau_2$ be vocabularies.
  \begin{enumerate}
  \item
    A \emph{parameterized $\FPC[\tau_1,\tau_2]$-trans\-duc\-tion}  is a tuple
    \begin{align*}
     \qquad\quad\Theta(\tup{X}) \,=\,
      \Bigl(& 
        \theta_\dom(\tup{X}),
        \theta_U(\tup{X},\tup{u}),
        \theta_\approx(\tup{X},\tup{u},\tup{u}'),\bigl(
          \theta_R(\tup{X},\tup{u}_{R,1},\dots,\tup{u}_{R,{\ar(R)}})
        \bigr)_{R \in \tau_2}
      \Bigr)
    \end{align*}
    of $\FPC[\tau_1]$-formulas,
     where 
     $\tup{X}$ is a tuple of individual or relational variables, and
	$\tup{u},\tup{u}'$ and $\tup{u}_{R,i}$ for every ${R \in \tau_2}$ and $i \in [\ar(R)]$  
	are compatible tuples of individual~variables.
  \item
    The \emph{domain} of $\Theta(\tup{X})$ is the class $\Dom(\Theta(\tup{X}))$ of all 
    pairs $(A,\hspace{-1.5pt}\tup{P})$ such that 
    ${A\hspace{-1pt}\models \hspace{-1pt}\theta_\dom[\tup{P}]}$ und $\theta_U[A,\tup{P};\tup{u}]$ is not empty, 
    where $A$ is a $\tau_1$-struc\-ture and $\tup{P}\in A^{\tup{X}}\!$.
    The variables occurring in tuple $\tup{X}$ are called \emph{parameter variables}, and the ones occurring in $\tup{u}$ 
    are referred to as \emph{domain variables}. The elements in $\tup{P}$ are called \emph{parameters}.
	\item 
	Let $(A,\tup{P})$ be in the domain of $\Theta(\tup{X})$.
	We define a $\tau_2$-struc\-ture $\Theta[A,\tup{P}]$ as follows.  
	Let $\approx$ be the equivalence relation generated by $\theta_\approx[A,\tup{P};\tup{u},\tup{u}']$ on $A^{\tup{u}}$.
    We let
    \begin{align*}
      U(\Theta[A,\tup{P}]) \,:=\, \theta_U[A,\tup{P};\tup{u}]\modout_{\approx}
    \end{align*}
    be the universe of $\Theta[A,\tup{P}]$.
    Further, for each $R \in \tau_2$, we let 
    \begin{align*}
      \qquad R(\Theta[A,\tup{P}])\, :=\, \Big(\theta_R[A,\tup{P};\tup{u}_{R,1},\dots,\tup{u}_{R,{\ar(R)}}] \cap 
      \theta_U[A,\tup{P};\tup{u}]^{\ar(R)}\Big)\Big\modout_{\!\approx}.\\[-1em]
    \end{align*}
  \end{enumerate}
\end{definition}

\noindent
A parameterized $\FPC[\tau_1,\tau_2]$-trans\-duc\-tion  defines a parameterized mapping
from  $\tau_1$-struc\-tures into $\tau_2$-struc\-tures
via $\FPC[\tau_1]$-for\-mu\-las. 
If $\theta_{\dom}:=\true$ or $\theta_\approx:=\false$, we omit the respective formula in the presentation of the transduction.

A parameterized $\FPC[\tau_1,\tau_2]$-trans\-duc\-tion $\Theta(\tup{X})$ is an \emph{$\FPC[\tau_1,\tau_2]$-trans\-duc\-tion} if $\tup{X}$ is the empty tuple.
Let $\tup{X}$ be the empty tuple. 
For simplicity, we denote a trans\-duc\-tion $\Theta(\tup{X})$ by $\Theta$, 
and we write $A\in \Dom(\Theta)$ if $(A,\tup{X})$ is contained in the domain of~$\Theta$.

Let $\CC_1$ be a class of $\tau_1$-struc\-tures and $\CC_2$ be a class of $\tau_2$-struc\-tures.
We call a mapping $f$ from $\CC_1$ to $\CC_2$ \emph{$\FPC$-de\-fin\-a\-ble},
if there exists an $\FPC[\tau_1,\tau_2]$-trans\-duc\-tion $\Theta$ such that 
$\CC_1\subseteq\Dom(\Theta)$ and
for all $\tau_1$-struc\-tures $A\in \CC_1$ we have $f(A)=\Theta[A]$.

An important property of $\FPC[\tau_1,\tau_2]$-trans\-duc\-tions
is that, they allow to \emph{pull back} $\tau_2$-for\-mu\-las,
which means that for each $\tau_2$-for\-mu\-la there exists an $\tau_1$-for\-mu\-la that expresses essentially the same.
This property is the core of the Transduction Lemma.
A proof of the Transduction Lemma can be found in \cite{diss}.

\begin{proposition}[Transduction Lemma]  \label{prop:transduction-lemma}
  Let $\tau_1,\tau_2$ be vocabularies.
  Let $\Theta(\tup{X})$
  be a parameterized $\FPC[\tau_1,\tau_2]$-trans\-duc\-tion, where $\ell$-tu\-ple $\tup{u}$ is the tuple of domain variables.
  Further, let $\psi(x_1,\dotsc,x_\kappa,p_1,\dotsc,p_\lambda)$
  be an $\FPC[\tau_2]$-for\-mu\-la where 
  $x_1,\dotsc,x_\kappa$  are structure variables and $p_1,\dotsc,p_\lambda$ are number variables.
  Then there exists an $\FPC[\tau_1]$-for\-mu\-la
  $\psi^{-\Theta}(\tup{X}, \tup{u}_1,\dotsc,\tup{u}_\kappa,
      \tup{q}_1,\dotsc,\tup{q}_\lambda),$
  where 
  $\tup{u}_1,\dotsc,\tup{u}_\kappa$ are compatible with $\tup{u}$ and
  $\tup{q}_1,\dotsc,\tup{q}_\lambda$ are $\ell$-tu\-ples of number variables,
  such that for all $(A,\tup{P})\in \Dom(\Theta(\tup{X}))$,  
  all $\tup{a}_1,\dotsc,\tup{a}_\kappa \in A^{\tup{u}}$ and
  all $\tup{n}_1,\dotsc,\tup{n}_\lambda \in N(A)^\ell$,
  \begin{align*}
    A \models
    \psi^{-\Theta}[\tup{P},\tup{a}_1,\dots,\tup{a}_\kappa,
      \tup{n}_1,\dots,\tup{n}_\lambda]
    \iff\
    & {\tup{a}_1}\modout_{\approx},\dots,{\tup{a}_\kappa}\modout_{\approx}
      \in U(\Theta[A,\tup{P}]), \\
    & \!\num[A]{\tup{n}_1},\dots,\num[A]{\tup{n}_\lambda}
      \in N(\Theta[A,\tup{P}])\text{ and}\\
    & \Theta[A,\tup{P}] \models
      \psi\bigl[
       {\tup{a}_1}\modout_{\approx},\dots,{\tup{a}_\kappa}\modout_{\approx},
        \num[A]{\tup{n}_1},\dots,\num[A]{\tup{n}_\lambda}       
      \bigr],
  \end{align*}
  where $\approx$ is the equivalence relation generated by $\theta_\approx[A,\tup{P};\tup{u},\tup{u}']$ on $A^{\tup{u}}$.
\end{proposition}
\smallskip

\noindent
The following proposition shows that the composition of a parameterized transduction and a transduction
is again a parameterized transduction.

\begin{propC}[\cite{diss}]\label{prop:composition}
Let $\tau_1$, $\tau_2$ and $\tau_3$ be vocabularies. 
	Let $\Theta_1\big(\tup{X}\big)$ be a parameterized $\FPC[\tau_1,\tau_2]$-trans\-duc\-tion and 
	$\Theta_2$ be an $\FPC[\tau_2,\tau_3]$-trans\-duc\-tion.
	Then there exists a parameterized $\FPC[\tau_1,\hspace{-1pt}\tau_3]$-trans\-duc\-tion $\Theta\big(\hspace{-1pt}\tup{X}\big)$\hspace{-0.5pt}
	such that  
	for all $\tau_1$-struc\-tures $A$ and all~$\tup{P}\hspace{-2pt}\in\hspace{-2pt} A^{\tup{X}}$\!, 
	\begin{align*}
		\big(A,\tup{P}\big)\in \Dom\big(\Theta\big(\tup{X}\big)\big)\iff
		\big(A,\tup{P}\big)\in \Dom\big(\Theta_1\big(\tup{X}\big)\big) \text{ and }\;	
		\Theta_1\big[A,\tup{P}\big]\in \Dom\big(\Theta_2\big),
	\end{align*}
	and for all $\big(A,\tup{P}\big)\in \Dom\big(\Theta\big(\tup{X}\big)\big)$,
	\begin{align*}	
		\Theta\big[A,\tup{P}\big]\cong\Theta_2\big[\Theta_1\big[A,\tup{P}\big]\big].\\[-1em]
	\end{align*}
\end{propC}

\noindent
In the following we introduce the new notion of parameterized counting transductions.
The universe of the structure $\Theta^{\raute}[A,\tup{P}]$ defined by a parameterized counting transduction $\Theta^{\raute}(\tup{X})$
automatically includes the number sort $N(A)$ of $A$, for all structures $A$ and tuples~$\tup{P}$ of parameters 
from the domain of $\Theta^{\raute}(\tup{X})$.
Parameterized counting transductions are as powerful as parameterized transductions.
Presenting a parameterized counting transduction instead of a parameterized transduction will contribute to a clearer presentation.

\begin{definition}[Parameterized $\FPC$-Counting Transduction]\label{def:CountingTransduction}
\hspace{-1pt}Let $\tau_1,\hspace{-1pt}\tau_2$ be vocabularies. 
  \begin{enumerate}
  \item
    A  \emph{parameterized $\FPC[\tau_1,\tau_2]$-count\-ing transduction} is a tuple
    \begin{align*}
      \qquad\qquad\Theta^{\raute}(\tup{X}) \,=\,
      \Bigl( \theta^{\raute}_{\dom}(\tup{X}),
        \theta^{\raute}_U(\tup{X},\tup{u}),
        \theta^{\raute}_\approx(\tup{X},\tup{u},\tup{u}'),
        \bigl(
          \theta^{\raute}_R(\tup{X},\tup{u}_{R,1},\dots,\tup{u}_{R,{\ar(R)}})
        \bigr)_{R \in \tau_2}
      \Bigr)
    \end{align*}
    of $\FPC[\tau_1]$-formulas,
     where 
     $\tup{X}$ is a tuple of individual or relational  variables,
	$\tup{u},\tup{u}'$ are compatible tuples of individual variables but not tuples of number variables of length~$1$,
    and for every $R \in \tau_2$ and $i \in [\ar(R)]$,
    $\tup{u}_{R,i}$~is a tuple of variables
    that is compatible to $\tup{u}$ or a tuple of number variables of length $1$. 
\item
    The \emph{domain} of $\Theta^{\raute}(\tup{X})$ 
    is the class $\Dom(\Theta^{\raute}(\tup{X}))$ 
    of all pairs $(A,\hspace{-1pt}\tup{P})$ where $A$ is a $\tau_1$-struc\-ture, 
    $\tup{P}\in A^{\tup{X}}$ and $A\models \theta^{\raute}_{\dom}[\tup{P}]$.\vspace{2pt}
\item 
    Let $(A,\hspace{-1pt}\tup{P})$ be in the domain of $\Theta^{\raute}(\tup{X})$.
    We define a $\tau_2$-struc\-ture $\Theta^{\raute}[A,\hspace{-1pt}\tup{P}]$ as follows.
    Let $\approx$ be the equivalence relation generated by $\theta^{\raute}_\approx[A,\hspace{-1pt}\tup{P}\hspace{0.5pt};\tup{u},\tup{u}']$
    on the set $A^{\tup{u}}\dcup N(A)$.
    We let\vspace{-4pt}
    \begin{align*}
      U(\Theta^{\raute}[A,\hspace{-1pt}\tup{P}]) \,:=\, \big(\theta^{\raute}_U[A,\hspace{-1pt}\tup{P}\hspace{0.5pt};\tup{u}]\; 
      \dcup\ N(A)\big)\modout_{\approx}
    \end{align*}
    be the universe of $\Theta^{\raute}[A,\hspace{-1pt}\tup{P}]$.
    Further, for each $R \in \tau_2$, we let 
    \begin{align*}
    \qquad\qquad R(\Theta^{\raute}[A,\hspace{-1pt}\tup{P}])\,:=\, 
     \Big(
     \theta_R^{\raute}[A,\hspace{-1pt}\tup{P}\hspace{0.5pt};\tup{u}_{R,1},\dots,\tup{u}_{R,{\ar(R)}}] \cap 
     \big(\theta_U^{\raute}[A,\hspace{-1pt}\tup{P}\hspace{0.5pt};\tup{u}]\dcup N(A)\big)^{\hspace{-1pt}\ar(R)}
     \Big)\Big\modout_{\!\approx}.
\end{align*}
  \end{enumerate}
\end{definition}

\begin{propC}[\cite{diss}]\label{thm:CountingTransduction}
Let $\tau_1,\tau_2$ be vocabularies. 
 Let $\Theta^{\raute}(\bar{X})$ be a parameterized\hspace{1pt} \vspace{0.5pt}$\FPC[\tau_1,\tau_2]$-count\-ing transduction. 
 Then there exists a parameterized\hspace{1pt} $\FPC[\tau_1,\tau_2]$-trans\-duc\-tion $\Theta(\bar{X})$ such that  
 \begin{itemize}
 	\item $\Dom(\Theta(\tup{X}))=\Dom(\Theta^{\raute}(\tup{X}))$ and 
 	\item $\Theta[A,\tup{P}] \cong \Theta^{\raute}[A,\tup{P}]$ for all $(A,\tup{P})\in \Dom(\Theta(\tup{X}))$.
 \end{itemize}
\end{propC}

\subsection{Canonization}\label{sec:canonization}

In this section we introduce ordered structures, (definable) canonization and the capturing of $\PTIME$.
A more detailed introduction 
can be found in~\cite{groheDC} and \cite{ebbflu99}.

Let $\tau$ be a vocabulary with $\leq\ \not \in \tau$.
A $\tau\cup\{\leq\}$-struc\-ture $A'$ is \emph{ordered} if the relation symbol~$\leq$ is interpreted 
as a linear order on the universe of~$A'\!$.
Let $A$ be a $\tau$-struc\-ture. An ordered $\tau\cup\{\leq\}$-struc\-ture $(A'\!,\leq_{A'})$ is an \emph{ordered copy} of $A$
if $A'\cong A$.
Let $\CC$ be a class of $\tau$-struc\-tures. 
A mapping $f$ is a \emph{canonization mapping} of $\CC$ if
it assigns every structure $A\in \CC$ to an ordered copy $f(A)=(A_f,\leq_{A_f})$ of $A$
such that for all structures $A,B\in\CC$ we have $f(A)\cong f(B)$ if $A\cong B$.
We call the ordered structure $f(A)$ the \emph{canon} of $A$.

Let $\Theta(\bar{x})$ be a parameterized $\FPC[\tau,\tau\cup\{\leq\}]$-trans\-duc\-tion, where $\bar{x}$ is a tuple of individual variables.
We say $\Theta(\bar{x})$ \emph{canonizes} a $\tau$-struc\-ture $A$ 
if there exists a tuple $\bar{p}\in A^{\bar{x}}$ such  that ${(A,\bar{p})\in \Dom(\Theta(\bar{x}))}$,
and for all tuples $\bar{p}\in A^{\bar{x}}$ with $(A,\bar{p})\in \Dom(\Theta(\bar{x}))$, 
the $\tau\cup\{\leq\}$-struc\-ture $\Theta[A,\bar{p}]$ is an ordered copy of $A$.\footnote{
Note that if the tuple $\tup{x}$ of parameter variables is the empty tuple, $\FPCf[{\tau,\tau\cup\{\leq\}}]$-trans\-duc\-tion $\Theta$ 
canonizes a $\tau$-struc\-ture $A$ 
if ${A\in \Dom(\Theta)}$ and the ${\tau\cup\{\leq\}}$-struc\-ture $\Theta[A]$ is an ordered copy of~$A$.}
A \emph{(parameterized) $\FPC$-can\-on\-iza\-tion} of a class $\CC$ of $\tau$-struc\-tures 
is a (parameterized) $\FPC[\tau,\tau\cup\{\leq\}]$-trans\-duc\-tion that canonizes all~${A\in \CC}$. 
A class~$\CC$ of $\tau$-struc\-tures \emph{admits $\FPC$-de\-fin\-a\-ble (parameterized) canonization} 
if $\CC$ has a (parameterized) $\FPC$-can\-on\-iza\-tion.

The following lemma shows that parameters can be eliminated from $\FPC$-can\-on\-iza\-tions.
\begin{lemC}[{\cite[Lemma~3.3.18]{groheDC}}\footnote{
\cite[Lemma~3.3.18]{groheDC} is shown for $\IFPCf$. 
Note that Lemma~3.3.18 states that there exists an $\IFPCf$-can\-on\-iza\-tion of $\CC$ 
without parameters that is also \emph{normal}.
}]\label{lem:parametersInCanonizations}
Let $\CC$ be a class of $\tau$-struc\-tures.
If $\CC$ admits $\FPC$-de\-fin\-a\-ble parameterized canonization, then there exists 
an $\FPC$-can\-on\-iza\-tion of $\CC$ without parameters.
\end{lemC}

\noindent
We can use definable canonization of a 
graph class 
to prove that $\PTIME$ is captured on this graph class.
Let $\Logic$ be a logic and $\CC$ be a graph class.
$\Logic$ \emph{captures} $\PTIME$ \emph{on}~$\CC$
if for each class $\CA\subseteq \CC$,
there exists an $\Logic$-sen\-tence defining $\CA$ 
if and only if $\CA$ is $\PTIME$-de\-cid\-able.\footnote{
A precise definition of what it means that a logic \emph{(effectively strongly) captures} 
a complexity class can be found in~\cite[Chapter~11]{ebbflu99}. 
}
If $\Logic$ captures $\PTIME$ on the class of all graphs, 
then $\Logic$ \emph{captures} $\PTIME$ \cite[Theorem 11.2.6]{ebbflu99}.\pagebreak[1]
A fundamental result was shown by Immerman and Vardi:\footnote{
Immerman and Vardi proved this capturing result not only for the class of ordered graphs but for the class of ordered structures.}
\begin{thmC}[\cite{imm86,var82}] 
$\FP$ captures $\PTIME$ on the class of all ordered graphs.
\end{thmC}

\noindent
Let us suppose there exists a parameterized $\FPC$-can\-on\-iza\-tion of a graph class $\CC$.
Since $\FP$ captures $\PTIME$ on ordered graphs and
we can pull back each $\FP$-sen\-tence that defines a poly\-no\-mi\-al-time property on ordered graphs under this canonization,
the capturing result of Immerman and Vardi transfers from ordered structures 
to the class~$\CC$.
\begin{proposition} 
Let $\CC$ be a class of graphs.
If $\CC$ admits $\FPC$-de\-fin\-a\-ble parameterized canonization, 
then $\FPC$ captures $\PTIME$ on $\CC$.
\end{proposition}

\section{Defining the Modular Decomposition in \texorpdfstring{STC+C}{$\STCC$}}\label{sec:ModDecompSTCC}

In this section we show that the modular decomposition of a graph 
is definable in $\STCC$. 

First, we introduce modules and modular decomposition in this section.
In order to show that the modular decomposition is definable in $\STCC$,
we consider modules that are spanned by two vertices, that is,
modules that contain the two vertices and are minimal with this property.
We use the concept of edge classes introduced by Gallai in \cite{gallai67transitiv}
to show that these spanned modules are definable in $\STCC$.
Afterwards, we show how the spanned modules are related to the modules occurring in the modular decomposition.
We obtain that the modular decomposition is definable in $\STCC$.
Consequently, it is computable in logarithmic space~\cite{rei05}.
Thus, the modular decomposition tree is computable in logarithmic space.
We conclude that cograph recognition and cograph canonization is in logarithmic space.

We use that the modular decomposition is $\STCC$-de\-fin\-able (actually we only require $\FPC$-de\-fin\-able)
in order to prove the Modular Decomposition Theorem in Section~\ref{sec:ModDecompThm}.

\subsection{Modules and their Basic Properties}\label{sec:modulesproperties}
Let $G=(V,E)$ be a graph. 
A non-emp\-ty subset $M\subseteq V$ is a \emph{module} of a graph $G$ if for 
all vertices $v\in V\setminus M$ 
and all $w,w'\in M$ we have \vspace{-1mm}
\begin{align*}
\{v,w\}\in E \iff \{v,w'\}\in E.
\end{align*}
All vertex sets of size $1$ are modules.
We call them \emph{singleton modules}.
Further, the vertex set $V$ is a module. 
We also refer to the module $V$ and the singleton modules as \emph{trivial modules}.
The connected components of $G$ or of the complement graph $\overline{G}$ are modules as well (see Figure~\ref{fig:module1}).
The same holds for unions of connected components.
Figure~\ref{fig:module2} shows a further example of modules in a graph.
A module $M$ is a \emph{proper module} if $M\subset V\!$.
We call a graph \emph{prime} if all of its modules are trivial modules. 
The path $P_i$ with $i\geq 4$ vertices, e.g., is a prime graph.
Notice that if $M$ is a module of a graph $G$, then $M$ is also a module of $\overline{G}$.
Therefore, a graph $G$ is prime if and only if  $\overline{G}$ is prime.

\begin{figure}[hbtp]
\centering
\begin{subfigure}[b]{0.58\textwidth}
\hspace{-2.4cm}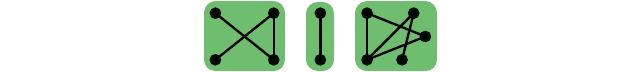\vspace{1mm}
\caption{The connected components of $G_1$ and $\overline{G_2}$ are modules. \phantom{aaaaaaaaaaaaaaaaa}}
\label{fig:module1}
\end{subfigure}
\hspace{0.6cm}%
\begin{subfigure}[b]{0.35\textwidth}
\hspace{-1.8cm}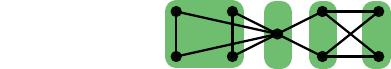\vspace{1mm}
\caption{The highlighted sets, for example, are modules of graph $G_3$.}
\label{fig:module2}
\end{subfigure}
\caption{Modules of Graphs}
\label{fig:module}
\end{figure}

The following observations contain fundamental but easily provable properties of modules.

\pagebreak[2]
\begin{observation}\label{obs:ModuleEdges}
	If $M_1$ and $M_2$ are modules of a graph $G$ with $M_1\cap M_2=\emptyset$, 
	then either there exist no edges between  vertices in $M_1$ and vertices in $M_2$, or
	every vertex in $M_1$ is adjacent to each vertex in $M_2$.
\end{observation}

\begin{observation}\label{obs:ModuleIntersection}
	If $M_1$ and $M_2$ are modules of a graph $G$  with $M_1\cap M_2\not=\emptyset$, 
	then $M_1\cap M_2$ and $M_1\cup M_2$ are modules as well.
\end{observation}

\begin{observation}\label{obs:ModuleInclusion}
Let $M'$ be a module of $G$, and let $M\hspace{-1pt}\subseteq\hspace{-1pt} M'\!$.
Then $M$ is a module of $G$ if and only if it is a module of $G[M']$.
\end{observation}

\subsection{Modular Decomposition}\label{sec:moddecompdefinition}
In the following we present the modular decomposition of a graph, which was introduced by Gallai in 1967 \cite{gallai67transitiv}.
The modular decomposition decomposes a graph with at least two vertices into proper modules. It can be applied recursively.

Let $G=(V,E)$ be an arbitrary graph with $|V|>1$.
We let $n$ be the number of vertices in $G$.
If $G$ (or~$\overline{G}$) is not connected, then every connected component of $G$ (or~$\overline{G}$) is a module, and 
we can partition the vertex set of $G$ (or $\overline{G}$) into its connected components.
If $G$ and $\overline{G}$ are connected, then there also exists a unique partition of $V$ into proper modules.
Gallai showed that in this case the maximal proper modules of $G$ form a partition of $V$ 
(Satz 2.9 and 2.11 in~\cite{gallai67transitiv}).
Figure~\ref{fig:module2} depicts the maximal proper modules of a graph $G$ where $G$ and $\overline{G}$ are connected. 

Thus, we can canonically partition each graph $G$ with $n>1$ into proper modules.
For a vertex $v$ of graph $G$ we let $D_{G}(v)$ be the respective proper module containing $v$.
Hence, for a vertex $v$ of a graph $G=(V,E)$ with $|V|>1$, the set $D_{G}(v)$ is%
\footnote{
We can also say $D_{G}(v)$ is the maximal strong proper module of $G$ that contains $v$.
(A module~$M$ is \emph{strong}, if  $M\cap M'=\emptyset$, $M\subseteq M'$ or $M'\subseteq M$ for all other modules $M'\!$.)
Gallai proved that the maximal strong proper modules partition $V(G)$ (Satz 2.11 in \cite{gallai67transitiv}), 
and that for each graph $G$ they coincide with the sets $D_{G}(v)$ as they are defined here \cite[Satz 2.9 and 2.10]{gallai67transitiv}.}
\begin{itemize}
	\item the connected component of $G$ that contains $v$ if $G$ is not connected,
	\item the connected component of $\overline{G}$ that contains $v$ if $\overline{G}$ is not connected, or
	\item the maximal proper module of $G$ that contains $v$ if $G$ and $\overline{G}$ are connected.
\end{itemize}
If the graph $G$ has only one vertex $v$, we let $D_{G}(v):=\{v\}$.

We define the \emph{(recursive) modular decomposition}
of $G$ as the following family 
of subsets $D_{i,v}\subseteq V$ with $i\in[0,n]$, $v\in V$. We let $D_{0,v}:=V$ for all $v\in V$,
and for $i\in[0,n-1]$ we define $D_{i+1,v}$ for all $v\in V$ recursively: 
\[D_{i+1,v}:= D_{G[D_{i,v}]}(v).\]
As an example, a graph and its modular decomposition is illustrated in Figure~\ref{fig:modular-dec-tree-ex}.
\begin{figure}[htbp]
\centering
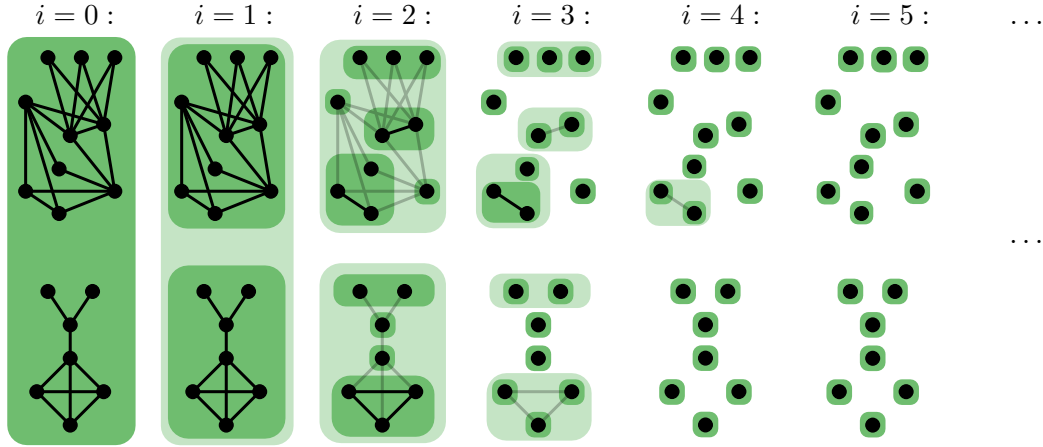
\caption{Modular decomposition of a graph}
\label{fig:modular-dec-tree-ex}
\end{figure}

It is easy to see that there exists a $k\in [0,n]$ such that
${V\hspace{-2pt}=\hspace{-1pt}D_{0,v}\supset D_{1,v}\supset\dots \supset D_{k,v}\hspace{-1pt}=\hspace{-1pt}\{v\}}$ and
that $D_{i,v}=\{v\}$ for all $i\geq k$.
Thus, $D_{n,v}=\{v\}$ for all $v\in V\!$.
For all $i\in [0,n]$ and all $v\in V$ the set $D_{i,v}$ is a module of $G$  as we can apply
Observation~\ref{obs:ModuleInclusion} inductively.
Further, an easy induction shows that the set $\{D_{i,v}\mid v\in V\}$ is a partition of the vertex set $V$ for all 
$i\in [0,n]$. Hence, we can conclude the following:
\begin{observation}
For all $v,w\in V$ and all $i\in [0,n]$, the modules $D_{i,v}$ and $D_{i,w}$ are equal if and only if $w\in D_{i,v}$.
\end{observation}

\subsection{Spanned Modules and (W)edge Classes}

Let $v,w\in V$ be vertices of $G$. 
We let \emph{$M_{v,w}$} be the intersection of all modules of $G$ that contain $v$ and $w$.
Since $V$ is a module, 
Observation~\ref{obs:ModuleIntersection} implies that $M_{v,w}$ is a module.
Consequently, $M_{v,w}$ is the smallest module containing $v$ and $w$.
We say the vertices $v,w\in V$ \emph{span} a module $M$ if $M=M_{v,w}$, 
and call $M$ a \emph{spanned module} if there exists $v,w\in V$ that span $M$.
Trivially, $M_{v,v}=\{v\}$.

Let $e,e'\!\in E$ be two edges of $G$. 
The edges $e$ and $e'$ \emph{form a wedge} in $G$ (we also write \emph{$e\edgewedge e'$}) if there exist three distinct vertices $u,v,w\in V$
such that $e=\{u,v\}$, $e'\!=\{u,w\}$ and there is no edge between $v$ and $w$.
Clearly, $e\edgewedge e'$ implies $e'\edgewedge e$.
We call 
$\edgewedge$ the \emph{wedge relation}
on $E$.
The edges $e$ and $e'$ are \emph{wedge connected} if there exists a $k\geq 1$ and a sequence of edges $e_1,\dots, e_k$,
such that $e=e_1$, $e'=e_k$ and $e_i\edgewedge e_{i+1}$ for all $1\leq i<k$.
It is not hard to see that wedge connectivity is an equivalence relation on the set of edges of the graph.
The equivalence classes are the \emph{edge classes}
of $G$.\footnote{ Edge classes (or Kantenklassen) are defined in \cite{gallai67transitiv}. 
We extend this definition to wedge classes.} 
Thus, the edge classes partition the set of edges of a graph.
Note that the edge classes of the complement graph~$\overline{G}$ of $G$ partition the set of edges of~$\overline{G}$,
and therefore, they partition the set ${V\choose 2}\setminus E$ of non-edges of $G$.
We define the \emph{wedge class} of $\{v,w\}\in{V\choose 2}$
as the edge class of $G$ that contains $\{v,w\}$ if $\{v,w\}$ is an edge of $G$,
or as the edge class of $\overline{G}$ that contains $\{v,w\}$ otherwise. 
For distinct vertices $v$ and $w$ we let \emph{$W_{v,w}$} be the union of all elements in the wedge class of $\{v,w\}$. 
Clearly, we have $v,w\in W_{v,w}$.

\begin{example}
    Consider the graph $H$ that is depicted in Figure~\ref{fig:examplewedge-a}.
  The edge classes of $H$ are illustrated in Figure~\ref{fig:examplewedge-b} and the 
  edge classes of $\overline{H}$ in Figure~\ref{fig:examplewedge-c}.
  Further, we have $W_{e,f}=\{d,e,f\}$ and $W_{b,f}= V(H)\setminus\{c\}$.
\end{example}
\begin{figure}[htbp]
\centering
\begin{subfigure}[b]{0.25\textwidth}
\hspace{0.2cm}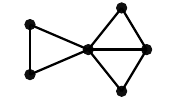\vspace{1mm}
\caption{Graph $H$ }
\label{fig:examplewedge-a}
\end{subfigure}
\hspace{0.6cm}%
\begin{subfigure}[b]{0.25\textwidth}
\hspace{0.3cm}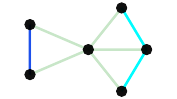\vspace{1mm}
\caption{Edge classes of $H$}
\label{fig:examplewedge-b}
\end{subfigure}
\hspace{0.6cm}%
\begin{subfigure}[b]{0.25\textwidth}
\hspace{0.3cm}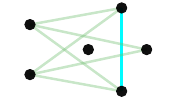\vspace{1mm}
\caption{Edge classes of $\overline{H}$}
\label{fig:examplewedge-c}
\end{subfigure}
\caption{}
\label{fig:examplewedge}
\end{figure}

\noindent
The following lemma follows directly from Satz 1.5 in \cite{gallai67transitiv} where the lemma is shown for all $v,w\in V$ with $\{v,w\}\in E$.
\begin{lemma}
Let $v,w\in V$ with $v\not=w$. Then $W_{v,w}\subseteq M_{v,w}$ and $W_{v,w}$ is a module.
\end{lemma}
\begin{corollary}\label{cor:MvwgleichWvw}
 $M_{v,w}= W_{v,w}$ for all $v,w\in V$ with $v\not=w$.
\end{corollary}

\noindent
In the following lemma we show that spanned modules are definable in symmetric transitive closure logic. 

\begin{lemma}\label{lem:spannedmodule-stc}
	There exists an $\STC$-for\-mu\-la $\varphi_M(x_1,x_2,y)$ 
	such that for all pairs $(v_1,v_2)\in V^2$ of vertices of $G$, the set
	$\varphi_M[G,v_1,v_2;y]$ is the module spanned by  $v_1$ and $v_2$.
\end{lemma}

\begin{proof}

Clearly, there exists an $\STC$-formula that defines the module spanned by two vertices if the vertices are equal.
In order to define the module $M_{v_1,v_2}$ spanned by two distinct vertices $v_1$ and $v_2$,
we apply Corollary~\ref{cor:MvwgleichWvw} and use the definition of $W_{v_1,v_2}$.

First of all, we need a formula 
for the wedge relation, that is, a formula
which is satisfied for vertices $w_1,w_2,w_1',w_2'\in V$
if, and only if, $\{w_1,w_2\}\edgewedge \{w_1',w_2'\}$ in $G$.
Clearly, this is precisely the case if there exist $i,j\in[2]$ such that 
\begin{align*}
&\text{$w_i=w_j'$, $w_{3-i}\not =w_{3-j}'$,\quad and}\\
&\text{$\{w_1,w_2\}\in E$, $\{w_1',w_2'\}\in E$, $\{w_{3-i},w_{3-j}'\}\not\in E$.}
\end{align*}
Thus, we obtain an $\FO$-for\-mu\-la for the wedge relation by taking the disjunction of the above statement over all $i,j\in[2]$.
Since the wedge relation is symmetric, we can use the $\STC$-op\-er\-a\-tor to express wedge connectivity.
Hence, there exists an $\STC$-for\-mu\-la that expresses wedge connectivity in $G$, and similarly we obtain one for wedge connectivity in $\overline{G}$, as well.
Using these formulas we are able to define the wedge classes of a graph. 
Consequently, we can also define the set $W_{v_1,v_2}$ for distinct vertices $v_1,v_2\in V$ in $\STC$.

Now, we can define $\varphi_M$ such that it distinguishes between the cases of whether the spanning vertices are equal or not 
and defines the spanned module accordingly.
\end{proof}

\subsection{Defining the Modular Decomposition in \texorpdfstring{STC+C }{$\STCC$}}\label{sec:definingMDinSTCC}

Let us fix a vertex $v\in V\!$.
In this section, our goal is to define the sets $D_{i,v}$ for $i\in[0,n]$ in $\STCC$. 
In order to do this, we show that each set $D_{i,v}$ can be constructed out of certain modules $M_{v,w}$ of $G$ with $w\in V\!$.

First, we take a look at two results (Lemma~\ref{lem:different-modules} and~\ref{lem:ModulesConnectedComponentsMvw}) of Gallai.
They will help us to gain a better understanding of the connection between $D_{i,v}$ and the sets $M_{v,w}$.

\begin{lemC}[{\cite[Satz 2.9 and 2.11 in connection with Satz 1.2 (3b)\footnote{
	In~\cite{gallai67transitiv}
	Gallai showed this lemma for the set $W_{v,w}$ instead of $M_{v,w}$.\label{ftn:GallaiWvw}}]{gallai67transitiv}}]\label{lem:different-modules}
	Suppose $G$ and $\overline{G}$ are connected and let $M'\!,M''$ be maximal proper modules of $G$ with $M'\not=M''\!$.
	Further let $v\in M'$ and $w\in M''\!$. Then $M_{v,w}=V\!$.
\end{lemC}

\begin{corollary}\label{cor:characterization-D_iv-conn}
	Let $i\in [0,n-1]$ and $v\in V\!$.
	If $G[D_{i,v}]$ and its complement are connected and $|D_{i,v}|>1$, then
	for all vertices $w\in D_{i,v}\setminus D_{i+1,v}$ we have
	$D_{i,v}=M_{v,w}$.\footnote{ The module $M_{v,w}$ always refers to the graph $G$.\label{ftn:Mvw}}
\end{corollary}

\begin{lemC}[{\cite[Satz 1.2 (2)\textsuperscript{\ref{ftn:GallaiWvw}}]{gallai67transitiv}}]\label{lem:ModulesConnectedComponentsMvw}
	Suppose $G$ is not connected and let $v$ and $w$ be in different connected components  $C_v$ and $C_w$ of $G$.
	Then $M_{v,w}=C_v\dcup C_w$.
\end{lemC}

\begin{corollary}\label{cor:characterization-D_iv-not-conn}
	Let $i\in [0,n-1]$ and $v\in V$.
	If $G[D_{i,v}]$ or its complement is not connected,
	then for all $w\in D_{i,v}\setminus D_{i+1,v}$  we have
	$M_{v,w}=D_{i+1,w}\dcup D_{i+1,v}$.\textsuperscript{\textup{\ref{ftn:Mvw}}}
\end{corollary}

\noindent
From Corollary~\ref{cor:characterization-D_iv-conn} and~\ref{cor:characterization-D_iv-not-conn}
we can conclude that there exists a vertex $w\in V$ such that ${D_{i,v}=M_{v,w}}$ 
if $G[D_{i,v}]$ and its complement are connected, 
or if $G[D_{i,v}]$ or its complement consists of two connected components.
If $G[D_{i,v}]$ or its complement consists of more than two connected components,
then for each $w\in D_{i,v}$ we have $M_{v,w}\not = D_{i,v}$. 
However, $D_{i,v}$ is the union of all connected components $D_{i+1,w}$ with $w\in D_{i,v}$.
Thus, Corollary~\ref{cor:characterization-D_iv-not-conn} shows that $D_{i,v}$ is the union of all $M_{v,w}$ where $w\in D_{i,v}$ is in a connected component 
different from the one containing $v$.

Let $v\in V$ be fixed. So far, we have seen that we obtain each set $D_{i,v}$ by taking the union of certain submodules $M_{v,w}$ of $D_{i,v}$.
We show in the following that we can partition the vertex set $V$ into $A^v_0,\dots,A^v_k$ such that 
\[D_{i,v}=\bigcup \{M_{v,w}\mid{w\in A^v_i}\},\]
where $k$ is minimal with $D_{k,v}=\{v\}$.
In order to obtain this partition, we order the modules
$M_{v,w}$ with $w\in V$ with respect to proper inclusion.
This order is a strict weak order (Lemma~\ref{lem:IncomparabilityRelationSuccv}).
Hence, incomparability is an equivalence relation.
We define the relation \emph{$\prec_v$} on $V$ by letting 
\[w_1\prec_v w_2:\iff M_{v,w_2}\subset M_{v,w_1}.\]
Then incomparability regarding $\prec_v$ is
an equivalence relation on the vertex set $V\!$.
The resulting equivalence classes form the partition $\{A^v_0,\dots,A^v_k\}$.
Consequently, we obtain the sets $D_{i,v}$ by taking the union of all sets $M_{v,w}$ that are incomparable with respect to proper inclusion.
An example showing the connection between $D_{i,v}$, $M_{v,w}$ for $w\in V$, $\prec_v$ and the sets $A^v_0,\dots,A^v_k$ 
for a specific vertex $v\in V$ is given in Figure~\ref{fig:divmvwa1ak}.

\begin{figure}[htbp]
\newcolumntype{C}[1]{>{\hsize=#1\hsize\centering\arraybackslash}X}%
\begin{subfigure}[t]{\textwidth}
\renewcommand{\arraystretch}{2.0}
\begin{tabularx}{\textwidth}{|p{2.2cm}| C{1.2}  C{1.1}  C{1.0} C{0.7} |}\cline{2-5}
\multicolumn{1}{r|}{}&$i=0$: &$i=1$:& $i=2$:& $i=3$:\\\hline\vspace*{-0.18cm}
$D_{i,a}$:&
\multicolumn{1}{c}{\vtop{\vskip-2ex\hbox{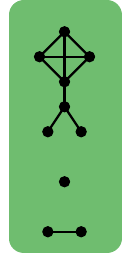}\vskip1ex}}&
\multicolumn{1}{c}{\vtop{\vskip-1.8ex\hbox{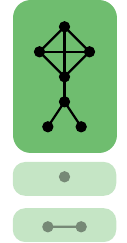}}}&
\multicolumn{1}{c}{\vtop{\vskip-1.6ex\hbox{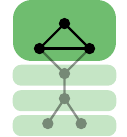}}}&
\multicolumn{1}{c|}{\vtop{\vskip-1.4ex\hbox{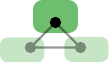}}}\\\hline
\vspace*{-0.18cm}
Sets $M_{a,w}$\newline for $w\in V$:\newline (ordered\newline according to\newline set inclusion)
&
\multicolumn{1}{c}{\vtop{\vskip-2ex\hbox{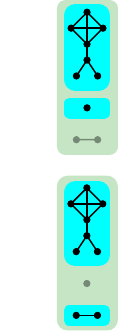}\vskip1ex}}
&
\multicolumn{1}{c}{\vtop{\vskip-2ex\hbox{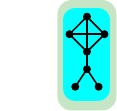}}}
&
\multicolumn{1}{c}{\vtop{\vskip-2ex\hbox{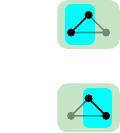}}}
&
\multicolumn{1}{c|}{\vtop{\vskip-2ex\hbox{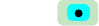}}}\\\hline
Relation $\prec_a$:&
\multicolumn{4}{l|}{
~~~~~~~~~$x,y,z$ 
~~~~~~~$\prec_a$
~~~~~~$d,e,f,g$ 
~~~~~~~$\prec_a$
~~~~~~$b,c$ 
~~~~~~~~$\prec_a$
~~~~~~$a$}
\\\hline
Partition $A^a_0,\dots,A^a_3$:&
$A^a_0=\{x,y,z\}$&$A^a_1=\{d,e,f,g\}$&$A^a_2=\{b,c\}$&$A^a_3=\{a\}$\\\hline
\end{tabularx}
\renewcommand{\arraystretch}{1.0}
\caption{$D_{i,a}$, $M_{a,w}$, $\prec_a$ and $A^a_0,\dots, A^a_3$}
\label{fig:divmvwa1ak-a}
\end{subfigure}
\vspace*{1cm}

\begin{subfigure}[b]{\textwidth}
\renewcommand{\arraystretch}{2.0}
\begin{tabularx}{\textwidth}{|p{2.2cm}| C{1.2} C{1} C{1.1} C{0.7} |}\hline
$G[D_{i,a}]$ is &
not connected &
connected&
connected&
connected\\
$\overline{G}[D_{i,a}]$ is&
connected&
connected&
not connected&
connected\\\hline
Result:&
\mbox{$D_{0,a}\!=\!M_{a,x}\cup M_{a,w}$}\newline for $w\!\in\!\{y,z\}$&
$D_{1,a}\!=\!M_{a,w}$\newline \mbox{for $w\!\in\!\{d,e,f,g\}$}&
$\!{D_{2,a}\!=\!M_{a,b}\cup M_{a,c}}$&
${D_{3,a}\!=\!M_{a,a}}$\\\hline
\end{tabularx}
\renewcommand{\arraystretch}{1.0}
\caption{Connection between the sets $D_{i,a}$ and $M_{a,w}$}
\label{fig:divmvwa1ak-c}
\end{subfigure}
\caption{}
\label{fig:divmvwa1ak}
\end{figure}

\begin{lemma}\label{lem:IncomparabilityRelationSuccv}
  For every $v\in V$ the relation $\prec_v$
  is a strict weak order.
\end{lemma}

\begin{proof}
 Let $v\in V$. It is easy to see that $\prec_v$ is transitive and irreflexive.
 Let us show that incomparability is transitive.
 Thus, let $w_1$ and $w_2$, and $w_2$ and $w_3$ be incomparable with respect to $\prec_v$, 
 and let us assume that $w_1$ and $w_3$ are comparable, that is, 
 without loss of generality we have $w_1\prec_v w_3$, which means $M_{v,w_1}\supset M_{v,w_3}$.
 Let $i\in\{0,\dots,n\}$ be maximal such that $D_{i,v}$ contains $M_{v,w_1}$, $M_{v,w_2}$ and $M_{v,w_3}$.
 
  First of all, we show that  $M_{v,w_j}\not=D_{i,v}$ for all $j\in\{1,2,3\}$.
  $M_{v,w_3}$ cannot be equal to  $D_{i,v}$ as $M_{v,w_3}$ is a proper subset of  $M_{v,w_1}$.
  If the module $M_{v,w_2}$ was equal to $D_{i,v}$, then 
	$M_{v,w_3}\subset M_{v,w_2}$, 
	and $w_3$ and $w_2$ would be comparable with respect to $\prec_v$. Thus, $M_{v,w_2}\not= D_{i,v}$.
	Finally, $M_{v,w_1}$ cannot be equal to $D_{i,v}$ either, 
	since $M_{v,w_1}=D_{i,v}$ implies that $M_{v,w_2}\subset M_{v,w_1}$ and then
	$w_2$ and $w_1$ would be comparable. 
	Consequently, neither of  $M_{v,w_1}$, $M_{v,w_2}$ and $M_{v,w_3}$ is equal to  $D_{i,v}$,
	and $|D_{i,v}|>1$.

 Now, if $G[D_{i,v}]$ and its complement are connected, we can partition
 $D_{i,v}$ into maximal proper modules, 
 and for all $j\in [3]$ we obtain that $M_{v,w_j}$ is a subset of module  $D_{i+1,v}$ if $w_j\in D_{i+1,v}$
 or equal to  $D_{i,v}$ if $w_j\in D_{i,v}\setminus D_{i+1,v}$ (Corollary~\ref{cor:characterization-D_iv-conn}).
 As we have shown above that $M_{v,w_j}\not=D_{i,v}$ for all $j\in[3]$, we have 
  $M_{v,w_1},M_{v,w_2},M_{v,w_3}\subseteq D_{i+1,v}$, which is a contradiction to the choice of $i$.

 If $G[D_{i,v}]$ is not connected, we can partition
 $D_{i,v}$ into its connected components.
 The case of $\overline{G}[D_{i,v}]$ being not connected can be treated analogously.
 For every $u\in D_{i,v}$, the set
 $D_{i+1,u}$ is the connected component of $G[D_{i,v}]$ containing $u$.
 Let us denote $D_{i+1,u}$ by $C_u$.
	Since $i$ has been chosen maximal, there has to be a $j\in\{1,2,3\}$ 
	such that $M_{v,w_j}$ is not contained in $C_v$.
	For this $j$, vertex $w_j$ must be a vertex in
	$M_{v,w_j}\setminus C_v$, and by Corollary~\ref{cor:characterization-D_iv-not-conn} 
	we obtain that $M_{v,w_j}=C_v\dcup C_{w_j}$.
As $w_1$ and $w_2$ are incomparable and $w_2$ and $w_3$ are incomparable,
independent from our choice of $j$, there exists an index $k\in\{1,2,3\}\setminus\{j\}$
such that $w_j$ and $w_k$ are incomparable.
Thus, $M_{v,w_k}$ cannot be a proper subset of  $C_v\dcup C_{w_j}$, and consequently,
$M_{v,w_k}\setminus C_v\not=\emptyset$. As above, we obtain that the module  $M_{v,w_k}$
is equal to $ C_v\dcup C_{w_k}$.
Let us assume $j=3$ or $k=3$.
The the module $M_{v,w_3}=C_v\dcup C_{w_3}$ is a proper subset of the module $M_{v,w_1}$.
Thus,  $M_{v,w_1}\setminus C_v\not=\emptyset$ and we can deduce $M_{v,w_1}=C_v\dcup C_{w_1}$ as we did before.
Since both $M_{v,w_1}$ and  $M_{v,w_3}$ are the union of two connected components, 
$M_{v,w_3}$ cannot be a proper subset of $M_{v,w_1}$.
Therefore, $j=1$ and $k=2$, or $j=2$ and $k=1$. 
As a consequence, we have $M_{v,w_1}=C_v\dcup C_{w_1}$ and  $M_{v,w_2}=C_v\dcup C_{w_2}$.
Now, if 
 $M_{v,w_3}\setminus C_v\not=\emptyset$, then
 $M_{v,w_3}$ is the disjoint union of the connected components 
 $C_v$ and $C_{w_3}$, a contradiction to $M_{v,w_3}\subset M_{v,w_1}$.
If $M_{v,w_3}$ is a subset of $C_v$,
then $M_{v,w_3}$ is a proper subset of $M_{v,w_2}$, which yields that $w_2$ and $w_3$ are comparable, a contradiction.
Hence, incomparability is transitive.
\end{proof}

\noindent
There exists an $\STC$-for\-mu\-la $\varphi_{_\prec}(x,y_1,y_2)$ such that
for all vertices $v,w_1,w_2\in V$ we have ${G\models \varphi_{_\prec}[v,w_1,w_2]}$ if, and only if,
$w_1\prec_v w_2$, that is,
the module spanned by $v,w_2$ is a proper subset of the module spanned by $v,w_1$.
Let $\varphi_M$ be the formula from Lemma~\ref{lem:spannedmodule-stc}.
Then 
\begin{align}\label{equ:varphiprec}
	\varphi_{_\prec}(x,y_1,y_2):=\ \ & \forall z \bigl(\varphi_M(x,y_2,z)\limplies \varphi_M(x,y_1,z) \bigr)\;\land\nonumber\\
	 &\exists z \bigl(\varphi_M(x,y_1,z) \land \lnot \varphi_M(x,y_2,z) \bigr).
\end{align}

According to Lemma~\ref{lem:IncomparabilityRelationSuccv} incomparability
with respect to $\prec_v$ is transitive. Hence, incomparability is an equivalence relation.
We write $w\sim_v w'$ if the vertices $w$ and $w'$ are incomparable.
We let  $[w]_v$ be the equivalence class of $w$, and $V\modout_{\sim_v}$ be the set of all equivalence classes.
Then ${V\modout_{\sim_v}=\{A^v_0,\dots,A^v_k\}}$.
We let $[z]_v\prec_v[w]_v$ if there exist vertices $z'\in[z]_v$ and $w'\in[w]_v$ such that $z'\prec_v w'$.
If $w$ and $w'$, and $z$ and $z'$ are incomparable with respect to the strict weak order $\prec_v$, then 
$z \prec_v w$ implies $z' \prec_v w'$, and $\prec_v$ induces a strict linear order on $V\modout_{\sim_v}$.

We use the strict linear order on the equivalence classes of the incomparability relation induced by $\prec_v$
to assign numbers to the equivalence classes, which match their position within the strict linear order.
We assign $0$ to the smallest equivalence class regarding~$\prec_v$.
The largest equivalence class regarding~$\prec_v$ is $[v]_v=\{v\}$.
Let $p_v\colon V\modout_{\sim_v}\to \N$ be this assignment.
Then  $p_v(A_i^v)=i$ for all $i\in[0,k]$.
\begin{lemma}\label{lem:pvwvi}
 For all $i\in[0,n-1]$, $v\in V$ and $w\in D_{i,v}\setminus D_{i+1,v}$, we have $p_v([w]_v)=i$.
\end{lemma}
\begin{proof}
 In order to show Lemma~\ref{lem:pvwvi}, we first prove the following three claims.
\begin{claim}\label{clm:pv1}
 For all $i\in[0,n-1]$, and all vertices $v\in V$ and $w,w'\!\in D_{i,v}\setminus D_{i+1,v}$, it holds that ${p_v([w]_v)=p_v([w']_v)}$.
\end{claim}
\begin{proofofclaim}
Let $i\in[0,n-1]$, $v\in V$ and $w,w'\in D_{i,v}\setminus D_{i+1,v}$.
If the graph $G[D_{i,v}]$ and its complement are connected,
then $M_{v,w}=D_{i,v}$ and $M_{v,w'}=D_{i,v}$ according to Corollary~\ref{cor:characterization-D_iv-conn}. 
If $G[D_{i,v}]$ or its complement are not connected,
then $M_{v,w}=D_{i+1,v}\dcup D_{i+1,w}$ and ${M_{v,w'}=D_{i+1,v}\dcup D_{i+1,w'}}$ by Corollary~\ref{cor:characterization-D_iv-not-conn}.
In both cases, $M_{v,w}$ and $M_{v,w'}$ are incomparable with respect to proper set inclusion, 
and therefore $w$ and $w'$ are incomparable with respect to $\prec_v$.
Consequently, $p_v([w]_v)=p_v([w']_v)$.
\end{proofofclaim}

\begin{claim}\label{clm:pv2}
 For all $i\in[0,n-1]$, $v\in V$ and $w\in D_{i,v}\setminus D_{i+1,v}$, we have $D_{i+1,v}\subset M_{v,w}\subseteq D_{i,v}$.
\end{claim}
\begin{proofofclaim}
Let $i\in[0,n-1]$, $v\in V$ and $w\in D_{i,v}\setminus D_{i+1,v}$. 
Then $D_{i+1,v}\subset D_{i,v}$.
 If the graph $G[D_{i,v}]$ and its complement are connected,
 then $M_{v,w}=D_{i,v}$ according to Corollary~\ref{cor:characterization-D_iv-conn}. 
 If $G[D_{i,v}]$ or its complement are not connected,
then $M_{v,w}=D_{i+1,v}\dcup D_{i+1,w}$ by Corollary~\ref{cor:characterization-D_iv-not-conn}.
Clearly, we have $D_{i+1,v}\subset M_{v,w}\subseteq D_{i,v}$ in both cases.
\end{proofofclaim}

\begin{claim}\label{clm:pv3}
 For all $i\in[0,n-1]$, and all vertices $v\in V$, $w\in D_{i,v}\setminus D_{i+1,v}$ and $u\in D_{i+1,v}$, we have ${p_v([u]_v)> p_v([w]_v)}$.
\end{claim}
\begin{proofofclaim}
 Let $i\in[0,n-1]$, $v\in V$, $w\in D_{i,v}\setminus D_{i+1,v}$ and $u\in D_{i+1,v}$.
 Since $u\in D_{i+1,v}$, we have $M_{v,u}\subseteq D_{i+1,v}$.
 According to Claim~\ref{clm:pv2}, $D_{i+1,v}\subset M_{v,w}$.
 Hence, $M_{v,u}\subset M_{v,w}$.
 It follows that $p_v([u]_v)> p_v([w]_v)$.
\end{proofofclaim}

\noindent
We prove Lemma~\ref{lem:pvwvi} by induction on $i\in[0,n-1]$, that is,
we show that $p_v([w]_v)=i$ for  $v\in V$ and $w\in D_{i,v}\setminus D_{i+1,v}$.
Let $v\in V$.

First of all, let us consider the base case. Suppose $i=0$.
Claim~\ref{clm:pv1} and  Claim~\ref{clm:pv3} imply that the equivalence classes $[w]_v$ are minimal with respect to $\prec_v$
for all $w\in D_{0,v}\setminus D_{1,v}$.
Hence, $p_v([w]_v)=0$ for all $w\in D_{0,v}\setminus D_{1,v}$.

Next, let us consider the inductive case. Suppose $i>0$.
By inductive assumption we have $p_v([z]_v)=i-1$ for all $z\in D_{i-1,v}\setminus D_{i,v}$, and 
$p_v([z']_v)<i-1$ for all $z\in V\setminus D_{i-1,v}$.
For arbitrary vertices $w\in D_{i,v}\setminus D_{i+1,v}$ and $z\in D_{i-1,v}\setminus D_{i,v}$, 
we show that $M_{v,w}\subset M_{v,z}$. 
Then it follows from Claim~\ref{clm:pv1} and  Claim~\ref{clm:pv3} that $p_v([z]_v)=i$ for all $w\in D_{i,v}\setminus D_{i+1,v}$.

Let $w\in D_{i,v}\setminus D_{i+1,v}$ and $z\in D_{i-1,v}\setminus D_{i,v}$. 
By Claim~\ref{clm:pv2}, we have $M_{v,w}\subseteq D_{i,v}$.
If the graph $G[D_{i-1,v}]$ or its complement are not connected,
then $M_{v,z}=D_{i,v}\dcup D_{i,z}$ by Corollary~\ref{cor:characterization-D_iv-not-conn},
and $M_{v,w}\subset M_{v,z}$.
Now suppose the graph $G[D_{i-1,v}]$ and its complement are connected. 
Then $M_{v,z}=D_{i-1,v}$ according to Corollary~\ref{cor:characterization-D_iv-conn}.  
Since $z\in D_{i-1,v}\setminus D_{i,v}$, we have $D_{i,v}\subset D_{i+1,v}$.
Consequently,  $M_{v,w}\subset M_{v,z}$ holds also in this case.
\end{proof}

\smallskip

\noindent
We define \vspace{-1mm}
\[S_{i,v}:= \{v\}\cup\bigcup \{M_{v,w}\mid p_v([w]_v)=i,\ w\in V\}\] 
for all $i\in [0,n]$.
Thus,  $S_{i,v}$ is the union of $\{v\}$ and all modules $M_{v,w}$ where $w$ belongs to the equivalence class at position $i$ regarding $\prec_v$.
If $k+1$ is the number of equivalence classes of~$\sim_v$, then\vspace{-1mm} 
\begin{align*}
 S_{i,v}= 
 \begin{cases}
 \bigcup \{M_{v,w}\mid p_v([w]_v)=i,\ w\in V\} 
  & \text{if $i\leq k$}\\
    \{v\} 
    & \text{if $i\geq k$.}
 \end{cases}\\[-1em]                                                                
\end{align*}

\pagebreak[2]
\begin{lemma}\label{lem:Divpv}
 For all $i\in\{0,\dots,n\}$ and $v\in V$, we have $D_{i,v}=S_{i,v}$.
\end{lemma}
\begin{proof}
Let $i\in[0,n]$ and $v\in V$.
Suppose $|D_{i,v}|=1$. Then $D_{i,v}=\{v\}$.
If $i=n$, then clearly $S_{i,v}=\{v\}$  and $D_{i,v}=S_{i,v}$.
Now assume that $i<n$. 
Lemma~\ref{lem:pvwvi} implies that $w\in D_{i,v}\setminus D_{i+1,v}$ if and only if $p_v([w]_v)=i$ for all $w\in V\setminus\{v\}$.
We have $D_{i,n}=\{v\}$ precisely if $D_{i,v}\setminus D_{i+1,v}=\emptyset$.
Hence, $D_{i,n}=\{v\}$ if and only if there does not exists a vertex $w\in V\setminus\{v\}$ with $p_v([w]_v)=i$.
It follows that $D_{i,v}=S_{i,v}$.

Suppose $|D_{i,v}|>1$. Then $i<n$ and $D_{i,v}\setminus D_{i+1,v}$ is not empty.
If $G[D_{i,v}]$ and its complement are connected,
then for all $w\in D_{i,v}\setminus D_{i+1,v}$ we have $D_{i,v}=M_{v,w}$ (Corollary~\ref{cor:characterization-D_iv-conn}).
If $G[D_{i,v}]$ or its complement are not connected,
then for all vertices $w\in D_{i,v}\setminus D_{i+1,v}$ we have $M_{v,w}=D_{i+1,v}\dcup D_{i+1,w}$ (Corollary~\ref{cor:characterization-D_iv-not-conn}).
Therefore in both cases, we have ${D_{i,v}=\bigcup\{M_{v,w}\mid w\in D_{i,v}\setminus D_{i+1,v}\}}$.
Since $w\in D_{i,v}\setminus D_{i+1,v}$ if and only if $p_v([w]_v)=i$ for all $w\in V\setminus\{v\}$ by Lemma~\ref{lem:pvwvi},
we obtain $D_{i,v}=\bigcup\{M_{v,w}\mid p_v([w]_v)=i, w\in V\setminus\{v\}\}$.
As the vertex $\{v\}$ is contained in $D_{i,v}$, it follows that $D_{i,v}=\{v\}\cup\bigcup\{M_{v,w}\mid p_v([w]_v)=i, w\in V\}$.
Hence, $D_{i,v}=S_{i,v}$.
\end{proof}

\begin{theorem}\label{thm:varphiD}
There is an $\STCC$-for\-mu\-la $\varphi_D(p,x,z)$ such that for all graphs $G$, all $i\in N(G)$ and all vertices $v\in V(G)$
the set $\varphi_D[G,i,v;z]$ is the set $D_{i,v}$ of the modular decomposition of $G$.
\end{theorem}
\begin{proof}
First we need a formula $\varphi_{\textup{ord}}$ 
that assigns to vertices $v,w\in W$  the position $p_v([w]_v)$ of $[w]_v$ within the  
strict linear order of the equivalence classes of the incomparability relation induced by $\prec_v$.
More precisely, $G\models \varphi_{\textup{ord}}[v,w,i]$ if and only if  $p_v([w]_v)=i$, for all $v,w\in V(G)$ and $i\in N(G)$.
Clearly, $\varphi_{\textup{ord}}$ is satisfied for $v,w\in V$ and $i\in N(V)$ 
exactly if $i$ is the number of equivalence classes that are smaller than $[w]_v$ regarding $\prec_v$.
Thus, we need an $\STCC$-for\-mu\-la which counts the number of equivalence classes smaller than $[w]_v$.

Let $\varphi_\prec$ be the formula defined in \eqref{equ:varphiprec}.
Then the formula 
\[\varphi_{\sim_v}(x,y,y'):=\lnot \varphi_\prec(x,y,y')\land \lnot \varphi_\prec(x,y,y')\]
defines the equivalence relation $\sim_v$ for every $v\in V$,
that is, $G\models \varphi_{\sim_v}[v,w,w']$ if and only if  $w\sim_v w'$, for all~vertices ${v,w,w'\!\in V(G)}$.
The construction from the proof of Lemma 2.4.3 in~\cite{laubner11diss} shows how to count definable equivalence classes 
in deterministic transitive closure logic with counting $\DTCC$, 
which is a logic that is contained in $\STCC$.
With the help of the formula $\varphi_\prec$, we can use this construction to count
the equivalence classes smaller than $[w]_v$ regarding $\prec_v$ for all $v,w\in V(G)$.
Thus, the formula $\varphi_{\textup{ord}}$ can be defined in $\STCC$.

Now, we let $\varphi_M$ be the formula from Lemma~\ref{lem:spannedmodule-stc}, 
and we apply Lemma~\ref{lem:Divpv}, that is, we use that $D_{i,v}=S_{i,v}$.
Then it is easy to see that the following formula is as desired:
\begin{minipage}{0.9\textwidth}
\[ \varphi_{D}(p,x,z):=\exists y\big( \varphi_{\textup{ord}}(x,y,p)\land\varphi_{M}(x,y,z)\big)\lor x=z.\]
\end{minipage}\hfill
\begin{minipage}{0.05\textwidth}
\qedhere
\end{minipage}
\end{proof}

\smallskip

\noindent
As $\STCC$-for\-mu\-las can be evaluated in logarithmic space~\cite{rei05}, we obtain the following corollary.
\begin{corollary}\label{cor:moddecompinL}
 There exists a log\-a\-rith\-mic-space deterministic Turing machine that, given a graph $G=(V,E)$, a number $i\leq |V|$ and 
 vertices $v,w\in V\!$, decides whether ${w\in D_{i,v}}$.
\end{corollary}

\noindent
Let the family of subsets $D_{i,v}$ with $i\in [0,n]$, $v\in V$ be the modular decomposition of graph $G=(V,E)$.
The \emph{modular decomposition tree} of $G$ is the directed tree $T=(V_T,E_T)$ with
\begin{align*}
V_T& :=\big\{D_{i,v} \,\big\vert\, i\in [0,n], v\in V\big\},\\
E_T& :=\big\{ (D_{i,v},D_{i+1,v'})\in V_T^2 \,\big\vert\, D_{i+1,v'}\subset D_{i,v} \big\}.
\end{align*}

\begin{corollary}\label{cor:moddecomptreeinL}
 There exists a log\-a\-rith\-mic-space deterministic Turing machine that, given a graph $G=(V,E)$, outputs the modular decomposition tree of $G$.
\end{corollary}

\noindent
A graph is a \emph{cograph} if it can be constructed from isolated vertices by disjoint union and join operations.
We obtain the \emph{join} of two graphs $G$ and $H$ by taking the disjoint union of $G$ and $H$ and 
adding all edges $\{v,w\}$ where $v$ is a vertex of $G$ and $w$ is a vertex of $H$. 

The modular decomposition trees of cographs have a special property:
Each inner node is the disjoint union or the join of its children.
Moreover, only modular decomposition trees of cographs have this property.

\begin{corollary}
 Cograph recognition is in $\LOGSPACE$.
\end{corollary}

\noindent
We obtain the \emph{cotree} of a cograph $G$ by 
coloring each inner node $v$ of the modular decomposition tree of $G$ with $0$ 
if $v$ is the disjoint union of its children, and with $1$ if $v$ is the join of its children.
It is well known that for each cograph the cotree is a canonical tree representation.
In \cite{Lindell:Tree-Canon}, Lindell presented a log\-a\-rith\-mic-space algorithm for tree canonization, which can easily be extended to 
cotrees. Thus, Corollary~\ref{cor:moddecomptreeinL} also implies the following:

\begin{corollary}
 There exists a log\-a\-rith\-mic-space algorithm for cograph canonization.
\end{corollary}

\section{Modular Decomposition Theorem}\label{sec:ModDecompThm}
For suitable graph classes $\CC$ that are closed under induced subgraphs,
the Modular Decomposition Theorem is a tool which can be used to show that $\CC$ admits $\FPC$-de\-fin\-a\-ble canonization.
More precisely, for graph classes $\CC$ that are closed under induced subgraphs,
the Modular Decomposition Theorem states that 
$\CC$ admits $\FPC$-de\-fin\-a\-ble canonization
if the class of $\LO$-col\-ored graphs with prime underlying graphs from $\CC$ 
admits $\FPC$-de\-fin\-a\-ble (parameterized) canonization.
Note that the Modular Decomposition Theorem also holds for reasonable extensions of~$\FPC$ that are closed under parameterized $\FPC$-trans\-duc\-tions.

In this section, we first introduce modular contractions and representations of ordered graphs. 
Then we present the Modular Decomposition Theorem and a proof of it. 
Finally, we show variations of the Modular Decomposition Theorem:
We show that  $\CC$ admits $\FPC$-de\-fin\-a\-ble canonization if the class of prime graphs of $\CC$ admits $\FPC$-de\-fin\-a\-ble orders,
and we 
present an analog of the Modular Decomposition Theorem for poly\-no\-mi\-al-time computable canonization.

\subsection{Modular Contraction}\label{sec:modcontraction}
The modular contraction is the graph that we obtain 
by contracting the maximal proper modules of a graph to vertices.

For a graph $G=(V,E)$, let \emph{$\sim_{G}$} be the equivalence relation on $V$ defined by the partition $\{D_{G}(v)\mid v\in V\}$ 
(see Section~\ref{sec:moddecompdefinition}). Then the equivalence class $v\modout_{\sim_{G}}$ of a vertex $v\in V$ is the module $D_{G}(v)$ of $G$.
We let \emph{$G_{\sim}$} be the graph consisting of the vertex set 
${V\modout_{\sim_{G}}=\{v\modout_{\sim_{G}}\mid v\in V\}}$, 
where two distinct vertices $w\modout_{\sim_{G}}$ and $w'\!\modout_{\sim_{G}}$ are adjacent if and only if $w$ and $w'$ are adjacent in~$G$.
According to Observation~\ref{obs:ModuleEdges}, the edges of $G_{\sim}$ are well-defined.
We call  $G_{\sim}$ the \emph{modular contraction} of~$G$.
Thus, the modular contraction of a graph~$G$~is 
\begin{itemize}
 \item an edgeless graph with as many vertices as there are connected components in $G$ if
$G$ is not connected,
\item a complete graph with as many vertices as there are connected components in $\overline{G}$ if 
$\overline{G}$ is not connected, or
\item if $G$ and $\overline{G}$ are connected and $|V(G)|>1$, a set of vertices, one vertex for each maximal proper module, 
where there is an edge between two vertices exactly if the corresponding modules are (completely) connected with edges; or a single vertex if $|V(G)|=1$.
\end{itemize}\vspace{1mm}
Figure~\ref{fig:modcontraction} depicts the graphs from Figure~\ref{fig:module} and their modular contractions.
\begin{figure}[htbp]\vspace{1mm}
\centering
\scalebox{0.90}{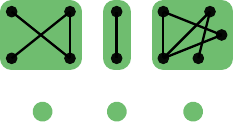\hspace{5.0em}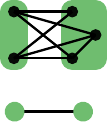\hspace{5.0em}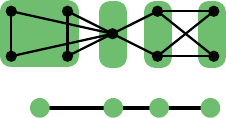
}
\caption{Graphs and their respective modular contractions}
\label{fig:modcontraction}
\end{figure}

\begin{observationC}[{\cite[Satz~1.8]{gallai67transitiv}}]\label{obs:modcontractionprime}
	If $G$ and $\overline{G}$ are connected, then the modular contraction~$G_{\sim}$ of $G$ is prime.
\end{observationC}

\begin{observation}
For every graph $G$, the modular contraction of $G$ is isomorphic to an induced subgraph of $G$.
\end{observation}

\noindent
For all modules $D_{i,v}$ of $G$, we denote  
the modular contraction $G[D_{i,v}]_\sim$ of $G[D_{i,v}]$ for all $i\leq n$ and $v\in V\!$ by \emph{$G_{i,v}$}.
Notice that $G_{0,v}$ is the modular contraction $G_\sim$ of $G$. 

\subsection{Representation of an Ordered Graph}\label{sec:representation}
In the following we introduce the representation of an ordered graph $G$.

Each ordered graph $G=(V,E,\leq)$ is isomorphic to an ordered graph where the vertex set is $[|V|]$ 
and the linear order is the natural order $\leq_{[|V|]}$ on $[|V|]$.
Thus, we suppose the vertex set of our ordered graph $G$ is $[|V|]$ and the linear order of $G$ is $\leq_{[|V|]}$.
We use the representation to encode the ordered graph in a binary relation.
Later, when we want to color the vertices of graphs with ordered graphs, we use these representations as colors instead.
As a result we obtain an $\LO$-col\-ored graph.

Let $G$ be an ordered graph with vertex set $[n]$ and linear order $\leq_{[n]}$.
We encode the ordered graph $G$ in a symmetric binary relation $\rep(G)\subseteq [n]^2$:
\[\rep(G):=\{(l,l')\mid \{l,l'\}\in E(G)\}\cup \{(n,n)\}.\]
We call $\rep(G)$ the \emph{representation} of $G$.
We can reinterpret every representation $R\subseteq N(G)^2$ of an ordered graph as an ordered graph $\graph(R)$.
Let $n'\!\in N(G)$ be the only number with $(n'\!,n')\in R$.
We let
\begin{align*}
	V(\graph(R))&:=[n']\;,\\
	E(\graph(R))&:=\big\{\{l_1,l_2\}\;\big\vert\; (l_1,l_2)\in R \setminus\{(n'\!,n')\}\big\}\text{\quad and}\\
	\leq\!\!(\graph(R))&:=\; \leq_{[n']}.
\end{align*}
We call $\graph(R)$ the \emph{ordered graph} of relation~$R$.
It is easy to see that $\graph(\rep(G))=G$.

\subsection{The Modular Decomposition Theorem}\label{sec:moddecompuntersection}

In the following we present the Modular Decomposition Theorem and a proof of it.

For a class $\CC$ of graphs that is closed under induced subgraphs, 
we let \emph{$\CC^*_{\textup{prim}}$} be the class of all \LO-col\-ored graphs ${H^*\hspace{-2pt}=\hspace{-1pt}(U\hspace{-1pt},V\hspace{-2pt},E,M,\trianglelefteq,L)}$ 
where the underlying graph  $(V,E)$ is a prime graph in $\CC$ and $|V|\geq 4$.

\begin{theorem}[Modular Decomposition Theorem]\label{thm:ModDecompTh}
 Let $\CC$ be a class of graphs that is closed under induced subgraphs.
 If $\CC^*_{\textup{prim}}$ admits $\FPC$-de\-fin\-a\-ble (parameterized) canonization, 
 then $\CC$ admits $\FPC$-de\-fin\-a\-ble canonization. 
\end{theorem}

\noindent
Let $\KKI$  be the class of all $\LO$-col\-ored graphs where the underlying graph is complete or edgeless. 
For a class $\CC$ of graphs that is closed under induced subgraphs, we let
 ${\CC_{\CK\CI}^*:=\CC^*_{\textup{prim}}\cup \KKI}$. 
Notice that $\CC_{\CK\CI}^*$ contains all $\LO$-col\-ored graphs where the underlying graph is a prime graph in $\CC$, 
because every prime graph with less than $4$ vertices is complete or edgeless.
The following two observations show that it suffices to prove the Modular Decomposition Theorem under the assumption that
the class $\CC_{\CK\CI}^*$ admits $\FPC$-de\-fin\-a\-ble canonization. 

Observation~\ref{obs:parametersInCanonizations} is a direct consequence of Lemma~\ref{lem:parametersInCanonizations}. 
\begin{observation}\label{obs:parametersInCanonizations}
If there exists a parameterized $\FPC$-can\-on\-iza\-tion of $\CC^*_{\textup{prim}}$, 
then there exists an $\FPC$-can\-on\-iza\-tion of $\CC^*_{\textup{prim}}$.
\end{observation}

\begin{observation}\label{obs:transductionCKI}
 Let $\CC$ be a class of graphs that is closed under induced subgraphs.
 If $\CC^*_{\textup{prim}}$ admits $\FPC$-de\-fin\-a\-ble canonization, 
 then $\CC_{\CK\CI}^*$ admits $\FPC$-de\-fin\-a\-ble canonization.
\end{observation}
\begin{proof}
 Let $\CC$ be closed under induced subgraphs, and let $\Theta^{\text{c}}$ be an $\FPC$-can\-on\-iza\-tion of $\CC^*_{\textup{prim}}$.
We extend $\Theta^{\text{c}}$ to an $\FPC$-can\-on\-iza\-tion of the class $\CC_{\CK\CI}^*$.

It is easy to describe in $\FPC$ whether the underlying graph $H$ of an \LO-col\-ored graph $H^*$ is complete or edgeless.
Also, it is not hard to define the canon of an $\LO$-col\-ored graph~$H^*\!\in \KKI$ in $\FPC$.
We can use the lexicographical order of the vertices' 
natural colors and the linear order of the basic color elements to define an ordered copy of $H^*$
(see \cite[Example 17]{diss}). 
Thus, we can extend $\Theta^{\text{c}}$ in such a way that it first detects
whether $\LO$-col\-ored graph $H^*$ is in $\KKI$ or not.
If $H^*\!\in\KKI$, then  $\Theta^{\text{c}}$ defines the canon as explained above.
If $H^*\!\not\in\KKI$, then $\Theta^{\text{c}}$ behaves as originally intended.
\end{proof}

\bigskip
\noindent
For the remainder of this section, let $\CC$ be a graph class that is closed under induced subgraphs.
Further, let $\Theta^{\text{c}}$ be an $\FPC$-can\-on\-iza\-tion of the class 
$\CC_{\CK\CI}^*$, 
and let $f^*$ be the canonization mapping defined by  $\Theta^{\text{c}}\!$.
We show that there exists an $\FPC$-can\-on\-iza\-tion of $\CC$.

\subsection*{\textbf{Sketch of the Proof}}
In order to show the Modular Decomposition Theorem
the idea is to construct 
the canon of each $G\in \CC$ recursively using the modular decomposition. 
Let $n$ be the number of vertices of $G$.
Then for all $i\in \{n,\dots,0\}$,
starting with $i=n$,
we inductively define the canons 
of the induced subgraphs $G[D_{i,v}]$
for all $v\in V\!$.
We can trivially define the canon for each module that is a singleton.
For the inductive step 
we consider the modular contraction $G_{i,v}$ of $G[D_{i,v}]$.
For all $i<n$ and $v\in V\hspace{-1pt}$, the graph~$G_{i,v}$ is prime if $G[D_{i,v}]$ and $\overline{G}[D_{i,v}]$ are connected,
complete if  $\overline{G}[D_{i,v}]$ is not connected or
edgeless if $G[D_{i,v}]$ is not connected.
We transform $G_{i,v}$ into an \LO-col\-ored graph~$G^*_{i,v}$ by coloring
every vertex \raisebox{0cm}[\ht\strutbox][6pt]{$w\modout_{{\sim_{G[D_{i,v}]}}}$} of $G_{i,v}$ with 
the representation of the canon of 
the graph $G[D_{i+1,w}]$. 
The canon of $G[D_{i+1,w}]$ is definable by inductive assumption.
Then $G^*_{i,v}\in \CC_{\CK\CI}^*$.
Thus, we can apply  $f^*$ to get $G^*_{i,v}$'s canon $K^*_{i,v}$.
Now each vertex of  $K^*_{i,v}$ stands for a module, 
and the color of every vertex is the representation of the canon of the graph induced by this module.
Therefore, we can use the coloring to replace each vertex of  $K^*_{i,v}$ by the graph induced by the module that the vertex represents.
From the coloring, we also obtain a linear order on each module.
We use the order on the vertices of $K^*_{i,v}$ to extend the linear orders on the modules to a linear order on the vertex set of the resulting~graph.

\bigskip

\noindent
In the following we present a detailed proof of the Modular Decomposition Theorem.

First, we shortly introduce notation that simplifies the construction of an $\FPC$-can\-on\-iza\-tion.  
Then, we start by recursively defining the  canonization mapping $f$ which maps each graph ${G\in \CC}$ to its canon $f(G)$. 
Afterwards we show that this canonization mapping is $\FPC$-de\-fin\-a\-ble.

\subsubsection*{\textbf{Notation}}

Throughout this section, we use 
$x,y,z$ and variants like $x_1,y'\!,z^*$ of these letters for structure variables, 
and $o,p,q,r,s$ and variants for number variables.
There exist $\FOC$-for\-mu\-las $\zero(p)$, $\one(p)$ and $\ord(p)$ that define the numbers $0$,~$1$ and~$|U(A)|$ for all structures $A$, 
and an $\FPC$-for\-mu\-la $\plus(p,q,r)$ that defines the addition function \cite[Example~2.3.5]{groheDC}.
We write $p=0$, $p=1$ and $p+q=r$ instead of $\zero(p)$, $\one(p)$ and $\plus(p,q,r)$, respectively, and use similar abbreviations.
We denote $\lnot u=v$ by $u\not= v$ and $p\leq q\land \lnot p=q$ by $p < q$,
and abbreviate $\exists u_1 \dots \exists u_k$ by $\exists u_1,\dots,u_k$ and 
$\lnot\#o\,\varphi = 0$ by~$\exists o\, \varphi$.

\subsubsection*{\textbf{Canonization Mapping}}\label{app:sec:MDT-canonmapping}
In the following we define the canonization mapping $f$,
which maps each graph ${G\in \CC}$ to the canon $f(G)$.  
We let the vertex set of canon $f(G)$ 
be $[|V(G)|]$. The linear order on the vertex set is the natural order $\leq_{[|V(G)|]}$ on $[|V(G)|]$.

If ${|V(G)|\hspace{-1.5pt}=\hspace{-1.5pt}1}$, then the canon of $G$ is ${f(G)\hspace{-1.5pt}:=\hspace{-1.5pt}(\{1\},\emptyset,\leq_{\{1\}}\hspace{-1pt})}$. 
Now in order to define the canonization mapping $f$ on graphs $G$ with $|V(G)|>1$,  
we use their decomposition into modules to recursively
construct the canon of a graph from the canons of the induced subgraphs on its decomposition modules. 
In a first step we define $G_\sim^*$, the \emph{\LO-col\-ored graph of $G$}, which has $G_\sim$, 
the modular contraction of $G$, as underlying graph.
To obtain $G_\sim^*$ we color every vertex $w\modout_{\sim_G}$ of $G_\sim$ with 
the representation of the canon $f(G[D_G(w)])$ of $G[D_G(w)]$. 
More precisely, we let 
\begin{align*}
G_\sim^*&:=(U_{G_\sim^*},V_{G_\sim^*},E_{G_\sim^*},M_{G_\sim^*},\unlhd_{G_\sim^*},L_{G_\sim^*})\qquad\qquad\
\end{align*}
where\vspace{-1mm}
\begin{align*}
U_{G_\sim^*}&:=V_{G_\sim^*}\dcup M_{G_\sim^*},\\
(V_{G_\sim^*},E_{G_\sim^*})&:=G_\sim,\\
M_{G_\sim^*}&:=[0,|V(G)|],\\
\unlhd_{G_\sim^*}&:=\;\leq_{[0,|V(G)|]}, \text{\quad and}\\
 L_{G_\sim^*}&:= \{(v,i,j)\in V_{G_\sim^*}\times M_{G_\sim^*}^2 \mid (i,j)\in \rep(f(G[D_G(v)]))\}.\\[-1em]
\end{align*}
The construction of $G_{\sim}^*$ is illustrated in Figure~\ref{fig:contraction}.

\begin{figure}[htbp]
\captionsetup[subfigure]{position=b, textfont=normalfont,justification=raggedright, singlelinecheck=on,labelformat=empty}
\centering
\vspace{2mm}

\subcaptionbox{Graph $G$\label{fig:contraction1}}[0.20\textwidth]{\scalebox{1.25}{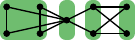}\vspace*{7.7mm}}
\hspace*{0.5mm}\subcaptionbox{Modular contraction $G_{\sim}$ and the maximal proper modules $D_{G}(w)$, $w\in V$\label{fig:contraction2}}[0.26\textwidth]{%
\scalebox{1.25}{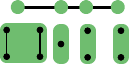}\vspace*{5.7mm}}
\hspace*{2mm}\subcaptionbox{Modular contraction $G_{\sim}$ and the canons $f(G[D_{G}(w)])$, $w\in V$\label{fig:contraction3}}[0.27\textwidth]{%
\scalebox{1.25}{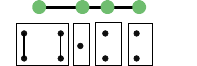}\vspace*{5.7mm}}
\hspace*{1mm}\subcaptionbox{The graph $G_{\sim}^*$\label{fig:contraction4}}[0.22\textwidth]{%
\hspace*{2.5mm}\scalebox{1.25}{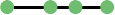}\newline
\begin{minipage}[t]{0.049\textwidth}
 \tiny{$\{\hspace{-0.5pt}(\hspace{-0.5pt}1,\!2\hspace{-0.5pt}),
 \newline \hphantom{\{}(\hspace{-0.5pt}2,\!1\hspace{-0.5pt}),
 \newline \hphantom{\{}(\hspace{-0.5pt}3,\!4\hspace{-0.5pt}), 
 \newline \hphantom{\{}(\hspace{-0.5pt}4,\!3\hspace{-0.5pt}), 
 \newline \hphantom{\{}(\hspace{-0.5pt}4,\!4\hspace{-0.5pt})\hspace{-0.5pt}\}$}
\end{minipage}
\begin{minipage}[t]{0.046\textwidth}
\tiny{$\{\hspace{-0.5pt}(\hspace{-0.5pt}1,\!1\hspace{-0.5pt})\hspace{-0.5pt}\}$}
\end{minipage}
\begin{minipage}[t]{0.046\textwidth}
\tiny{$\{\hspace{-0.5pt}(\hspace{-0.5pt}2,\!2\hspace{-0.5pt})\hspace{-0.5pt}\}$}
\end{minipage}
\begin{minipage}[t]{0.046\textwidth}
\tiny{$\{\hspace{-0.5pt}(\hspace{-0.5pt}2,\!2\hspace{-0.5pt})\hspace{-0.5pt}\}$}
\end{minipage}
}
\caption{Construction of $G_{\sim}^*$}
\label{fig:contraction}
\end{figure}

As $G_\sim$, the underlying graph of $G_\sim^*$, is a modular contraction, 
$G_\sim$ is prime, complete or edgeless. 
Therefore,  we can use the given canonization mapping $f^*$ to obtain the canon of~$G_\sim^*$:\vspace{-1mm}
\begin{align*}
K_\sim^*=(U_{K_\sim^*},V_{K_\sim^*},E_{K_\sim^*},M_{K_\sim^*},\unlhd_{K_\sim^*},L_{K_\sim^*},\leq_{K_\sim^*}).
\end{align*}

To get the canon of $G$, we replace each 
vertex $w\in V_{K_\sim^*}$ of the ordered \LO-col\-ored graph~$K_\sim^*$ by the graph represented 
by $w$'s natural color $L_w^{\N}$.
Since each $\LO$-col\-ored graph consists of a linear order on the basic color elements, 
the natural colors of isomorphic \LO-col\-ored graphs are equal.
Hence, the natural colors of $K_\sim^*$ match the (natural) colors of $G_\sim^*$, 
which again encode the canons of the subgraphs induced by the modules the vertices represent.
Thus, we replace the vertices of $K_\sim^*$ by the corresponding canons. 
We use the linear order on the vertices 
(given by the linear order $\leq_{K_\sim^*}$ restricted to the vertex set $V_{K_\sim^*}$) 
to replace one vertex  after the other.
We name the new vertices consecutively according to the time of their installment (and their order in the respective canon).
Figure~\ref{fig:contractionteil2} shows the construction of $f(G)$.

\begin{figure}[htbp]
\captionsetup[subfigure]{position=b, textfont=normalfont,justification=raggedright, singlelinecheck=on,labelformat=empty}
\centering
\subcaptionbox{A possible canon $K_{\sim}^*$ of graph $G_{\sim}^*$ where 
$w\!<_{K_\sim^*}\!\hspace{0.5pt}x\hspace{0.5pt}\!<_{K_\sim^*}\!\hspace{0.5pt}y\hspace{0.5pt}\!<_{K_\sim^*}\!z$,
$a\,\unlhd_{K_\sim^*}\hspace{-0.55pt}b\hspace{0.15pt}\unlhd_{K_\sim^*}\hspace{-0.85pt}c\hspace{0.15pt}\unlhd_{K_\sim^*}\hspace{-0.5pt}d$
\label{fig:contraction5}}[0.23\textwidth]{%
\hspace*{-1.5mm}\scalebox{1.0}{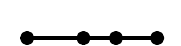}\newline
\hspace*{-2mm}\begin{minipage}[t]{0.049\textwidth}
 \tiny{$\{\hspace{-0.5pt}(\hspace{-0.5pt}a,\!b\hspace{-0.5pt}),
 \newline \hphantom{\{}(\hspace{-0.5pt}b,\!a\hspace{-0.5pt}),
 \newline \hphantom{\{}(\hspace{-0.5pt}c,\!d\hspace{-0.5pt}), 
 \newline \hphantom{\{}(\hspace{-0.5pt}d,\!c\hspace{-0.5pt}), 
 \newline \hphantom{\{}(\hspace{-0.5pt}d,\!d\hspace{-0.5pt})\hspace{-0.5pt}\}$}
\end{minipage}
\begin{minipage}[t]{0.046\textwidth}
\tiny{$\{\hspace{-0.5pt}(\hspace{-0.5pt}a,\!a\hspace{-0.5pt})\hspace{-0.5pt}\}$}
\end{minipage}
\begin{minipage}[t]{0.046\textwidth}
\tiny{$\{\hspace{-0.5pt}(\hspace{-0.5pt}b,\!b\hspace{-0.5pt})\hspace{-0.5pt}\}$}
\end{minipage}
\begin{minipage}[t]{0.046\textwidth}
\tiny{$\{\hspace{-0.5pt}(\hspace{-0.5pt}b,\!b\hspace{-0.5pt})\hspace{-0.5pt}\}$}
\end{minipage}
}
\hspace*{2mm}\subcaptionbox{The underlying graph and the natural colors of  $K_{\sim}^*$\label{fig:contraction6}}[0.22\textwidth]{%
\hspace*{-0.5mm}\scalebox{1.0}{\input{graphics/fig22.pdf_tex}}\newline
\hspace*{-1mm}\begin{minipage}[t]{0.049\textwidth}
 \tiny{$\{\hspace{-0.5pt}(\hspace{-0.5pt}1,\!2\hspace{-0.5pt}),
 \newline \hphantom{\{}(\hspace{-0.5pt}2,\!1\hspace{-0.5pt}),
 \newline \hphantom{\{}(\hspace{-0.5pt}3,\!4\hspace{-0.5pt}), 
 \newline \hphantom{\{}(\hspace{-0.5pt}4,\!3\hspace{-0.5pt}), 
 \newline \hphantom{\{}(\hspace{-0.5pt}4,\!4\hspace{-0.5pt})\hspace{-0.5pt}\}$}
\end{minipage}
\begin{minipage}[t]{0.046\textwidth}
\tiny{$\{\hspace{-0.5pt}(\hspace{-0.5pt}1,\!1\hspace{-0.5pt})\hspace{-0.5pt}\}$}
\end{minipage}
\begin{minipage}[t]{0.046\textwidth}
\tiny{$\{\hspace{-0.5pt}(\hspace{-0.5pt}2,\!2\hspace{-0.5pt})\hspace{-0.5pt}\}$}
\end{minipage}
\begin{minipage}[t]{0.046\textwidth}
\tiny{$\{\hspace{-0.5pt}(\hspace{-0.5pt}2,\!2\hspace{-0.5pt})\hspace{-0.5pt}\}$}
\end{minipage}
}
\hspace*{2mm}\subcaptionbox{The underlying graph and the graphs represented by the natural colors of  $K_{\sim}^*$\label{fig:contraction7}}[0.23\textwidth]{%
\hspace*{-10mm}\scalebox{1.0}{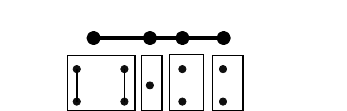}\vspace*{5.0mm}}
\hspace*{1.5mm}\subcaptionbox{Graph $f(G)$\label{fig:contraction8}}[0.23\textwidth]{%
\hspace*{-10mm}\scalebox{1.1}{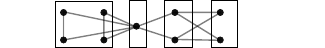}\vspace*{9mm}}
\caption{Construction of $f(G)$}
\label{fig:contractionteil2}
\end{figure}

In the following we describe the construction of the canon $f(G)$ more precisely.
For all vertices $w\in V_{K_\sim^*}$, let $L_w^{\N}$ be the natural color of~$w$, and 
let $n_w$ be  the only element
with $({n_w},{n_w})\in L_w^{\N}$. 
Since the module that $w$ stands for consists of at least one vertex,  such 
an $n_w$ exists and $0< n_w\leq |M_{K_\sim^*}|$. 
\begin{samepage}%
To construct the canon we assign each vertex $n$ of the
ordered graph $\graph(L_w^{\N})$ of representation $L_w^{\N}$ to the number
\begin{align}
	\numb(w,n):= n\ +\!\!\!\!\sum_{\genfrac{}{}{0pt}{}{\ w'<_{K_\sim^*}w,}{\ w'\,\in\,\, V_{K_\sim^*}}}\!\!\!\! n_{w'},\label{eqn:nbmn}
\end{align}
where $w'<_{K_\sim^*}w$ if and only if $w'\leq_{K_\sim^*}w$ and $w'\not=w$.%
\end{samepage}
Clearly, the mapping $\numb$ is a bijection,
that maps $(n,v)$, where $n$ is a vertex in the graph represented 
by vertex $v$'s natural color, to $m\in[|V(G)|]$.

We add a pair of numbers to the edges of $f(G)$ 
if they represent vertices from different modules, and the modules are completely connected; or
they represent vertices from the same module that are connected by an edge.
Thus, we add $\{m_1,m_2\}$ to the edges of $f(G)$ if
\begin{enumerate}
\item\label{enum:newedges1} there exist an edge ${\{w_1,w_2\}\in E_{K_\sim^*}}$ and numbers ${n_1,n_2\in [|M_{K_\sim^*}|]}$ such that
${n_1\leq n_{w_1}}$, ${n_2\leq n_{w_2}}$ and  $(m_1,m_2)=\big(\numb(w_1,n_1),\numb(w_2,n_2)\big)$, or
\item\label{enum:newedges2} there exist a vertex ${w\in V_{K_\sim^*}}$ 
and a pair ${(n_1,n_2)\in L_w^{\N}}$ such that ${n_1\not =n_2}$ and 
 $(m_1,m_2)=\big(\numb(w,n_1),\numb(w,n_2)\big)$.
\end{enumerate}

\noindent
Clearly, the ordered graph $f(G)$ is an ordered copy of~$G$ on the number sort.
Observation~\ref{obs:fcanonization} shows that 
$f$ maps isomorphic graphs from $\CC$ to the same ordered graph.
Hence, $f$ is a canonization mapping. 
Note that Observation~\ref{obs:fcanonization} follows also directly from $f$ being $\FPC$-de\-fin\-a\-ble, 
which is shown in the remainder of Section~\ref{app:sec:MDT-canonmapping}.

\begin{observation}\label{obs:fcanonization}
 For all $G,G'\in \CC$, we have $f(G)=f(G')$ if $G\cong G'\!$.
\end{observation}
\begin{proof}
 Let $h$ be an isomorphism between $G$ and $G'\!$. 
 We show that  ${f(G)=f(G')}$ by induction.
 Clearly, this is the case if $G$ and $G'$ consist of only one vertex.
 Therefore, suppose ${|V(G)|=|V(G')|>1}$.
 As the modular decomposition of a graph is unique, 
the isomorphism~$h$ maps every decomposition module of $G$ to a decomposition module of~$G'$.
Hence, the respective graphs  induced by the decomposition modules of $G$ and $G'$ are isomorphic,
and by inductive assumption $f$ maps them to the same canon.
Further, $h$ induces an isomorphism $h_{\sim}$ between  $G_{\sim}$ and $G'_{\sim}$.
Consequently, the graphs $G_{\sim}^*$ and  $G_{\sim}'^{\,*}$ are isomorphic.
They are mapped  to isomorphic copies   $K_{\sim}^*$ and  $K_{\sim}'^{\,*}$ by $f^*\!$.
Let $g$ be an isomorphism between them.
Clearly, for each vertex $v\in V_{K_{\sim}^*}$, the vertices $v$ and $g(v)$ have the same natural color.
Further, we have $v_1\leq_{K_{\sim}^*} v_2$ if and only if $g(v_1)\leq_{K_{\sim}'^{\,*}} g(v_2)$.
As a consequence, $f(G)=f(G')$. 
\end{proof}

\subsubsection*{\textbf{Defining the Canonization Mapping in {\normalfont$\FPC$}}}
We show that the canonization mapping $f$ is $\FPC$-de\-fin\-a\-ble in the following five steps.

\subsubsection*{\textbf{Step 1: Counting Transduction from Graphs to \LO-Colored Graphs}}

For all modules   $D_{i,v}$ of $G$ with $i\leq V(G)$ and $v\in V(G)$, we denote the \LO-col\-ored graph
$(G[D_{i,v}])_\sim^*$ of $G[D_{i,v}]$ by $G_{i,v}^*$. Notice that the underlying graph of $G_{i,v}^*$ is $G_{i,v}$.

The first step in constructing an $\FPC$-trans\-duc\-tion that defines $f$ is to define
the \LO-col\-ored graph $G^*_{i,v}$ for all ${G\in \CC}$, ${i\in N(G)}$\vspace{-1.5pt} and ${v\in V(G)}$.
For this purpose, we define a parameterized counting transduction
$\Theta^{\raute}(o,z,X)$,\vspace{0.5pt} 
where $o$ is a number variable, $z$ is a structure variable, 
and $X$ is a relational variable of arity~$4$ that ranges over relations $R\subseteq N(G)\times V(G)\times N(G)^2$.
It is a parameterized $\FPC[\{E\},\{V,E,M,\trianglelefteq,L\}]\vspace{1pt}$-count\-ing transduction, 
which maps every graph~$G$ to 
an \LO-col\-ored graph $G_{i,v}^R:=\Theta^{\raute}[G,i,v,R]$ for $(G,i,v,R)\in \Dom(\Theta^{\raute}(o,z,X))$.
For some triples $(i,v,R)\in G^{(o,z,X)}$\vspace{0.5pt} where $R$ is a specific relation depending on $i$ and $v$,
the \LO-col\-ored graph $G_{i,v}^R$ is isomorphic to $G_{i,v}^*$. 
We let
\begin{align*}
\Theta^{\raute}(o,z,X)=\big(& \theta_{\dom}(o,z,X), \theta_{U}(o,z,X,y), \theta_{\approx}(o,z,X,y,y'), \\
&\theta_{V}(o,z,X,y),\theta_{E}(o,z,X,y,y'), \\
&\theta_{M}(o,z,X,p),
\theta_{\trianglelefteq}(o,z,X,p,p'),\theta_{L}(o,z,X,y,p,p')\big),
\end{align*}
where
\begin{align*}
	\theta_{\dom}(o,z,X)&:=\ \lnot\ord(o),\\
	\theta_{U}(o,z,X,y)&:=\ \varphi_D(o,z,y),\\
	\theta_{\approx}(o,z,X,y,y')&:=\	\exists o'\big(\text{$o+1=o'$}\land\varphi_D(o'\!,y,y')\big),\\
	\theta_{V}(o,z,X,y)&:=\ \true,\\
	\theta_{E}(o,z,X,y,y')&:=\  E(y,y'),\\
	\theta_{M}(o,z,X,p)&:=\ \true,\\
	\theta_{\trianglelefteq}(o,z,X,p,p')&:=\  p\leq p'\quad \text{ and}\\
	\theta_{L}(o,z,X,y,p,p')&:= \  \exists o'\big(\text{$o+1=o'$}\land X(o',y,p,p')\big).
\end{align*}
As a reminder, the formula $\varphi_D(o,z,y)$, which was introduced in Theorem~\ref{thm:varphiD}, 
defines the set $D_{i,v}$ of the modular decomposition, i.e., 
for all $i\in N(G)$ and all vertices $v\in V(G)$
we have $\varphi_D[G,i,v;y]=D_{i,v}$.	

Let $G\in \CC$.
We say a triple $(i,v,R)\in G^{(o,z,X)}$ is \emph{suitable} for $G$ if it satisfies  $i<|V(G)|$
and  the following property:
For all  $w\in D_{i,v}$ 
the relation  
\[R_{i+1,w}:=\{(n_1,n_2)\mid (i+1,w,n_1,n_2\in R)\}\] 
is the representation of the canon of 
$G[D_{i+1,w}]$ defined by $f$.
We let \emph{$\suit(G)$} be the set of all suitable triples for $G$.

\begin{lemma}\label{lem:GivRisomGiv*}
Let $G\in\CC$ and let ${(i,v,R)\in G^{(o,z,X)}}$ be a suitable triple for graph $G$.
Then ${(G,i,v,R)\in \Dom(\Theta^{\raute}(o,z,X))}$ and 
$G_{i,v}^R=G_{i,v}^*$.
\end{lemma}

\begin{proof}
Let ${G\in\CC}$. Further, let ${i\in N(G)}$, ${v\in V(G)}$ and  ${R\subseteq N(G)\times V(G)\times N(G)^2}$ be
such that  ${(i,v,R)}\in\suit(G)$. Then, ${i<|V(G)|}$.  Therefore, 
${(G,i,v,R)\in \Dom(\Theta^{\raute}(o,z,X))}$.
Clearly, $\theta_{U}[G,i,v,R;y]$ is the set $D_{i,v}$.
Further, $\theta_{\approx}[G,i,v,R;y,y']$ is the equivalence relation
${\{D_{i+1,w}\mid w\in V(G)\}}$. Let $\approx$ denote this equivalence relation.
Then the universe of $G_{i,v}^R$ is the set $D_{i,v}\modout_\approx\,\dcup\,\, [0,|V(G)|]$.
The vertex set  $V(G_{i,v}^R)$ is  $D_{i,v}\modout_\approx$, and it is not hard to see that
the formulas $\theta_{V}$, $\theta_{\approx}$ and $\theta_{E}$ of transduction $\Theta^{\raute}(o,z,X)$ 
define the graph $G_{i,v}$.
Further,  ${M(G_{i,v}^R)=[0,|V(G)|]}$ and $\trianglelefteq \!\!(G_{i,v}^R)$\vspace{0.5pt} is the natural order on $[0,|V(G)|]$.
Finally, the formula $\theta_{L}$ defines the color relation.
As ${(i,v,R)\in \suit(G)}$,\vspace{0.5pt} the relation
${\{(m_1,m_2)\mid (i+1,w,m_1,m_2)\in R\}}$ is the representation \vspace{-0.5pt} of the canon of $G[D_{i+1,w}]$
for all $w\in V(G)$, and we obtain that $G_{i,v}^R$, that is, $\Theta^{\raute}[G,i,v,R]$, is equal to  $G_{i,v}^*$
for all $(i,v,R)\in \suit(G)$.
\end{proof}

\noindent
Later, we will make sure that the triple of parameters $(o,z,X)$ 
is always interpreted by a suitable triple.

\subsubsection*{\textbf{Step 2: Transduction from Graphs to ordered \LO-Colored Graphs}}

$\Theta^{\raute}(o,z,X)$ is a parameterized $\FPC$-count\-ing transduction. 
Thus, there exists a 
parameterized $\FPC$-trans\-duc\-tion $\Theta^*(o,z,X)$
with the same domain, such that $\Theta^{\raute}[G,i,v,R]$  and $\Theta^*[G,i,v,R]$
are isomorphic for all $(G,i,v,R)$ in the domain (Proposition~\ref{thm:CountingTransduction}).
As a consequence, Lemma~\ref{lem:GivRisomGiv*} holds  for $\FPC$-trans\-duc\-tion $\Theta^*(o,z,X)$ in a similar way:
For a graph $G\in \CC$ and a suitable triple $(i,v,R)\in G^{(o,z,X)}$ the tuple
$(G,i,v,R)$ is in the domain of $\Theta^*(o,z,X)$, and the $\LO$-col\-ored graph $\Theta^*[G,i,v,R]$ is isomorphic to $G_{i,v}^*$.

Let $G\in \CC$ and let  $(i,v,R)$ be a suitable triple  for~$G$.
Then $\Theta^*[G,i,v,R]$, as it is isomorphic to $G_{i,v}^*$, is an $\LO$-col\-ored graph in $\CC^*_{\CK\CI}$.
Transduction $\Theta^{\text{c}}$ is an $\FPC$-can\-on\-iza\-tion for the class $\CC^*_{\CK\CI}$\vspace{0.5pt} of $\LO$-col\-ored graphs.
According to Proposition~\ref{prop:composition} we can compose $\Theta^*(o,z,X)$ and $\Theta^{\text{c}}\!$.
We obtain a parameterized $\FPC[\{E\},\{V,E,M,\trianglelefteq,L,\leq\}]$-trans\-duc\-tion $\Theta^{*\text{c}}(o,z,X)$
where $(G,i,v,R)\in \Dom(\Theta^{*\text{c}}(o,z,X))$ for all $G\in\CC$ and $(i,v,R)\in \suit(G)$.

As $\Theta^*[G,i,v,R]$ and $G_{i,v}^*$ are isomorphic
for $G\in \CC$ and  suitable triples $(i,v,R)$~for~$G$, and $\Theta^{\text{c}}$ is a canonization,
the ordered $\LO$-col\-ored graph $\Theta^{\text{c}}[\Theta^*[G,i,v,R]]$ is an ordered copy of $G_{i,v}^*$.
Further, for all $(G,i,v,R)\in \Dom(\Theta^{*\text{c}}[o,z,X])$ 
the ordered $\LO$-col\-ored graphs $\Theta^{*\text{c}}[G,i,v,R]$ and $\Theta^{\text{c}}[\Theta^*[G,i,v,R]]$ are isomorphic.
Thus, $\Theta^{*\text{c}}[G,i,v,R]$ also is an ordered copy of $G_{i,v}^*$ for $G\in \CC$ and  suitable triples $(i,v,R)$~for~$G$.
We denote the ordered copy $\Theta^{*\text{c}}[G,i,v,R]$ of $G_{i,v}^*$ by \emph{$K_{i,v}^*$}.

The relations between the different parameterized transductions used in Step~2 are illustrated in Figure~\ref{fig:overview-transductions}.
\begin{figure}[htbp]
    \begin{tikzpicture}[scale=1]
        \tikzstyle{vertex}=[anchor=west, draw,inner sep=3pt,align=left]
        
        \begin{scope}
            \node[draw,text width=3cm,align=center] (1) at (0,0) {$G\in\CC$ $(i,v,R)\in\suit(G)$};
            \node[vertex] (2a) at (4.4,2.45) {$G_{i,v}^*$};
            \node[anchor=west,inner sep=1.5pt,align=left] at (4.4,1.85) {\hspace*{2mm} \rotatebox{90}{$=$}};
            \node[vertex] (2b) at (4.4,1.25) {$G_{i,v}^R:=\Theta^{\text{\begin{tiny}\#\end{tiny}}}[G,i,v,R]$};
            \node[anchor=west,inner sep=1.5pt,align=left] at (4.4,0.6) {\hspace*{1mm} \rotatebox{90}{$\cong$}};
            \node[vertex] (2c) at (4.4,0) {$\Theta^{*}[G,i,v,R]$};
            \node[vertex] (3a) at (8.7,0) {$\Theta^{\text{c}}[\Theta^*[G,i,v,R]]$};
            \node[anchor=west,inner sep=1.5pt,align=left] at (8.7,-0.7) {\hspace*{4mm} \rotatebox{90}{$\cong$}};
            \node[vertex] (3b) at (8.7,-1.4) {$K_{i,v}^*:=\Theta^{*\text{c}}[G,i,v,R]$};
            \draw[->] (1) to node[above]{$\Theta^{*}(o,z,X)$} (2c);
            \draw[->,out=23,in=180] (1) to node[above]{\vspace*{5mm}\rotatebox{5}{$\Theta^{\text{\begin{tiny}\#\end{tiny}}}(o,z,X)$}} (2b) ;
            \draw[->,out=-23,in=180] (1) to node[above]{$\Theta^{*\text{c}}(o,z,X)$} (3b);
            \draw[->] (2c) to node[above]{$\Theta^{\text{c}}$}  (3a);
        \end{scope}
    \end{tikzpicture}
\caption{Overview of the parameterized transductions in Step 2}
\label{fig:overview-transductions}
\end{figure}
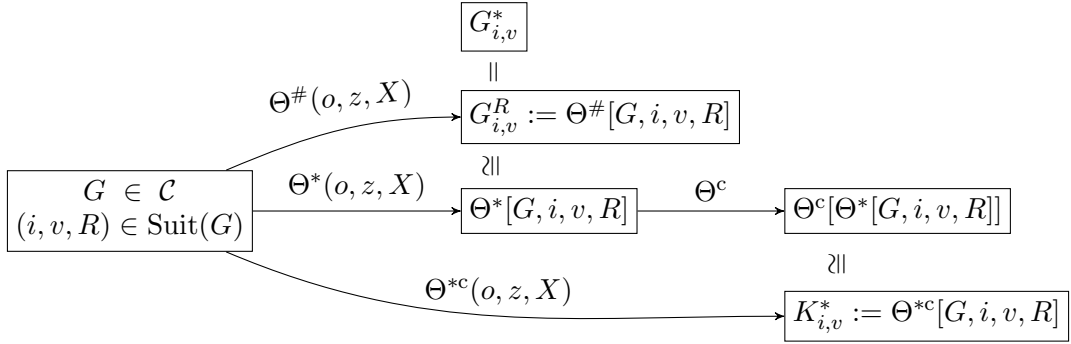

\subsubsection*{\textbf{Step 3: Defining the Edge Relation of the Canon $\bm{f(G[D_{i,v}])}$}}

In the following we construct an $\FPC[\{V,E,M,\trianglelefteq,L,\leq\}]$-for\-mula 
that given an ordered $\LO$-col\-ored graph $K_{i,v}^*$ defines the edge relation of $f(G[D_{i,v}])$.

In order to do this, we have to define the function $\numb(w,n)$ from \eqref{eqn:nbmn} in $\FPC$.
For every vertex $w$ of $K_{i,v}^*$ and each vertex $n$
occurring in the ordered graph of the natural color of the vertex $w$,
$\numb(w,n)$ is the number that the vertex $n$ is assigned to 
in the canon of $G[D_{i,v}]$.
The function $\numb(w,n)$ depends on the values $n_{w'}$ for certain vertices $w'$.
For a vertex $w$, the value $n_{w}$ is the number of vertices in the graph represented by the natural color of the vertex $w$.
We can determine this number 
by finding the only vertex $u$ for which $(u,u)$ belongs to the color of $w$.
Then $n_w$ is the number of vertices that are smaller than $u$ with respect to 
the linear order $\trianglelefteq\!\!(K_{i,v}^*)$ of the basic color elements.
We define $n_w$ via formula~$\varphi_{n_w}$:
\begin{align*}
\varphi_{n_w}(x,p):=\exists y\Big( L(x,y,y)\;\land\; \#y'\big(\! \trianglelefteq\!\!(y'\!,y)\land y'\!\hspace{-1pt}\not=\hspace{-1pt}y \big)=p\Big).
\end{align*}
Then we have $K_{i,v}^*\models\varphi_{n_w}[w,n_w]$
if, and only if, $\graph(L_w^\N)$ has $n_w$ vertices, 
where $L_w^\N$ is the natural color of $w$ in $K_{i,v}^*$.
Notice that the formula $\varphi_{n_w}$ cannot be satisfied 
if $w$ is a basic color element.

In order to define function $\numb(w,n)$, we first check whether ${n\in [n_w]}$.
Then we count the vertices $n'$ in the graph of the natural color of $w$ with $0<n'\leq n$,
and
the vertices occurring in the graphs of the natural colors of all vertices $w'$ 
that are smaller than $w$ with respect to the linear order $\leq\!\!(K_{i,v}^*)$.
Thus, we let\footnote{
Note that we have $|M(K_{i,v}^*)|>V(G)$, and thus, ${|N(K_{i,v}^*)|>V(G)}$.
Consequently,  ${\numb(w,n)\in N(K_{i,v}^*)}$ and a single number variable can represent $\numb(w,n)$.
} 
\begin{align*}
\varphi_{\numb}(x,r,s):=&\
\exists p\ \big(\varphi_{n_w}(x,p)\land \text{``}0<r\leq p\text{''}\big)\,\land\\
& 
\#(x'\!,r')\, \Big(\big(x'\!=x \land \text{``}0<r'\!\leq r\text{''}\big)\, \lor\,   \\
&
\hspace*{1.8cm}\exists p'\, \big(\varphi_{n_w}(x'\!,p')\,\land 
\leq\!\!(x'\!,x)\land x'\!\not=x \land \text{``}0<r'\!\leq p'\text{''} \big)\Big)=s.
\end{align*}
Then $K_{i,v}^*\models \varphi_{\numb}[w,n,m]$ if and only if 
$w$ is a vertex, ${n\in [n_w]}$ and $\numb(w,n)=m$ in $K_{i,v}^*$.

With the formula $\varphi_{\numb}$ we are able to define the edge relation of the canon of $G[D_{i,v}]$. We let
\begin{align*}
\varphi_{E}(s_1,s_2):=&\ \varphi_{E,\ref{enum:newedges1}}(s_1,s_2)\lor \varphi_{E,\ref{enum:newedges2}}(s_1,s_2) 
\end{align*}
where
\begin{align*}
\varphi_{E,\ref{enum:newedges1}}(s_1,s_2):=
&\ \exists x_1,x_2,r_1,r_2 \Big( E(x_1,x_2) \land \!\!\!\bigwedge_{j\in\{1,2\}}\!\!\! \varphi_{\numb}(x_j,r_j,s_j)\Big), \\
\varphi_{E,\ref{enum:newedges2}}(s_1,s_2):=
&\ \exists x,y_1,y_2,r_1,r_2 \Big(L(x,y_1,y_2)\land r_1\not= r_2\, \land \\
&\ \hspace{7em}\bigwedge_{j\in\{1,2\}}\!\!\! \Big( r_j=\#y\big(\!\!\trianglelefteq\!\!(y,y_j)\land y\not=y_j\big)\land
\varphi_{\numb}(x,r_j,s_j)\Big)\Big) 
\end{align*}
It is not hard to see that 
$\varphi_{E,\ref{enum:newedges1}}[K_{i,v}^*;s_1,s_2]$ and 
$\varphi_{E,\ref{enum:newedges2}}[K_{i,v}^*;s_1,s_2]$ are exactly the edges of the canon of 
$G[D_{i,v}]$ obtained by rule~\ref{enum:newedges1} and rule~\ref{enum:newedges2} 
from page~\pageref{enum:newedges1}.

\subsubsection*{\textbf{Step 4: Pulling Back the Formula for the Edge Relation}}

The formula $\varphi_E(s_1,s_2)$ is an 
$\FPC[{\{V,E,M,\trianglelefteq,L,\leq\}}]$-for\-mu\-la that defines
the edge relation of the canon $f(G[D_{i,v}])$ for a given ordered $\LO$-col\-ored graph $K_{i,v}$.
To construct an equivalent $\FPC[\{E\}]$-for\-mu\-la for the $\FPC$-can\-on\-iza\-tion of the class $\CC$, 
we pull back  $\varphi_E$  under the parameterized ${\FPC[\{E\},\{V,E,M,\trianglelefteq,L,\leq\}]}$-trans\-duc\-tion $\Theta^{*\text{c}}(o,z,X)$.
Hence, we apply the Transduction Lemma (Proposition~\ref{prop:transduction-lemma})
to the formula $\varphi_E(s_1,s_2)$.
We obtain the $\FPC[\{E\}]$-for\-mu\-la
$\varphi_E^{-\Theta^{*\text{c}}}(o,z,X,\bar{q}_1,\bar{q}_2)$.
Let $G\in\CC$ and let $(i,v,R)$ be a suitable triple for $G$.
Then $(G,i,v,R)\in \Dom(\Theta^{*\text{c}}(o,z,X))$. 
Thus, for all
tuples $\bar{m}_1\in G^{\bar{q}_1}$ and $\bar{m}_2\in G^{\bar{q}_2}$, we have
\begin{align}\label{eqn:pullbackEdgeRelation}
	G\models \varphi_E^{-\Theta^{*\text{c}}} [i,v,R,\bar{m}_1,\bar{m}_2] 
	\quad\Longleftrightarrow&\quad
	\langle\bar{m}_1\rangle_G,\langle\bar{m}_2\rangle_G\in N(K_{i,v}^*)\text{\ \ and }\\ \nonumber
	&\quad  K_{i,v}^*\models  \varphi_E[\langle\bar{m}_1\rangle_G,\langle\bar{m}_2\rangle_G].
\end{align}

The length of tuples $\bar{q}_1,\bar{q}_2$, and therefore also of $\bar{m}_1,\bar{m}_2$, 
is the same and depends on the length of the tuple of domain variables of the canonization $\Theta^{\text{c}}$.
Let $\ell$ be the length of the listed tuples. Let $\bar{m}_1=(m_1^1,\dots,m_1^\ell)$ and let the other tuples be defined analogously.
In the following we show that in each tuple of variables we only need the first number variable, as the others are always assigned to $0$.

\pagebreak[1]
Again, let $G$ be a graph in $\CC$ and let $(i,v,R)$ be a suitable triple for $G$.
Since the vertex set of $f(G)$ is $[|V(G)|]$,  
we have $\langle\bar{m}_1\rangle_G,\langle\bar{m}_2\rangle_G\in [|V(G)|]$
for all $\bar{m}_1,\bar{m}_2\in N(G)^\ell$ with 
$K_{i,v}^*\models \varphi_E[\langle\bar{m}_1\rangle_G,\langle\bar{m}_2\rangle_G]$.
Now remember that for a tuple $\tup{n}=(n_1,\dots,n_\ell)\in N(G)^\ell\!$,
\[\num[G]{\tup{n}}=\sum_{i=1}^\ell n_i\cdot (|V(G)|+1)^{i-1}.\]
Consequently, we have $m_1^j=0$ and $m_2^j=0$ for all $j>1$, which means that
$m_1^1=\langle\bar{m}_1\rangle_G$ and $m_2^1=\langle\bar{m}_2\rangle_G$.

We define $\phi_E$ as follows:
\begin{align*}
\phi_E(o,z,X,q_1,q_2):= \varphi_E^{-\Theta^{*\text{c}}}\big(o,z,X,(q_1,0,\dots,0),(q_2,0,\dots,0)\big).
\end{align*}
Then, for $G\in \CC$,\, $(i,v,R)\in \suit(G)$ and ${m_1,m_2\in N(G)}$ we have
\begin{align*}
G \models \phi_E[i,v,R,m_1,m_2] \iff\  &
 \text{Vertices $m_1$ and $m_2$ are adjacent in $f(G[D_{i,v}])$.} 
\end{align*}

\subsubsection*{\textbf{Step 5: Inductive Definition of the Canon $\bm{f(G)}$}}

We are now able to inductively define the edge relation of the canon $f(G)$ of $G\in \CC$. We let
\begin{align*}
\phi_K(s_1,s_2):=\exists o'\!,z'\, \Big(&o'=0\land\, s_1\not=s_2\, \land\,
 \ifpx{X}{(o,z,q_1,q_2)}{\phi}(o'\!,z'\!,s_1,s_2)\,\Big)
\end{align*}
where 
\begin{align*}
\phi:= \phi_1\lor \big(\phi_2 \land 
(\phi_{E}\lor\phi_{n_w})\big)
\end{align*}
and 
\begin{align*}
\phi_1(o,z,q_1,q_2):=&\  \ord(o)\,\land\, q_1=1\land q_2=1,\\
 \phi_2(o,z,X,q_1,q_2):=&\   \lnot\ord(o)\,   \land\,
 \exists o'\!,z'\!,q_1',q_2'\big(o+1=o'\land X(o'\!,z'\!,q_1',q_2')\big),\\
   \phi_{n_w}(o,z,q_1,q_2):=&\  q_1=q_2\, \land\, q_1=\nb y\ \varphi_D(o,z,y).
\end{align*}

\noindent
The relational variable $X$ within the inflationary fixed-point operator of the $\FPC$-formula $\phi_K$ 
is of arity~$4$ and ranges over relations $R\subseteq N(G)\times V(G)\times N(G)^2$.
Let $X^\infty$ be the relation assigned to the variable $X$ after the fixed-point is reached.
We show in Lemma~\ref{lem:XkgleichSk} that for each $i\in N(G)$ and $v\in V(G)$ the set of pairs $\{(n_1,n_2)\mid (i,v,n_1,n_2)\in X^\infty\}$ 
is the representation of the canon $f(G[D_{i,v}])$.
For $i=0$ and any vertex $v\in V(G)$ we have $D_{i,v}=V(G)$.
Therefore, Lemma~\ref{lem:XkgleichSk} implies the following corollary.
\begin{corollary}\label{cor:edgesCanon}
 For all $G\in \CC$ and all $n_1,n_2\in N(G)$,
\begin{align*}
 G\models\phi_K[n_1,n_2] \iff \text{$\{n_1,n_2\}$ is an edge of the canon $f(G)$ of $G$.}
\end{align*}
\end{corollary}

\noindent
The formula $\phi$ of the inflationary fixed-point operator is constructed such that
$\phi_1$ defines the basis of the inductive definition.
For $i=|V(G)|$ and all vertices $v\in V(G)$, it ensures that the tuples describing 
the representation of the canon of $G[D_{i,v}]$ are added to the fixed-point relation in the first step.
Thus, all tuples in $\{(|V(G)|,v,1,1)\mid v\in V(G)\}$ are added in the first step.
The formulas $\phi_2$ and $\phi_E\lor \phi_{n_w}$ take effect in the inductive step.
In step $k$ we add all tuples $(i,v,n_1,n_2)\in X^\infty$ to the fixed-point relation with $i=|V(G)|-k+1$.
The formula $\phi_2$ ensures that we add only tuples $(i,v,n_1,n_2)$ if the tuples for 
$i+1$ have already been included to the fixed-point relation.
This way, $i$, $v$ and the fixed-point relation form a suitable triple.
Then, $\phi_E\lor \phi_{n_w}$ defines 
the representation of the canon of  $G[D_{i,v}]$.

In the following lemma 
we show inductively that the formula $\phi_K$ uses an inflationary fixed-point operator which in stage $k$ of its iteration 
defines the representations of the canons of all $G[D_{i,v}]$ with $v\in V(G)$ and $i\geq |V(G)|-k+1$. 

\begin{lemma}\label{lem:XkgleichSk}
	Let $X^k$ be the fixed-point relation that we get at stage $k$ of the iteration of 
	the inflationary fixed-point operator	
	in the formula $\phi_K$.
	Further, let $S^k$ be the set of all tuples $(i,v,n_1,n_2)\in  N(G)\times V(G)\times N(G)^2$
	where $i\geq |V(G)|-k+1$ and $(n_1,n_2)$ is in $\rep(f(G[D_{i,v}]))$, the representation of the canon of $G[D_{i,v}]$.
	Then $X^k=S^k$.
\end{lemma}

\begin{proof}
Of course, for $k=0$ we have $X^k=\emptyset$ and $S^k=\emptyset$.
For $k=1$, it is easy to see 
that there exists no tuple that satisfies the formula $\phi_2$ since   $X^0=\emptyset$.
Consequently,  $X^1$ is the set ${\phi_1[G;o,z,q_1,q_2]=\{(|V(G)|,v,1,1)\mid v\in V(G)\}}$.
Further, for all $v\in V(G)$ the representation of the canon of $G[D_{|V(G)|,v}]$ is
$\{(1,1)\}$, and therefore, $X^1=S^1\!$.
Now assume  $k\geq 1$, and let $X^k=S^k\!$. 
In the following we prove that $X^{k+1}=S^{k+1}$ 
by showing that~${X^{k+1}_j=S^{k+1}_j}$ for all $j\in N(G)$, where 
$S^{k+1}_j$ is the set of all tuples $(j,v,n_1,n_2)\in S^{k+1}$
and $X^{k+1}_j$ is the set~of tuples  $(j,v,n_1,n_2)\in X^{k+1}\!$.

It is easy to see that $X^{k+1}_{j}=S^{k+1}_{j}$ for $j=|V(G)|$:
We have already shown that $\phi_1[G;o,z,q_1,q_2]=X^1$ and that $X_1=S_1$.
Further, the relation $\phi_2[G,\alpha[X^k/X];o,z,q_1,q_2]$ cannot contain any tuples $(i,v,n_1,n_2)$ with $i=|V(G)|$.
Therefore,  ${X^{k+1}_{j}=S_1}$. Since $S_1=S^{k+1}_{j}$  for $j=|V(G)|$, we have $X^{k+1}_{j}=S^{k+1}_{j}\!$.

Next, let us consider ${j<|V(G)|-k}$.
Then ${j< |V(G)|}$, and there does not exist a tuple ${(i,v,n_1,n_2)\in \phi_1[G;o,z,q_1,q_2]}$ with $i=j$.
Further, by inductive assumption we have $X^k=S^k$, and by definition we know that the set $S^k$ does not contain any tuples $(j'\!,v,n_1,n_2)$ 
with ${j'< |V(G)|-k+1}$. Consequently, there cannot be a tuple $(i,v,n_1,n_2)$ in $\phi_2[G,\alpha[X^k/X];o,z,q_1,q_2]$ with $i=j$.
Thus, for ${j<|V(G)|-k}$ we have $X^{k+1}_{j}=\emptyset$, and since $S^{k+1}_{j}$ is also empty,  we obtain $X^{k+1}_{j}=S^{k+1}_{j}\!$.

Now, suppose ${|V(G)|-k\leq j<|V(G)|}$. 
Then the relation $\phi_1[G;o,z,q_1,q_2]$ does not contain any tuples  $(i,v,m_1,m_2)$ with ${i=j}$.
However,
there exist a vertex ${v\in V(G)}$ and numbers ${n_1,n_2\in N(G)}$ such that ${(j,v,n_1,n_2)\in \phi_2[G,\alpha[X^k/X];o,z,q_1,q_2]}$
because ${X^k=S^k}\!$, by inductive assumption, and $S^k_{j'}$ is non-emp\-ty for all ${j'\geq |V(G)|-k+1}$, by definition.
Since we have ${X^k=S^k}\!$, and ${j+1\geq |V(G)|-k+1}$ and $j< |V(G)|$, 
the relation ${\{(n_1,n_2)\mid (j+1,w,n_1,n_2)\in X^k\}}$
is the representation of the canon of $G[D_{j+1,w}]$ for all $w\in V(G)$.
Therefore, $(j,v,X^k)$ is a suitable triple for all $v\in V(G)$.
As shown in Step~4, the relation
$\phi_{E}[G,j,v,X^k;q_1,q_2]$ is the edge relation of the canon $f(G[D_{j,v}])$ of 
$G[D_{j,v}]$ for suitable triples $(j,v,X^k)$.
Further, ${\phi_{n_w}[G,j,v;q_1,q_2]=\{(|D_{j,v}|,|D_{j,v}|)\}}$.
Thus, the relation
$(\phi_{E}\lor\phi_{n_w})[G,j,v,X^k;q_1,q_2]$,
is the representation of the canon of $G[D_{j,v}]$ for all vertices $v\in V(G)$,
and it follows that $X^{k+1}_j=S^{k+1}_j\!$.
\end{proof}

\smallskip

\begin{proof}[Proof of Theorem~\ref{thm:ModDecompTh}]
 Let $\CC$ be a class of graphs that is closed under induced subgraphs.
 Further, let $\CC^*_{\textup{prim}}$ admit $\FPC$-de\-fin\-a\-ble parameterized canonization.
 Then there exists  an $\FPC$-can\-on\-iza\-tion $\Theta^{\text{c}}$ of the class $\CC_{\CK\CI}^*$
 (Observations~\ref{obs:parametersInCanonizations} and~\ref{obs:transductionCKI}).
 Now according to Corollary~\ref{cor:edgesCanon},
 the $\FPC$-for\-mula $\phi_K$ defines the edge relation of the canon $f(G)$ for all $G\in \CC$. 
Therefore, $\Theta'=(\theta_U',\theta_E',\theta_{\leq}')$  with 
\begin{align*}
 \theta_U'(s_1)&:=0\leq s_1,\\ 
 \theta_E'(s_1,s_2)&:=\phi_K(s_1,s_2),\\
 \theta_{\leq}'(s_1,s_2)&:=s_1\leq s_2
\end{align*}
is an $\FPC$-can\-on\-iza\-tion of the graph class $\CC$.
\end{proof}

\subsection{Variations of the Modular Decomposition Theorem}
In this section we show variations of the Modular Decomposition Theorem, which might be helpful in future applications.
Let $\CC$ be a graph class that is closed under taking induced subgraphs.
We prove that  $\CC$ admits $\FPC$-de\-fin\-a\-ble canonization if the class of prime graphs of $\CC$ admits $\FPC$-de\-fin\-a\-ble orders,
and we present an analog of the Modular Decomposition Theorem for poly\-no\-mi\-al-time computable canonization.
The Modular Decomposition Theorem and the just mentioned analog of it require a canonization
of the class $\CC^*_{\textup{prim}}$.
We also show that we can relax this requirement, and prove that 
a canonization of the class of all prime graphs from $\CC$ that are colored with elements from a linearly ordered set is sufficient.

An $\FPC$-for\-mu\-la $\varphi(\tup{x},y,y')$ \emph{defines orders} on a class $\CC$ of graphs if
for all graphs $G\in \CC$ there is a tuple $\tup{v}\in G^{\tup{x}}$ such that the binary relation
$\varphi[G,\tup{v};y,y']$ is a linear order on $V(G)$.
We say a graph class $\CC$ \emph{admits $\FPC$-de\-fin\-a\-ble orders}, 
if there exists an $\FPC$-for\-mu\-la that defines orders on $\CC$. 

For a graph class $\CC$, let $\CC_\text{prim}$ be the class of all prime graphs from~$\CC$. 

\begin{corollary}\label{cor:moddecthm}
 Let $\CC$ be a graph class that is closed under induced subgraphs.
 If $\CC_\textup{prim}$ admits $\FPC$-de\-fin\-a\-ble orders, 
 then $\CC$ admits $\FPC$-de\-fin\-a\-ble canonization. 
\end{corollary}
\begin{proof}
Let graph class $\CC$ be closed under induced subgraphs,
and let $\varphi(\tup{x},y_1,y_2)$ be an $\FPC$-for\-mu\-la that defines orders on $\CC_\textup{prim}$.
We use the formula $\varphi$ to define a parameterized $\FPC$-can\-on\-iza\-tion of the class $\CC^*_{\textup{prim}}$.
Then Corollary~\ref{cor:moddecthm} follows directly from the Modular Decomposition Theorem.

First of all, we use $\varphi(\tup{x},y_1,y_2)$ to define an $\FPC$-for\-mu\-la
$\varphi_\textup{lin}(\tup{x})$ where
for all $G\in\CC_\textup{prim}$ and $\tup{v}\in G^{\tup{x}}$ we have
\begin{align*}G\models \varphi_\textup{lin}(\tup{v})\iff \varphi[G,\tup{v};y,y']\text{ is a linear order on }V(G).\end{align*}
As it can be tested in first-or\-der logic whether a binary relation is a linear order, i.e., a transitive, antisymmetric and total relation, 
the formula $\varphi_\textup{lin}$ is $\FPC$-de\-fin\-able.

Since we can define orders on $\CC_\text{prim}$, 
we can also define orders on the underlying graphs of the $\LO$-col\-ored graphs from $\CC^*_{\textup{prim}}$.
We simply pull back the formula $\varphi$ under $\FPC[{(\{V,E,M,\trianglelefteq,L\}, \{E\})}]$-trans\-duc\-tion 
$\Theta=(V(x), E(x,x'))$,
which maps every $\LO$-col\-ored graph to (an isomorphic copy of) its underlying graph.
We do the same for the formula $\varphi_\textup{lin}$.

Let $\leq_V$ be a linear order on the vertex set $V(G^*)$ of the underlying graph of $G^*\!\in\CC^*_{\textup{prim}}$.
We can use $\leq_V$ and the linear order $\trianglelefteq$ on the set $M(G^*)$  of basic color elements
to construct a linear order~$\leq^*$ on the universe $U(G^*)$ of~$G^*\!$. We let
\begin{align}\label{eqn:leqw}
 \leq^*\,:=\,\leq_V\cup\trianglelefteq\cup\,\{(v,m)\mid v\in V(G^*),\,m\in M(G^*)\}.
\end{align}

We now define a parameterized $\FPC$-can\-on\-iza\-tion $\Theta_{\leq}(x)$, which maps each  
$\LO$-col\-ored graph $G^*\!\in \CC^*_{\textup{prim}}$ to an ordered copy $(G^*\!,\leq^*)$.
Valid parameters of this transduction
are all tuples $\tup{v}\in G^{\tup{x}}$ where $\varphi[G,\tup{v};y,y']$ is a linear order on the vertex set $V(G)$.
We let $${\Theta_{\leq}(\tup{x})=(\theta_{\dom},\theta_{U},\theta_{V},\theta_{E},\theta_{M},\theta_{\trianglelefteq},\theta_{L},\theta_{\leq})},$$ where
\begin{align*}
&\theta_{\dom}(\tup{x}):=\varphi_\textup{lin}^{-\Theta}(\tup{x}), 
&&\theta_{M}(\tup{x},y):=M(y),\\
&\theta_{U}(\tup{x},y):= \true,
&&\theta_{\trianglelefteq}(\tup{x},y,y'):=\,\trianglelefteq\!(y,y'),\\
&\theta_{V}(\tup{x},y):= V(y),
&&\theta_{L}(\tup{x},y,y'\!,y''):=L(y,y'\!,y''),\\
&\theta_{E}(\tup{x},y,y'):=E(y,y'),&&
\end{align*}
and
\begin{align*}
\theta_{\leq}(\tup{x},y,y')&:=  \varphi^{-\Theta}(\tup{x},y,y')\, \lor \trianglelefteq\!(y,y') \lor \big(V(y)\land M(y')\big).\qquad\quad\ \,
\end{align*}
The formula $\theta_{\dom}$, that is, the pull-back of $\varphi_\textup{lin}$, defines the valid parameters, and 
the formula $\theta_{\leq}$ defines the linear order $\leq^*$ from \eqref{eqn:leqw} by using the pull-back  of  the $\FPC$-for\-mu\-la
$\varphi$.
\end{proof}

\smallskip

\noindent
In Section~\ref{app:sec:MDT-canonmapping} we recursively defined a canonization mapping for the graph class $\CC$.
It is not hard to see that this canonization mapping can be computed in polynomial time if 
there exists a canonization mapping for $\CC^*_{\textup{prim}}$ that is computable in polynomial time.
Thus, the Modular Decomposition Theorem can be transferred to polynomial time:
\begin{corollary}\label{cor:ptimecanonizationMDT}
 Let $\CC$ be a class of graphs that is closed under induced subgraphs.
 If $\CC^*_{\textup{prim}}$ admits poly\-no\-mi\-al-time canonization, 
 then $\CC$ admits poly\-no\-mi\-al-time canonization. 
\end{corollary}

\smallskip
\noindent
Within the Modular Decomposition Theorem and Corollary~\ref{cor:ptimecanonizationMDT} it is possible
to relax the requirement of a canonization of the class $\CC^*_{\textup{prim}}$ to a canonization of the class $\CC'_{\textup{prim}}$,
that is, the class of all prime graphs from $\CC$ that are colored with elements from a linearly ordered set.
More precisely, for a class $\CC$ of graphs that is closed under induced subgraphs, 
we let $\CC'_{\textup{prim}}$ be the class of all
tuples $G'=(V,E,C,\lin,f)$ with the following properties:
\begin{enumerate}
 \item The pair $(V,E)$ is a prime graph from the class $\CC$, and $|V|\geq 4$.
 \item The set of \emph{colors} $C$ is a non-emp\-ty finite set with $C\cap V=\emptyset$.
 \item The binary relation $\lin \subseteq C^2$ is a linear order on $C$.
 \item The \emph{coloring} $f\subseteq V \times C$ is a binary relation where for each vertex $v \in V$ 
 there exists exactly one color $c\in C$ with $(v,c)\in f$. We also denote this color $c$ by $f(v)$.
 We say a color $c\in C$ is \emph{used} if there exists a vertex $v\in V$ with $c=f(v)$. 
\end{enumerate}
To represent a colored graph ${G'\!=(V,E,C,\lin,f)}\in \CC'_{\textup{prim}}$ as a logical structure we extend the 
$5$-tu\-ple by a set $U$ to a $6$-tu\-ple $(U,V,E,M,\trianglelefteq,L)$, 
and we require that $U=V\dcup C$  additionally to the properties above. 

In the following, we call the colors of an $\LO$-col\-ored graph \emph{$\LO$-col\-ors}, and 
we say an $\LO$-col\-or $D$ is \emph{used} if there exists a vertex that is colored with $D$.

\begin{lemma}\label{lem:coloredgraphlemma}
  Let $\CC$ be a class of graphs. 
  If there exists a (parameterized) $\FPC$-can\-on\-iza\-tion of $\CC'_{\textup{prim}}$, 
  then there also exists one of  $\CC^*_{\textup{prim}}$.
\end{lemma}
\begin{proof}
   Let $\CC$ be a graph class, and let us suppose that the class $\CC'_{\textup{prim}}$ admits $\FPC$-de\-fin\-a\-ble (parameterized) canonization.
   According to Lemma~\ref{lem:parametersInCanonizations}, 
   we can assume that there exists an $\FPC$-can\-on\-iza\-tion $\Theta'^{\textup{c}}$ of $\CC'_{\textup{prim}}$.
   We show that there exists an $\FPC$-can\-on\-iza\-tion $\Theta$ of $\CC^*_{\textup{prim}}$.
  
  In order to do this, we first map each $\LO$-col\-ored graph $G^*\!\in \CC^*_{\textup{prim}}$ to a colored graph $G'\!\in \CC'_{\textup{prim}}$. 
  Within this mapping we replace the set of used $\LO$-col\-ors with the initial segment of the natural numbers.
  We define the mapping by an $\FPC$-count\-ing transduction $\Theta^{\raute}\!$, and
  apply Proposition~\ref{thm:CountingTransduction}, to show that there exists an $\FPC$-trans\-duc\-tion ${\Theta^*}\vphantom{\Theta}'$ 
  that defines essentially the same mapping.
  Then we compose  ${\Theta^*}\vphantom{\Theta}'$ and the canonization $\Theta'^{\textup{c}}$ to a 
  transduction $\Theta^{*\textup{c}}$ (Proposition~\ref{prop:composition}). 
  An overview of the different transductions can be found in Figure~\ref{fig:overviewtransduction-coloredgraphs}.
  Finally, we construct $\Theta$ with the help of $\Theta^{*\textup{c}}\!$, 
  and substitute the colors used in the ordered colored graphs again by their corresponding $\LO$-col\-ors.

  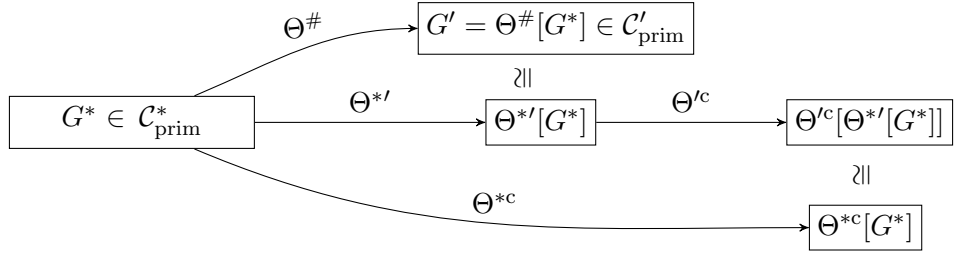
\begin{figure}[htbp]
    \begin{tikzpicture}[scale=1]
        \tikzstyle{vertex}=[anchor=west, draw,inner sep=3pt,align=left]
        
        \begin{scope}
            \node[draw,text width=3cm,align=center] (1) at (0,0) {$G^*\!\in\CC^*_{\textup{prim}}$};
            \node[vertex] (2b) at (3.8,1.25) {$G'=\Theta^{\raute}[G^*]\in \CC'_{\textup{prim}}$};
            \node[anchor=west,inner sep=1.5pt,align=left] at (4.8,0.6) {\hspace*{1mm} \rotatebox{90}{$\cong$}};
            \node[vertex] (2c) at (4.7,0) {${\Theta^*}\vphantom{\Theta}'[G^*]$};
            \node[vertex] (3a) at (8.7,0) {$\Theta'^{\textup{c}}[{\Theta^*}\vphantom{\Theta}'[G^*]]$};
            \node[anchor=west,inner sep=1.5pt,align=left] at (9.0,-0.7) {\hspace*{4mm} \rotatebox{90}{$\cong$}};
            \node[vertex] (3b) at (9.0,-1.4) {$\Theta^{*\textup{c}}[G^*]$};
            \draw[->] (1) to node[above]{${\Theta^*}\vphantom{\Theta}'$} (2c);
            \draw[->,out=23,in=180] (1) to node[above]{\vspace*{5mm}\rotatebox{5}{$\Theta^{\raute}$}} (2b) ;
            \draw[->,out=-23,in=180] (1) to node[above]{$\Theta^{*\textup{c}}$} (3b);
            \draw[->] (2c) to node[above]{$\Theta'^{\textup{c}}$}  (3a);
        \end{scope}
    \end{tikzpicture}
\caption{Overview of the transductions}
\label{fig:overviewtransduction-coloredgraphs}
\end{figure}

  First, we map each $\LO$-col\-ored graph $G^*\!=(U,V,E,M,\trianglelefteq, L)\in \CC^*_{\textup{prim}}$ to a colored graph 
  $G'\!=(U'\!,V'\!,E'\!,C'\!,\lin'\!,f')\in \CC'_{\textup{prim}}$.
  The linear order $\trianglelefteq$ on $M$ induces a linear order $\trianglelefteq_{\CL}$ on the set of $\LO$-col\-ors $\CL:=\{L_v\mid v\in V\}$.
  We construct $G'$ by substituting for each vertex $v\in V$ the $\LO$-col\-or $L_v$ by the number corresponding to the position
  of $L_v$ within the linear order $\trianglelefteq_{\CL}$ on $\CL$.
  It is not hard to see that there exists an $\FPC$-formula $\varphi^*_\textup{pos}(x,p)$ such that
  for all $G^*\!\in \CC^*_{\textup{prim}}$, $v\in V(G^*)$ and $i\in N(G^*)$ we have
  $G^*\models \varphi^*_\textup{pos}[v,i]$ if, and only if, the $\LO$-col\-or $L_v$ is at position $i$ regarding $\trianglelefteq_{\CL}$.
  We let $C'$ be the set of numbers $N(G^*)$, and we let $\lin'$ be the natural linear order on $N(G^*)$.
  
  The following $\FPC$-count\-ing transduction $\Theta^{\raute}$ 
  maps each $\LO$-col\-ored graph $G^*\!\in \CC^*_{\textup{prim}}$ to a colored graph $G'\!\in \CC'_{\textup{prim}}$ as described above.
  We let 
  $$\Theta^{\raute}=(\theta^{\raute}_{U}(x),\theta^{\raute}_{V}(x),\theta^{\raute}_{E}(x,x'),
  \theta^{\raute}_{C}(p),\theta^{\raute}_{\lin}(p,p'),\theta^{\raute}_{f}(x,p)),$$
  where 
  \begin{align*}
&\theta^{\raute}_{U}(x):= V(x), && \theta^{\raute}_{C}(p):= \true,\\
&\theta^{\raute}_{V}(x):= V(x), && \theta^{\raute}_{\lin}(p,p'):=p\leq p',\\
&\theta^{\raute}_{E}(x,x'):= E(x,x'),&& \theta^{\raute}_{f}(x,p):= \varphi^*_\textup{pos}(x,p).
\end{align*}

\noindent
According to Proposition~\ref{thm:CountingTransduction}, there exists an $\FPC$-trans\-duc\-tion ${\Theta^*}\vphantom{\Theta}'\!$,
such that $\Theta^{\raute}[G^*]$ and ${\Theta^*}\vphantom{\Theta}'[G^*]$ are isomorphic for all $\LO$-col\-ored graphs $G^*\!\in \CC^*_{\textup{prim}}$.
We compose the two transductions ${\Theta^*}\vphantom{\Theta}'$ and $\Theta'^{\textup{c}}$ (Proposition~\ref{prop:composition})
and obtain an $\FPC[\{V,E,M,\trianglelefteq, L\},\{V,E,C,\lin, f,\leq\}]$-trans\-duc\-tion 
$\Theta^{*\textup{c}}\!=(\theta_{U}^{*\textup{c}},\theta_{\approx}^{*\textup{c}},\theta_{V}^{*\textup{c}},
\theta_{E}^{*\textup{c}},\theta_{C}^{*\textup{c}},\theta_{\lin}^{*\textup{c}},\theta_{f}^{*\textup{c}},\theta_{\leq}^{*\textup{c}})$.
The transduction $\Theta^{*\textup{c}}$ 
maps each $\LO$-col\-ored graph $G^*\!\in \CC^*_{\textup{prim}}$ 
to the canon $\Theta^{*\textup{c}}[G^*]$ of the colored graph ${\Theta^*}\vphantom{\Theta}'[G^*]\in \CC'_{\textup{prim}}$.
We use the linear order on the set of colors of $\Theta^{*\textup{c}}[G^*]$ to replace the color of each vertex again with the corresponding $\LO$-col\-or.

Let $\tup{u}$ be the tuple of domain variables of 
$\Theta^{*\textup{c}}$.
Let $\approx$ be the equivalence relation generated by $\theta_{\approx}^{*\textup{c}}[G^*;\tup{u},\tup{u}']$ on $(G^*)^{\tup{u}}\!$.
Similarly to the formula $\varphi^*_\textup{pos}(x,p)$, we can define an 
$\FPC$-formula $\varphi'_\textup{pos}(\tup{u},p)$ where
for all $G^*\in \CC^*_{\textup{prim}}$, $\tup{a}\in (G^*)^{\tup{u}}$ and $i\in N(G^*)$ we have
$G^*\models \varphi'_\textup{pos}[\tup{a},i]$ if, and only if,
$\tup{a}\modout_{\approx}\in V(\Theta^{*\textup{c}}[G^*])$ and 
the position of the color of $\tup{a}\modout_{\approx}$ is $i$ regarding the linear order induced by $\lin(\Theta^{*\textup{c}}[G^*])$
on the set  $\{f(\Theta^{*\textup{c}}[G^*])(v)\mid v\in V(\Theta^{*\textup{c}}[G^*])\}$ of colors.

We use $\Theta^{*\textup{c}}$ and the formulas $\varphi^*_\textup{pos}$ and $\varphi'_\textup{pos}$
to define an $\FPC$-can\-on\-iza\-tion $\Theta$ for the class $\CC^*_{\textup{prim}}$.
We let 
$$\Theta=(\theta_{U},\theta_{\approx},\theta_{V},\theta_{E},\theta_{M},\theta_{\trianglelefteq},\theta_{L},\theta_{\leq}),$$ 
where
\begin{align*}
&\theta_{U}(\tup{u},z):= \theta^{*\textup{c}}_{V}(\tup{u}), \\
&\theta_{\approx}(\tup{u},z,\tup{u}'\!,z'):= \big(\theta^{*\textup{c}}_{\approx}(\tup{u},\tup{u}')\land\lnot M(z)\land \lnot M(z') \big)\lor
  \big(z=z'\land M(z)\land M(z') \big),\\
&\theta_{V}(\tup{u},z):= \theta^{*\textup{c}}_{V}(\tup{u}) \land \lnot M(z),\\
&\theta_{E}(\tup{u},z,\tup{u}'\!,z'):=\theta^{*\textup{c}}_{E}(\tup{u},\tup{u'})\land \lnot M(z)\land \lnot M(z'),\\
&\theta_{M}(\tup{u},z):= M(z),\\
&\theta_{\trianglelefteq}(\tup{u},z,\tup{u}'\!,z'):=\;\trianglelefteq\!(z,z')\land M(z)\land M(z'),\\
&\theta_{L}(\tup{u},z,\tup{u}'\!,z'\!,\tup{u}''\!,z''):=\theta_{V}(\tup{u},z) \land \theta_{M}(\tup{u}'\!,z') \land \theta_{M}(\tup{u}''\!,z'')\\
&\phantom{\theta_{L}(\tup{u},z,\tup{u}'\!,z'\!,\tup{u}''\!,z''):=}\land 
   \exists p \exists x \big( \varphi'_{\textup{pos}}(\tup{u},p)\land \varphi^*_{\textup{pos}}(x,p)\land L(x,z'\!,z'')\big),\\
&\theta_{\leq}(\tup{u},z,\tup{u}'\!,z'):= \big(\theta^{*\textup{c}}_{\leq}(\tup{u},\tup{u}')\land \lnot M(z)\land \lnot M(z')\big)\lor  
      \theta_{\trianglelefteq}(\tup{u},z,\tup{u}'\!,z') \lor \big(\lnot M(z)\land M(z')\big).
\end{align*}

\noindent
Let $G^*\!\in \CC^*_{\textup{prim}}$. Then the vertices of $\Theta[G^*]$ correspond to the vertices of $\Theta^{*\textup{c}}[G^*]$, 
and the underlying graph of $\Theta[G^*]$ is isomorphic to the underlying graph of $\Theta^{*\textup{c}}[G^*]$.
The basic color elements of $\Theta[G^*]$ and their linear order correspond to the basic color elements of $G^*$ and their linear order.
The $\LO$-col\-or of each vertex of $\Theta[G^*]$ is defined as follows:
First, the position $i$ of the color of the corresponding vertex of $\Theta^{*\textup{c}}[G^*]$ is determined.
Then we pick an (arbitrary) vertex $v$ of $G^*$ whose $\LO$-col\-or is at position $i$ with respect to $\trianglelefteq_{\CL}$.
We assign the $\LO$-col\-or of $v$ to the vertex of $\Theta[G^*]$.
To obtain the linear order on the elements of $\Theta[G^*]$, we combine 
the linear order on the vertices of $\Theta[G^*]$,
which is obtained from the linear on the vertices of $\Theta^{*\textup{c}}[G^*]$,
and the linear order $\trianglelefteq\!(\Theta[G^*])$ on the basic color elements of $\Theta[G^*]$.
Consequently, $\Theta[G^*]$ is an ordered copy of $G^*$ for all $G^*\!\in \CC^*_{\textup{prim}}$.
\end{proof}

\noindent
The following corollary is an direct consequence of Lemma~\ref{lem:coloredgraphlemma} and the Modular Decomposition Theorem.
\begin{corollary}\label{cor:ModDecompThLinOrd}
 Let $\CC$ be a class of graphs that is closed under induced subgraphs.
 If $\CC'_{\textup{prim}}$ admits $\FPC$-de\-fin\-a\-ble (parameterized) canonization, 
 then $\CC$ admits $\FPC$-de\-fin\-a\-ble canonization. 
\end{corollary}

\noindent
Within the proof of Lemma~\ref{lem:coloredgraphlemma}, it is described how to obtain a polynomial-time canonization mapping for $\CC^*_{\textup{prim}}$
if there exists one for $\CC'_{\textup{prim}}$. Hence, Corollary~\ref{cor:ptimecanonizationMDT} implies the following:
\begin{corollary}\label{cor:ptimecanonizationMDTLinOrd}
 Let $\CC$ be a class of graphs that is closed under induced subgraphs.
 If $\CC'_{\textup{prim}}$ admits poly\-no\-mi\-al-time canonization, 
 then $\CC$ admits poly\-no\-mi\-al-time canonization. 
\end{corollary}

\begin{remark}
Note that if $\CC_{\textup{prim}}$ admits poly\-no\-mi\-al-time canonization, then 
it does not necessarily follow that $\CC$ admits poly\-no\-mi\-al-time canonization. 
Let us suppose there is a deterministic Turing machine $M$ that computes in polynomial time for all graphs $H\in \CC_{\textup{prim}}$ a 
linear order $\leq_H$ on the vertex set of $H$, 
such that  $H\cong H'\!$ implies $(H,\leq_H)\cong (H'\!,\leq_{H'})$ for all graphs $H,H'\in \CC_{\textup{prim}}$.
Then $\CC_{\textup{prim}}$ admits poly\-no\-mi\-al-time canonization.
Let us consider two graphs $G$ and $G'$ that are isomorphic.
Assume there exists an isomorphism $h$ between their modular contractions $G_\sim$ and $G'_\sim$
such that $h(v\modout_\sim)= v'\!\modout_\sim$ but $v\modout_\sim$ and $v'\!\modout_\sim$ represent modules that induce non-iso\-mor\-phic graphs.
For example, suppose that $G$ and $G'$ are isomorphic to the graph in Figure~\ref{fig:module2}, 
and that $v\modout_\sim$ and $v'\!\modout_\sim$ correspond to the two ends of the modular contraction (a~path of length $4$).
Depending on the input strings that represent $G_\sim$ and $G'_\sim$,
it is possible that $M$ computes linear orders $\leq_{G_\sim}$ and $\leq_{G'_\sim}$ such that $v\modout_\sim$ and $v'\!\modout_\sim$
occur at the same position in $\leq_{G_\sim}$ and $\leq_{G'_\sim}$, respectively.
If we use $\leq_{G_\sim}$  and the linear orders on the maximal proper modules of $G$
to construct a linear order for an isomorphic copy of~$G$, 
and use $\leq_{G'_\sim}$  and the linear orders on the maximal proper modules of $G'$
to construct a linear order for an isomorphic copy of $G'$,
we obtain ordered copies of $G$ and $G'$ that are not isomorphic.
\end{remark}

\section{Capturing \texorpdfstring{Polynomial Time}{PTIME} on Permutation Graphs}\label{sec:permutationgraphs}

In this section we use the Modular Decomposition Theorem
to prove that there exists an $\FPC$-can\-on\-iza\-tion of the class of permutation graphs.
More precisely, 
for prime permutation graphs  $G$
we show that there exist parameterized $\FP$-for\-mu\-las that define the strict linear orders of a realizer of $G$.
This directly implies that the class of prime permutation graphs admits $\FP$-de\-fin\-able orders.
As the class of permutation graphs is closed under induced subgraphs, we can apply Corollary~\ref{cor:moddecthm}, and obtain 
that canonization of the class of permutation graphs is definable in $\FPC$.
As a consequence, $\FPC$ captures $\PTIME$ on the class of permutation graphs.

\subsection{Preliminaries}

 Let $G=(V,E)$ be a graph, and let $<_1$ and $<_2$ be two strict linear orders on the vertex set~$V\!$.
 We call $(<_1,<_2)$ a \emph{realizer} of $G$ if  $u,v\in V$ are adjacent in $G$ if and only if 
they occur in different order in $<_1$ and $<_2$, that is,
$u<_1 v$ and $v<_2 u$, or $v<_1 u$ and $u<_2 v$.
A graph $G$ is a \emph{permutation graph} if there exists a realizer of $G$.
Figure~\ref{fig:permrealizer} shows an example of a permutation graph and a realizer of it. 
A detailed introduction to permutation graphs can be found in~\cite{golumbic}.
\begin{figure}[htbp]
\centering
\begin{minipage}{0.2\textwidth}
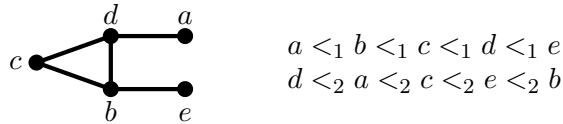
\end{minipage}\hspace{2em}
\begin{minipage}{0.3\textwidth}
$a<_1 b<_1 c<_1 d<_1 e$\\
$d<_2 a<_2 c<_2 e<_2 b$	
\end{minipage}
\caption{A permutation graph and a realizer}
\label{fig:permrealizer}
\end{figure}

\noindent
Let $\vartriangleleft_1$ and $\vartriangleleft_2$ be two binary relations. 
We call the pair $(\vartriangleleft_1,\vartriangleleft_2)$ \emph{transitive} if 
each of the binary relations $\vartriangleleft_1$ and $\vartriangleleft_2$ is transitive.
Further, we let the transitive closure $(\vartriangleleft_1,\vartriangleleft_2)^T$ of $(\vartriangleleft_1,\vartriangleleft_2)$
be the pair $(\vartriangleleft_1^T,\vartriangleleft_2^T)$ 
where $\vartriangleleft_1^T$ and $\vartriangleleft_2^T$ is the transitive closure of $\vartriangleleft_1$ and $\vartriangleleft_2$, respectively.
Let $G=(V,E)$ be a graph
and $(\vartriangleleft_1,\vartriangleleft_2)$ be a pair of binary relations on $V\!$.
The pair $(\vartriangleleft_1,\vartriangleleft_2)$ is \emph{closed under $E$}  
if for all vertices $u,v\in V$ and all $i\in[2]$ the following holds:
\begin{itemize}
 \item If $u\vartriangleleft_i v$ and $\{u,v\}\in E$, then $v\vartriangleleft_{3-i} u$.
 \item If $u\vartriangleleft_i v$ and $\{u,v\}\not\in E$, then $u\vartriangleleft_{3-i} v$.
\end{itemize}
Notice that each realizer of a graph $G=(V,E)$ is closed under the edge relation $E$. 
Moreover, we observe the following.
\begin{observation}
Let $G=(V,E)$ be a graph. Then 
a pair of binary relations $(\vartriangleleft_{1},\vartriangleleft_{2})$ is a realizer of $G$
if, and only if, $\vartriangleleft_{1}$ and $\vartriangleleft_{2}$ are strict linear orders and   
$(\vartriangleleft_{1},\vartriangleleft_{2})$ is closed under the edge relation~$E$.
\end{observation}

\noindent
Now for all $i\in[2]$ we let 
\begin{align*}
D_{3-i}^E&:=\big\{(v,u)\mid \text{$u\vartriangleleft_{i} v$ and $\{u,v\}\in E$}\big\}\text{ and}\\
D_{3-i}^{\not E}&:=\big\{(u,v)\mid \text{$u\vartriangleleft_{i} v$ and $\{u,v\}\not\in E$}\big\},
\end{align*}
and we let $(\vartriangleleft_1,\vartriangleleft_2)^E$ be the pair $(\vartriangleleft_1^E$, $\vartriangleleft_2^E)$ of relations where for all $i\in[2]$ we have
\begin{align*}
 {\vartriangleleft}_i^E&:= {\vartriangleleft}_i\cup\, D_i^E \cup D_i^{\not E}.
 \end{align*}
\begin{observation}
Let $G=(V,E)$ be a graph and $\vartriangleleft_1,\vartriangleleft_2$ be binary relations on $V\!$. 
Then $(\vartriangleleft_1,\vartriangleleft_2)^E$ is closed under $E$.
\end{observation}

\noindent
Let $\vartriangleleft_1,\vartriangleleft_2$ be binary relations on $V\!$.
It is not hard to see that the relations ${\vartriangleleft}_1^E$ and ${\vartriangleleft}_2^E$ are minimal with the property that
$\vartriangleleft_1\,\subseteq\, \vartriangleleft_1^E$, $\vartriangleleft_2\,\subseteq\, \vartriangleleft_2^E$ and 
$(\vartriangleleft_1^E,\vartriangleleft_2^E)$ is closed under $E$.
Thus, we call  $(\vartriangleleft_1,\vartriangleleft_2)^E$
the \emph{closure of $(\vartriangleleft_1,\vartriangleleft_2)$ under~$E$}.

\subsection{Defining Orders on Prime Permutation Graphs}\label{sec:defineRealizersinFPC}
We now show that the class of prime permutation graphs admits $\FP$-de\-fin\-able orders.
It is known that the realizer of a prime permutation graph is unique up to reversal and exchange \cite{Moehring85}.
Thus, a prime permutation graph has at most $4$ different realizers.
We prove that the strict linear orders of these realizers are definable in $\FP$.

Let $G=(V,E)$ be a prime permutation graph.
For each $w\in V$ we define two binary relations $\vartriangleleft_{1}^w$ and $\vartriangleleft_{2}^w$ on the vertex set $V\!$. 
If there exists a realizer $(<_1,<_2)$ of $G$ where $w$ is the first vertex of the first strict linear order $<_1$, 
then it will turn out that  $(\vartriangleleft_{1}^w,\vartriangleleft_{2}^w)=(<_1,<_2)$.

Let $w\in V\hspace{-1pt}$. In order to construct the binary relations $\vartriangleleft_{1}^w$ and $\vartriangleleft_{2}^w$, 
we recursively define relations $\vartriangleleft_{1,k}^w$ and $\vartriangleleft_{2,k}^w$ on the vertex set $V$ for all $k\geq 0$.
We begin with defining the relations for $k=0$.
As $w$ is the first element of the first strict linear order of the realizer that we want to reconstruct, we let 
  \begin{align*}
  	\vartriangleleft_{1,0}^w\,:=\,\{(w,v)\mid v\in V,v\not=w\}\text{\quad and\quad}
  	\vartriangleleft_{2,0}^w\,:=\,\emptyset.
  \end{align*}
Now, we recursively define $\vartriangleleft_{1,k+1}^w$ and $\vartriangleleft_{2,k+1}^w$ for all $k> 0$ as follows:\vspace{0.5pt}
\begin{align*}
(\vartriangleleft_{1,k+1}^w,\vartriangleleft_{2,k+1}^w):=((\vartriangleleft_{1,k}^w,\vartriangleleft_{2,k}^w)^{E})^T.\hspace*{2cm}
\end{align*}
Clearly, for all vertices $w\in V$ and all $k\geq 0$ the relations satisfy the property that 
\begin{align*}
 \vartriangleleft_{1,k}^w\ \subseteq\ \vartriangleleft_{1,k+1}^w\text{\quad and\quad} 
 \vartriangleleft_{2,k}^w\ \subseteq\ \vartriangleleft_{2,k+1}^w.\hspace*{1.5cm}
\end{align*}
Since the vertex set is finite, there exists an $m\geq 0$ such that 
$\vartriangleleft_{i,m}^w\,=\, \vartriangleleft_{i,m+1}^w$ for all $i\in[2]$. 
We define $\vartriangleleft_i^w\,:=\,\vartriangleleft_{i,m}^w$ for $i\in[2]$.

In the following, let $(<_1,<_2)$ be a realizer of the permutation graph $G$, and let $w$ be the first element of $<_1$.
We show that $(\vartriangleleft_{1}^w,\vartriangleleft_{2}^w)=(<_1,<_2)$.
By definition of $(\vartriangleleft_{1,0}^w,\vartriangleleft_{2,0}^w)$ 
we have $\vartriangleleft_{1,0}^w\,\subseteq\, <_1$ and $\vartriangleleft_{2,0}^w\,\subseteq\, <_2$. 
Further,  we obtain ${(\vartriangleleft_1^w,\vartriangleleft_2^w)}$ from $(\vartriangleleft_{1,0}^w,\vartriangleleft_{2,0}^w)$
by recursively taking the closure under the edge relation $E$ and the transitive closure.
Since the realizer ${(<_1,<_2)}$ is closed under both, the following observation holds.
\begin{observation}\label{obs:orderscontained}
For all $k\geq 0$, it holds that $\vartriangleleft_{1,k}^w\, \subseteq\, <_1$ and $\vartriangleleft_{2,k}^w\, \subseteq\,  <_2$.
\end{observation}

\noindent
For all $k\geq 0$, relations $\vartriangleleft_{1,k}^w$ and $\vartriangleleft_{2,k}^w$  are strict partial orders.
By induction on $k$, it can be shown that incomparability with respect to $\vartriangleleft_{i,k}^w$ for $i\in[2]$ is transitive. 
It follows that $\vartriangleleft_{1,k}^w$ and $\vartriangleleft_{2,k}^w$ are strict weak orders.

\pagebreak[2]
\begin{lemma}\label{lem:triagleleftstrictweakorder}
Relations $\vartriangleleft_{1,k}^w$ and $\vartriangleleft_{2,k}^w$ are strict weak orders for all $k\geq 0$.
\end{lemma}

\begin{proof}
In order to show that a relation is a strict weak order, we have to prove that it is a strict partial order and that incomparability is transitive. 
Let $k\geq 0$.
As $<_1$ and $<_2$ are irreflexive, it follows from $\vartriangleleft_{1,k}^w\, \subseteq\, <_1$ and 
$\vartriangleleft_{2,k}^w\, \subseteq\,  <_2$ (Observation~\ref{obs:orderscontained})
that $\vartriangleleft_{1,k}^w$ and $\vartriangleleft_{2,k}^w$ are irreflexive as well.
Further, the relations $\vartriangleleft_{1,k}^w$ and $\vartriangleleft_{2,k}^w$ are transitive.
Hence, $\vartriangleleft_{1,k}^w$ and $\vartriangleleft_{2,k}^w$ are strict partial orders.
It remains to show that incomparability is transitive. 
Two vertices $x$ and $y$ that are incomparable with respect to $\vartriangleleft_{i,k}^w$, are denoted by $x\sim_{i,k}^w y$. 
Let us consider $k=0$. 
With respect to  $\vartriangleleft_{1,0}^w$, all elements in $V\setminus\{w\}$ are pairwise incomparable and $w$ is incomparable to itself.
Further, all elements in $V$ are pairwise incomparable with respect to $\vartriangleleft_{2,0}^w$.
Thus, for $\vartriangleleft_{1,0}^w$ and $\vartriangleleft_{2,0}^w$ incomparability is transitive.
To show that incomparability is transive for $k>0$ we need the following claims.
\begin{claim}\label{clm:incompvartriangleleft}
 Let $\kappa\geq 0$, $i\in[2]$ and $x,y\in V\!$. If $x\sim_{i,\kappa+1}^w y$, then 
  $x\sim_{1,\kappa}^w y$ and $x\sim_{2,\kappa}^w y$.
\end{claim}
\begin{proofofclaim} 
 Let $\kappa\geq 0$, $i\in[2]$ and $x,y\in V\!$.
 Without loss of generality, suppose that $i=1$ and that $x\sim_{1,\kappa+1}^w y$.
 For a contradiction let us assume that
 $x$ and $y$ are comparable with respect to $\vartriangleleft_{1,\kappa}^w$ or $\vartriangleleft_{2,\kappa}^w$.
 If $x$ and $y$ are comparable with respect to $\vartriangleleft_{1,\kappa}^w$,
 then it follows directly that $x$ and $y$ are comparable with respect to $\vartriangleleft_{1,\kappa+1}^w$,
 since $\vartriangleleft_{1,\kappa}\,\subseteq\, \vartriangleleft_{1,\kappa+1}^w$, and we have a contradiction.
 Thus, suppose  $x$ and $y$ are comparable with respect to $\vartriangleleft_{2,\kappa}^w$.
 Then $x$ and $y$ are also comparable with respect to $(\vartriangleleft_{1,\kappa}^w)^E$, and therefore also with respect to 
 $((\vartriangleleft_{1,\kappa}^w)^{E})^T\!=\,\vartriangleleft_{1,\kappa+1}^w$, a contradiction.
\end{proofofclaim}
\begin{claim}\label{clm:incompdottriangleleft}
  Let $\kappa\geq 0$, $i\in[2]$ and $y,z\in V\!$. Further, let $\vartriangleleft_{1,\kappa}^w$ and $\vartriangleleft_{2,\kappa}^w$ be strict weak orders, 
  and suppose $y\sim_{i,\kappa+1}^w z$.
  Then for all vertices $v\in V$ the following holds: If $v(\vartriangleleft_{i,\kappa}^w)^E z$, then  $v(\vartriangleleft_{i,\kappa}^w)^E y$.
\end{claim}
\begin{proofofclaim}
 Let $\kappa\geq 0$, $i\in[2]$ and $v,y,z\in V\!$.
 Let $\vartriangleleft_{1,\kappa}^w$ and $\vartriangleleft_{2,\kappa}^w$ be strict weak orders, 
  and suppose that $y\sim_{i,\kappa+1}^w z$ and $v(\vartriangleleft_{i,\kappa}^w)^E z$.
 Without loss of generality, assume $i=1$.
 Relation ${(\vartriangleleft_{1,\kappa}^w)^E}$ contains only pairs that are in $\vartriangleleft_{1,\kappa}^w$, in
 $D_{1,\kappa}^E$ or in $D_{1,\kappa}^{\not E}$.
Therefore, $v(\vartriangleleft_{1,\kappa}^w)^E z$ implies that either $v\vartriangleleft_{1,\kappa}^w z$, 
$z\vartriangleleft_{2,\kappa}^w v$ or $v\vartriangleleft_{2,\kappa}^w z$.
If we have $v\vartriangleleft_{1,\kappa}^w z$, then we also have $v\vartriangleleft_{1,\kappa}^w y$, 
as $y$ and $z$ are incomparable with respect to $\vartriangleleft_{1,\kappa}^w$ by Claim~\ref{clm:incompvartriangleleft}
and $\vartriangleleft_{1,\kappa}^w$ is a strict weak order.
Analogously, $z\vartriangleleft_{2,\kappa}^w v$ and $v\vartriangleleft_{2,\kappa}^w z$ imply
$y\vartriangleleft_{2,\kappa} v$ and $v\vartriangleleft_{2,\kappa}^w y$, respectively.
Hence, in each of the cases we obtain
$v(\vartriangleleft_{1,\kappa}^w)^E y$.
\end{proofofclaim}

\noindent
Now, let us assume there exists a $k>0$ such that incomparability is not transitive for $\vartriangleleft_{1,k}^w$ or $\vartriangleleft_{2,k}^w$,
and suppose $k$ is minimal.
Without loss of generality, assume that incomparability is not transitive for $\vartriangleleft_{1,k}^w$.
Consequently, there exist vertices $x,y,z\in V$ such that $x\sim_{1,k} y$, $y\sim_{1,k} z$ and $x\not \sim_{1,k} z$.
Hence, $x$ and $z$ are comparable, which means $x\vartriangleleft_{1,k}^w z$ or $z\vartriangleleft_{1,k}^w x$. 
Without loss of generality, suppose $x\vartriangleleft_{1,k}^w z$.
Since $\vartriangleleft_{1,k}^w$ is the transitive closure of $(\vartriangleleft_{1,k-1}^w)^E$, 
there exists an $l\geq 0$ and $v_0,v_1,\dots,v_{l+1}$ such that 
\begin{align*}
x&=v_0(\vartriangleleft_{1,k-1}^w)^E \dots(\vartriangleleft_{1,k-1}^w)^E v_l(\vartriangleleft_{1,k-1}^w)^E v_{l+1}=z.\\
\intertext{%
As we know that the relations $\vartriangleleft_{1,k-1}^w$ and $\vartriangleleft_{2,k-1}^w$ are strict weak orders, 
that $y\sim_{1,k} z$, and that ${v_l(\vartriangleleft_{1,k-1}^w)^E z}$,
we can apply Claim~\ref{clm:incompdottriangleleft}. We obtain that $v_l\vartriangleleft_{1,k-1}^E y$.
Thus, we have}
x&=v_0(\vartriangleleft_{1,k-1}^w)^E\dots(\vartriangleleft_{1,k-1}^w)^E v_{l}(\vartriangleleft_{1,k-1}^w)^E y,
\end{align*}
and therefore, $x\vartriangleleft_{1,k}^w y$, a contradiction.
\end{proof}

\begin{corollary}\label{cor:triagleleftstrictweakorder}
Relations $\vartriangleleft_{1}^w$ and $\vartriangleleft_{2}^w$ are strict weak orders.
\end{corollary}

\noindent
According to Corollary~\ref{cor:triagleleftstrictweakorder},
incomparability with respect to $\vartriangleleft_{i}^w$ is an equivalence relation.
Each equivalence class of this equivalence relation is of size 1, 
as every equivalence class of size at least 2 would be a module.
As a consequence, we obtain Theorem~\ref{thm:stictlinearorders}.

\begin{theorem}\label{thm:stictlinearorders}
Relations $\vartriangleleft_{1}^w$ and $\vartriangleleft_{2}^w$ are strict linear orders.
\end{theorem}

\begin{proof}
 Let us assume that $\vartriangleleft_{1}^w$ is not a strict linear order. 
 Since $\vartriangleleft_{1}^w$ is a strict weak order by Corollary~\ref{cor:triagleleftstrictweakorder}, 
 there exist two distinct vertices $u,v$ such that $u\sim_1 v$, i.e., $u$ and $v$ are incomparable regarding $\vartriangleleft_{1}^w$.
 Hence, the equivalence class $u\modout_{\sim_1}$ contains at least two elements. 
 In the following we prove that $u\modout_{\sim_1}$ is a module.
 Let us assume $u\modout_{\sim_1}$ is not a module. Then there exists a vertex $z\not \in u\modout_{\sim_1}$
 and vertices $x,y\in u\modout_{\sim_1}$ such that $z$ and $x$ are adjacent and $z$ and $y$ are not adjacent.
 As $\vartriangleleft_{1}^w$ is a strict weak order, we either have $z\vartriangleleft_{1}^w x$ and $z\vartriangleleft_{1}^w y$, or
 $x\vartriangleleft_{1}^w z$ and $y\vartriangleleft_{1}^w z$.
 Let us assume $z\vartriangleleft_{1}^w x$ and $z\vartriangleleft_{1}^w y$. The other case can be shown analogously.
  Since there is an edge between $z$ and $x$ and no edge between $z$ and $y$, 
  and $(\vartriangleleft_{1}^w,\vartriangleleft_{2}^w)$ is closed under edge relation $E$,
  we have $x\vartriangleleft_{2}^w z$ and $z \vartriangleleft_{2}^w y$.
 Therefore, we must also have $x\vartriangleleft_{2}^w y$, by transitivity of $\vartriangleleft_{2}^w$.
 As $(\vartriangleleft_{1}^w,\vartriangleleft_{2}^w)$ is closed under $E$, we obtain
 that $x\vartriangleleft_{1}^w y$ or $y\vartriangleleft_{1}^w x$.
 Hence, $x$ and $y$ are comparable with respect to $\vartriangleleft_{1}^w$, a contradiction.
 Consequently,   $u\modout_{\sim_1}$ is a module with $|u\modout_{\sim_1}|\geq 2$.
 Clearly, $u\modout_{\sim_1}$ cannot be the vertex set $V$ since we know $w\vartriangleleft_{1}^w v$ for all $v\not=w$, 
 where $w$ is the initial vertex. 
 Thus, $u\modout_{\sim_1}$ is a non-triv\-ial module, a contradiction to $G$ being prime.
 
 Similarly we can prove that $\vartriangleleft_{2}^w$ is a strict linear order. 
 To show that a module $u\modout_{\sim_2}$ with $|u\modout_{\sim_2}|\geq 2$ for $u\in V$ cannot be the vertex set $V\!$,
 we argue as follows:
 Since $w\vartriangleleft_{1}^w v$ for all $v\in V$ with $v\not=w$ and  $(\vartriangleleft_{1}^w,\vartriangleleft_{2}^w)$ is closed under $E$,
 vertex $w$ is comparable to all $v\not=w$ with respect to $\vartriangleleft_{2}^w$.
 Hence, the equivalence relation $\sim_2$ has at least two equivalence classes.
\end{proof}

\begin{corollary}\label{cor:gleichheit}
	We have  $\vartriangleleft_{1}^w\, =\, <_1$ and $\vartriangleleft_{2}^w\, =\,  <_2$.
\end{corollary}

\noindent
The relations  $\vartriangleleft_{1}^w$ and $\vartriangleleft_{2}^w$ are definable in 
fixed-point logic, i.e.,
there are
$\FP$-for\-mu\-las $\varphi_{\vartriangleleft_{1}}(x,y,y')$ and $\varphi_{\vartriangleleft_{2}}(x,y,y')$
such that for all prime permutation graphs $G=(V,E)$ and all $w,v,v'\in V$ we have\vspace{-1mm}
\begin{align*}
G\models \varphi_{\vartriangleleft_{i}}[w,v,v'] \iff v\vartriangleleft_{i}^w v'\!.
\end{align*}
In order to define $\varphi_{\vartriangleleft_{i}}$ we use a simultaneous inflationary fixed-point operator. 
Within this fixed-point operator, two binary relational variables $X_1$ and $X_2$ are used
to create the strict linear orders $\vartriangleleft_{1}^w$ and $\vartriangleleft_{2}^w$.
Let $X_1^k$ and $X_2^k$ be the relations that we obtain after the $k$th iteration within the simultaneous fixed-point operator.
We can design the operator such that $X_1^k$ and $X_2^k$ are precisely $\vartriangleleft_{1,k}^w$ and $\vartriangleleft_{2,k}^w$.
Since the transitive closure and the closure under the edge relation are definable in $\FP$, this operator is definable in $\FP$.

As a consequence of Corollary~\ref{cor:gleichheit}, we obtain the following:
\begin{corollary} Let  $\varphi(x,y,y'):=\varphi_{\vartriangleleft_{1}}(x,y,y') \lor y=y'$.
Then the $\FP$-for\-mu\-la $\varphi$ defines orders on the class of prime permutation graphs.
\end{corollary}

\noindent
Thus, the class of prime permutation graphs admits $\FP$-de\-fin\-able orders.
Since the class of permutation graphs is closed under taking induced subgraphs, we can apply Corollary~\ref{cor:moddecthm}.
As a result we obtain the following theorem:
\begin{theorem}
 The class of permutation graphs admits $\FPC$-de\-fin\-able canonization. 
\end{theorem}
\begin{corollary}
 $\FPC$ captures $\PTIME$ on the class of permutation graphs.
\end{corollary}

\section{Conclusion}
So far, little is known about logics capturing $\PTIME$ on classes of graphs that are closed under induced subgraphs.
This paper makes a contribution in this direction.
We provide a tool, the Modular Decomposition Theorem, which 
simplifies proving that canonization is definable on such graph classes.
Therefore, it also simplifies proving that $\PTIME$ can be captured on them. 
By means of the Modular Decomposition Theorem, 
we have shown in this paper that
there exists an $\FPC$-can\-on\-iza\-tion of the class of permutation graphs. Thus,
$\FPC$ captures $\PTIME$ on this class of graphs. 
The Modular Decomposition Theorem can also be applied to show that 
the class of chordal comparability graphs admits $\FPC$-de\-fin\-able canonization (see~\cite{diss}).
It follows that $\FPC$ captures $\PTIME$ on the class of chordal comparability graphs 
and that there exists a poly\-no\-mi\-al-time algorithm for chordal comparability graph canonization.
The author is optimistic that the Modular Decomposition Theorem can be used to obtain new results on further classes of graphs.

It would be interesting to find out whether a tool similar to the Modular Decomposition Theorem
can also be obtained for split (or join) decomposition. 
Such a ``Split Decomposition Theorem'' could be used to prove that $\FPC$ captures $\PTIME$ on 
the class of circle graphs, which are a generalization of permutation graphs and well-struc\-tured with respect to split decompositions.

Within this paper, we have also shown that there exists a log\-a\-rith\-mic-space algorithm 
that computes the modular decomposition tree of a graph, and presented a 
variation of the Modular Decomposition Theorem for polynomial time.
In the context of algorithmic graph theory, 
where modular decomposition has been established as a fundamental tool, 
these should find various applications.
As a first application, we directly obtained that 
cograph recognition and cograph canonization is computable in logarithmic space.

\section*{Acknowledgments}
The author would like to thank Frank Fuhlbr\"uck, Martin Grohe, Nicole Schweikardt, Oleg Verbitsky
and the reviewers for helpful comments that contributed to improving the paper.

\bibliographystyle{alpha}
\bibliography{lics-lmcs.bib}
\end{document}